\documentclass{article}
\usepackage{amssymb}
\usepackage{amsmath}
\usepackage{amsthm}
\usepackage{amsfonts}
\usepackage{mathrsfs}
\usepackage{verbatim}
\usepackage{amsmath,amsfonts,amssymb,bm,cite,graphicx,algorithm,algorithmic,amsthm}

\newtheorem{thm}{Theorem}
\newtheorem{theorem}{Theorem}
\newtheorem{lem}[thm]{Lemma}

\newtheorem{defn}{Definition}[section]

\newenvironment{claim}[1]{\par\noindent\underline{Claim:}\space#1}{}

\theoremstyle{remark}

\usepackage{graphics}
\usepackage{graphicx}
\usepackage{geometry}
\usepackage{algorithm}
\usepackage{algorithmic}
\usepackage[english]{babel}
\usepackage{caption}
\usepackage{subcaption}

\usepackage{cite}
\def\x{{\mathbf x}}

\begin{document}

\title{Separation-Free Super-Resolution from Compressed Measurements is Possible: an Orthonormal Atomic Norm Minimization Approach}

\author{
\thanks{The first two authors contributed equally to this work.}
Weiyu Xu\thanks{Department of Electrical and Computer Engineering, University of Iowa, Iowa City, IA 52242. Email: \texttt{weiyu-xu@uiowa.edu}. }
\and Jirong Yi\thanks{Department of Electrical and Computer Engineering, University of Iowa, Iowa City, IA 52242. Yi is co-first author.}
\and~Soura Dasgupta\thanks{Department of Electrical and Computer Engineering, University of Iowa, Iowa City, IA 52242.}
\and
~Jian-Feng Cai\thanks{Department of Mathematics, Hong Kong University of Science and Technology, Hong Kong.}
\and
~Mathews Jacob \thanks{Department of Electrical and Computer Engineering, University of Iowa, Iowa City, IA.}
\and
~\and~Myung Cho \thanks{Department of Electrical and Computer Engineering, University of Iowa, Iowa City, IA.}
%\and
%~?\thanks{}
}

\maketitle

\begin{abstract}
We consider the problem of recovering the superposition of $R$ distinct complex exponential functions from compressed non-uniform time-domain samples.
Total Variation (TV) minimization or atomic norm minimization was proposed in the literature to recover the $R$ frequencies or the missing data. However,
it is known that in order for TV minimization and atomic norm minimization to recover the missing data or the frequencies, the underlying $R$ frequencies are required to be well-separated,
even when the measurements are noiseless.  This paper shows that the Hankel matrix recovery approach can super-resolve
the $R$ complex exponentials and their frequencies from compressed non-uniform measurements, regardless of how close their frequencies are to each other.
We propose a new concept of orthonormal atomic norm minimization (OANM), and demonstrate that the success of Hankel matrix recovery in separation-free super-resolution comes from the fact that
the nuclear norm of a Hankel matrix is an orthonormal atomic norm. More specifically, we show that, in traditional atomic norm minimization,  the underlying parameter values \emph{must} be well separated
to achieve successful signal recovery, if the atoms are changing continuously with respect to the continuously-valued parameter. In contrast, for OANM,
it is possible the OANM is successful even though the original atoms can be arbitrarily close.

As a byproduct of this research, we provide one matrix-theoretic inequality of nuclear norm, and give its proof from the theory of compressed sensing.
\end{abstract}

\section{Introduction}
In super-resolution, we are interested in recovering the high-end spectral information of signals from observations of it low-end spectral
components \cite{candes_towards_2014}. In one setting of super-resolution problems, one aims to recover a superposition of complex
exponential functions from time-domain samples.  In fact, many problems arising in science and engineering involve high-dimensional signals
that can be modeled or approximated by a superposition of a few complex exponential functions.
In particular, if we choose the exponential functions to be complex sinusoids, this superposition of complex exponentials models signals in
acceleration of medical imaging \cite{lustig_sparse_2007}, analog-to-digital conversion \cite{tropp_beyond_2010}, and signals in array signal processing \cite{roy_esprit-estimation_1989}.
Accelerated NMR  (Nuclear magnetic resonance) spectroscopy, which is a prerequisite for studying short-lived molecular systems and monitoring chemical reactions in real time, is another
application where signals can be modeled or approximated by a superposition of complex exponential functions.
How to recover the superposition of complex exponential functions or parameters of these complex exponential functions is of prominent importance in these applications.

In this paper, we consider how to recover those superposition of complex exponential from linear measurements. More specifically, let ${\bm{x}}\in\mathbb{C}^{2N-1}$ be a vector satisfying
\begin{equation}\label{eq:hatx}
\bm{x}_j=\sum_{k=1}^{R}c_k z_k^j, \qquad j=0,1,\ldots,2N-2,
\end{equation}
where $z_k\in\mathbb{C}$, $k=1,\ldots,R$, are some unknown complex numbers with $R$ being a positive integer. In other words, ${\bm{x}}$ is a superposition of $R$ complex exponential functions.
We assume $R\leq 2N-1$.  When $|z_k|=1$, $k=1,\ldots,R$, ${\bm{x}}$ is a superposition of complex sinusoid. When $z_k=e^{-\tau_k}e^{2\pi\imath f_k}$, $k=1,\ldots,R$, $\imath=\sqrt{-1}$, ${\bm{x}}$ can model the signal in NMR spectroscopy.

Since $R\leq 2N-1$ and often $R\ll 2N-1$, the degree of freedom to determine ${\bm{x}}$ is much less than the ambient dimension $2N-1$. Therefore, it is possible to recover ${\bm{x}}$ from its under sampling \cite{candes_robust_2006,donoho_compressed_2006}.
In particular, we consider recovering ${\bm{x}}$ from its linear measurements
\begin{equation}\label{eq:linmea}
\bm{b}=\mathcal{A}({\bm{x}}),
\end{equation}
where $\mathcal{A}$ is a linear mapping to $\mathbb{C}^{M}$, $M< 2N-1$. After ${\bm{x}}$ is recovered, we can use the single-snapshot MUSIC  or the Prony's method to recover the parameter $z_{k}$'s.

The problem on recovering  ${\bm{x}}$ from its linear measurements (\ref{eq:linmea}) can be solved
using Compressed Sensing (CS)\cite{candes_robust_2006}, by discretizing the dictionary of basis vectors into grid points corresponding to discrete values of $z_k$. When the parameters $f_k$'s in signals from spectral compressed sensing or  $(f_k,\tau_k)$'s from  signals in accelerated NMR spectroscopy indeed fall on the grid, CS is a powerful tool to recover those signals even when the number of samples is far below its ambient dimension ($R\ll 2N-1$) \cite{candes_robust_2006,donoho_compressed_2006}. Nevertheless, the  parameters in our problem setting often take continuous values, leading to a continuous dictionary, and may not exactly fall on a grid. The basis mismatch problem between the continuously-valued parameters and the grid-valued parameters degenerates the performance of conventional compressed sensing \cite{chi_sensitivity_2011}. %Furthermore, there is no uniform way for CS to recover general superposition of complex exponential functions with parameters $z_{k}$ not on the unit circle, such as in accelerated NMR spectroscopy.

In two seminal papers \cite{candes_towards_2014,tang_compressed_2013}, the authors proposed to use the total variation minimization or the atomic
 norm minimization to recover $\bm{x}$ or to recover the parameter $z_{k}$, when $z_{k}=e^{\imath 2 \pi f_k}$ with $f_{k}$ taking continuous values from $[0,1)$. In
 these two papers, the author showed that the TV minimization or the atomic norm minimization can recover correctly the
 continuously-valued frequency $f_k$'s when there are no observation noises. However, as shown in \cite{candes_towards_2014,tang_compressed_2013,resolutionlimit}, in
 order for the TV minimization or the atomic norm minimization to recover spectrally sparse data or the associated frequencies correctly, it is necessary to require that adjacent frequencies be separated far enough from each other. For example, for complex exponentials with $z_k$'s taking values on the complex unit circle, it is required that their adjacent frequencies $f_k \in [0,1]$'s be at least $\frac{2}{2N-1}$ apart. This separation condition is necessary, even if we observe the full $(2N-1)$ data samples, and even if the observations are noiseless.

This raises a natural question, ``Can we super-resolve the superposition of complex exponentials with \emph{continuously-valued} parameter $z_k$, without requiring frequency separations, from compressed measurements?''  In this paper, we answer this question in positive. More specifically, we show that a Hankel matrix recovery approach using nuclear norm minimization can super-resolve the superposition of complex exponentials with \emph{continuously-valued} parameter $z_k$, without requiring frequency separations, from compressed measurements. This separation-free super-resolution result holds even when we only compressively observe $\bm{x}$ over a subset $\mathcal{M} \subseteq \{0,..., 2N-2\}$.

In this paper, we give the worst-case and average-case performance guarantees of Hankel matrix recovery in recovering the superposition of complex exponentials.
In establishing the worst-case performance guarantees, we establish the conditions under which the Hankel matrix recovery can recover the underlying complex exponentials,
no matter what values the coefficients $c_{k}$'s of the complex exponentials take. For the average-case performance guarantee, we assume that
the phases of the coefficients $c_{k}$'s  are uniformly  distributed over $[0, 2\pi)$.  For both the worst-case and average-case performance guarantees,
we establish that Hankel matrix recovery can super-resolve complex exponentials with continuously-valued parameters $z_{k}$'s,
no matter how close two adjacent frequencies or parameters $z_{k}$'s are to each other. We further introduce a new concept of orthonormal atomic norm minimization (OANM),
and discover that the success of Hankel matrix recovery in separation-free super-resolution comes from the fact that the nuclear norm of the Hankel matrix is an orthonormal atomic norm.
In particular, we show that,  in traditional atomic norm minimization,  for successful signal recovery, the underlying parameters \emph{must} be well separated, if the atoms are
changing continuously with respect to the continuously-valued parameters; however, it is possible the OANM is successful even though the original atoms can be arbitrarily close.

As a byproduct of this research, we discover one interesting matrix-theoretic inequality
of nuclear norm, and give its proof from the theory of compressed sensing.

\subsection{Comparisons with related works on Hankel matrix recovery and atomic norm minimization}
%\cite{fazel_log-det_2003}
Low-rank Hankel matrix recovery approaches were used for recovering parsimonious models in system identifications, control, and signal processing.
In \cite{markovsky_structured_2008}, Markovsky considered low-rank approximations for Hankel structured matrices with applications in signal processing, system identifications and control.
In \cite{fazel_log-det_2003,fazel_hankel_2013}, Fazel et al. introduced low-rank Hankel matrix recovery via nuclear norm minimization, motivated by applications including realizations and identification of linear time-invariant systems,
inferring shapes (points on the complex plane) from moments estimation (which is related to super-resolution with $z_k$ from the complex plane),
and moment matrix rank minimization for polynomial optimization.  In \cite{fazel_hankel_2013}, Fazel et al. further designed optimization algorithms to solve the nuclear norm minimization problem for low-rank
Hankel matrix recovery.  In \cite{chen_robust_2014}, Chen and Chi proposed to use multi-fold Hankel matrix completion for spectral compressed sensing, studied the performance guarantees of spectral compressed sensing via structured multi-fold Hankel matrix completion, and derived performance guarantees of structured matrix completion.  However, the results in \cite{chen_robust_2014} require that the Dirichlet kernel associated with underlying frequencies
satisfies certain incoherence conditions, and these conditions require the underlying frequencies to be well separated from each other. In \cite{dai_nuclear_2015,HankelComon}, the authors
derived performance guarantees for Hankel matrix completion in system identifications. However, the performance guarantees in \cite{dai_nuclear_2015,HankelComon} require a very specific sampling patterns of
fully sampling the upper-triangular part of the Hankel matrix. Moreover, the performance guarantees in \cite{dai_nuclear_2015,HankelComon} require that the parameters $z_{k}$'s be very small
(or smaller than $1$) in magnitude. In our earlier work \cite{cai_robust_2016}, we established performance guarantees of Hankel matrix recovery for spectral compressed sensing under Gaussian measurements of
${\bm{x}}$, and by comparison, this paper considers direct observations of  $\bm{x}$ over a set $\mathcal{M} \subseteq \{0,1,2,...,2N-2\}$, which is a more relevant sampling model in many applications.
In the single-snapshot MUSIC algorithm \cite{liao_music_2016}, the Prony's method \cite{de_prony_essai_1795} or the matrix pencil approach \cite{hua_matrix_1990}, one would need the full $(2N-1)$ consecutive samples to perform frequency identifications, while
the Hankel matrix recovery approach can work with compressed measurements. When prior information of the locations of the frequencies are available, one can use weighted atomic norm minimization to relax the
separation conditions in successful signal recovery \cite{Prior}.  In \cite{schiebinger_superresolution_2015}, the authors consider super-resolution without separation using atomic norm minimization, but under the restriction that the
coefficients are non-negative and for a particular set of atoms.

\subsection{Organizations of this paper}
The rest of the paper is organized as follows. In Section \ref{sec:preliminaries}, we present the problem model, and introduce the Hankel matrix recovery approach.
In Section \ref{sec:RecoveryGuarantees}, we investigate the worst-case performance guarantees of recovering spectrally sparse signals regardless of frequency separation, using the Hankel matrix recovery approach.
In Section \ref{sec:average}, we study the Hankel matrix recovery's average-case performance guarantees of recovering spectrally sparse signals regardless of frequency separation.
In Section \ref{sec:separation}, we show that  in traditional atomic norm minimization,  for successful signal recovery, the underlying parameters \emph{must} be well separated, if the atoms are
changing continuously with respect to the continuously-valued parameters. %In Section \ref{sec:OANM}, we introduce the concept of orthonormal atomic norm minimization, and explain the
%separation-free super-resolution result from the perspective of orthonormal atomic norm minimization.
In Section \ref{sec:OANM}, we introduce the concept of orthonormal atomic norm minimization, and show that it is possible that the atomic
norm minimization is successful even though the original atoms can be arbitrarily close. In Section \ref{sec:nuclearnorminequality}, as a byproduct of this research, we provide one matrix-theoretic
inequality of nuclear norm, and give its proof from the theory of compressed sensing. Numerical results are given in Section \ref{sec:NumericalResult} to validate our theoretical predictions.
We conclude our paper in Section \ref{sec:OpenProblem}.

\subsection{Notations}
We denote the set of complex numbers and real number as $\mathbb{C}$ and $\mathbb{R}$ respectively. We use calligraphic uppercase letters to represent index sets, and use $|\cdot|$ to represent a set's cardinality. When we use an index set as the subscript of a vector, we refer to the part of the vector over the index set. For example, $\bm{x}_{\Omega}$ is the part of vector $\bm{x}$
over the index set $\Omega$.  We use $\mathbb{C}^{n_1\times n_2}_{{r}}$ to represent the set of matrices from $\mathbb{C}^{n_1\times n_2}$ with rank $r$.  We denote the trace of a matrix $\bm{X}$ by ${\rm Tr}(\bm{X})$, and denote the real part and imaginary parts of a matrix $\bm{X}$  by ${\rm Re}(X)$ and ${\rm Im}(X)$ respectively. The superscripts $T$ and $*$ are used to represent transpose, and conjugate transpose of matrices or vectors. The Frobenius norm, nuclear norm, and spectral norm of a matrix are denoted by $\|\cdot\|_F, \|\cdot\|_*$ and $\|\cdot\|_2$ (or $\|\cdot\|$) respectively. The notation $\|\cdot\|$ represents the spectral norm
if its argument is a matrix, and represents the Euclidean norm if its argument is vector. The probability of an event ${S}$ is denoted by $\mathbb{P}({S})$.

\section{Problem statement}
\label{sec:preliminaries}
\quad\ The underlying model for spectrally sparse signal is a mixture of complex exponentials
\begin{align}\label{Def:SignalModel}
{\bm{x}}_j=\sum_{k=1}^{R} c_{k}e^{(\imath 2\pi f_{k} -\tau_k) j}, j\in \{0,1,\cdots,2N-2\}
\end{align}
where $\imath=\sqrt{-1}, f_k\in[0,1)$, $c_k\in\mathbb{C}$, and $\tau_k\geq0$ are the normalized frequency, coefficients, and damping factor, respectively.  We observe $\bm{x}$ over a subset $\mathcal{M} \subseteq \{ 0,1, 2, ..., 2N-2\}$.

To estimate the continuous parameter $f_k$'s, in \cite{candes_towards_2014,tang_compressed_2013}, the authors proposed to use the total variation minimization or the atomic norm minimization
to recover $\bm{x}$ or to recover the parameter $z_{k}$, when $z_{k}=e^{\imath 2 \pi f_k}$ with $f_{k}$ taking continuous values from $[0,1)$. In these two papers,
the author showed that the TV minimization or the atomic norm minimization can recover correctly the continuously-valued frequency $f_k$'s when there are no
observation noises. However, as shown in \cite{candes_towards_2014,tang_compressed_2013,resolutionlimit}, in order for atomic norm minimization to recover spectrally sparse data or the associated frequencies
correctly, it is necessary to require that adjacent frequencies be separated far enough from each other. Let us define the minimum separation between frequencies as the following:
\begin{defn}
(minimum separation, see \cite{candes_towards_2014}) For a frequency subset $\mathcal{F}\subset [0,1)$ with a group of points, the minimum separation is defined as smallest distance two arbitrary different elements in $\mathcal{F}$, i.e.
\begin{align}
{\rm dist}(\mathcal{F})= \inf_{f_i,f_l\in \mathcal{F}, f_i\neq f_l} d(f_i, f_l),
\end{align}
where $d(f_i, f_l)$ is the wrap around distance between two frequencies.
\end{defn}

As shown in \cite{resolutionlimit}, for complex exponentials with $z_k$'s taking values on the complex unit circle,
it is required that their adjacent frequencies $f_k \in [0,1]$'s be at least $\frac{2}{2N-1}$ apart. This separation condition is necessary,
even if we observe the full $(2N-1)$ data samples, and even if the observations are noiseless.

Following the idea the matrix pencil method in \cite{hua_matrix_1990} and Enhanced Matrix Completion (EMaC) in \cite{chen_robust_2014}, we construct a Hankel matrix based on signal ${\bm{x}}$.
More specifically, define the Hankel matrix ${\bm{H}(\bm{x})}\in\mathbb{C}^{N\times N}$ by
\begin{equation}\label{eq:Hankel}
{\bm{H}}_{jk}(\bm{x})=\bm{x}_{j+k-2},\qquad j,k=1,2,\ldots,N.
\end{equation}

The expression \eqref{Def:SignalModel} leads to a rank-$R$ decomposition:
$$
{\bm{H} (\bm{x})}=
\left[
\begin{matrix}
1&\ldots&1\cr
z_1&\ldots&z_R\cr
%e^{2\pi\imath 2f_1}&\ldots&e^{2\pi\imath 2f_r}\cr
\vdots&\vdots&\vdots\cr
z_1^{N-1}&\ldots&z_R^{N-1}\cr
\end{matrix}
\right]
\left[\begin{matrix}
c_1\cr &\ddots\cr&&c_R
\end{matrix}
\right]
\left[
\begin{matrix}
1&z_1\ldots&z_1^{N-1}\cr
\vdots&\vdots&\vdots\cr
1&z_R\ldots&z_R^{N-1}\cr
\end{matrix}
\right]
$$

Instead of reconstructing ${\bm{x}}$ directly, we reconstruct the rank-$R$ Hankel matrix ${\bm{H}}$, subject to the observation constraints.
Low rank matrix recovery has been widely studied in recovering a matrix from incomplete observations \cite{candes_exact_2009}.
It is well known that minimizing the nuclear norm can lead to a solution of low-rank matrices. We therefore use the nuclear norm minimization to recover the low-rank matrix $\bm{H}$. More specifically,
for any given $\bm{x}\in{\mathbb{C}^{2N-1}}$, let $\bm{H}(\bm{x})\in\mathbb{C}^{N\times N}$ be the corresponding Hankel matrix. We solve the following optimization problem:
\begin{equation}\label{eq:min}
\min_{\bm{x}}\|\bm{H}(\bm{x})\|_*,\qquad\mbox{subject to}\quad \mathcal{A}(\bm{x})=\bm{b},
\end{equation}
where $\|\cdot\|_*$ is the nuclear norm, and $\mathcal{A}$ and $\bm{b}$ are the linear measurements and measurement results.
When there is noise $\bm{\eta}$ contained in the observation, i.e.,
$$
\bm{b}=\mathcal{A}{\bm{x}}+\bm{\eta},
$$
we solve
\begin{equation}\label{eq:minnoise}
\min_{\bm{x}}\|\bm{H}(\bm{x})\|_*,\qquad\mbox{subject to}\quad \|\mathcal{A}\bm{x}-\bm{b}\|_2\leq\delta,
\end{equation}
where $\delta=\|\bm{\eta}\|_2$ is the noise level.

It is known that the nuclear norm minimization (\ref{eq:min}) can be transformed into a semidefinite program
\begin{align}\label{SDP}
& \min_{\bm{\x},\bm{Q}_1,\bm{Q}_2} \frac{1}{2}({\rm Tr}(\bm{Q}_1)+{\rm Tr}(\bm{Q}_2)) \nonumber\\
& {\rm s.t.\ } \bm{b}=\mathcal{A} \bm{x},\nonumber\\
&\ \ \ \
\begin{array}{l}
\left[\begin{array}{*{20}{c}}
\bm{Q}_1 & \bm{H}(\bm{x})^*\\
\bm{H}(\bm{x}) & \bm{Q}_2
\end{array}\right]
\end{array}
\succeq 0
\end{align}
which can be easily solved with existing convex program solvers such as interior point algorithms. After successfully recovering all the time samples, we can
use the single-snapshot MUSIC algorithm (as discussed in \cite{liao_music_2016}) to identify the underlying frequencies $f_k$.

For the recovered Hankel matrix $\bm{H}(\bm{x})$,  let its SVD be
\begin{align}\label{Def:SVDofZ}
\bm{H}(\bm{x})=[\bm{U}_1\ \bm{U}_2]
        \begin{array}{l}
        \left[\begin{array}{*{20}{c}}
        \bm{\Sigma}_1 &0\\
        0 &0
        \end{array}\right]
        \end{array}[\bm{V}_1\ \bm{V}_2]^*, \bm{U}_1, \bm{V}_1\in\mathbb{C}^{N\times R},
\end{align}
and we define the vector $\bm{\phi}^N(f)$ and imaging function $J(f)$ as
\begin{align}\label{Def:AtomVector}
\bm{\phi}^N(f)=(e^{\imath {2 \pi f 0}}, e^{\imath {2 \pi f 1}},..., e^{\imath {2 \pi f (N-1)}} )^T,
J(f)=\frac{\|\bm{\phi}^N(f)\|_2}{\|\bm{U}_2^*\bm{\phi}^N(f)\|_2}, f\in[0,1).
\end{align}
The single-snapshot MUSIC algorithm is given in Algorithm \ref{alg:MUSIC}.
 \begin{algorithm}
    \caption{The Single-Snapshot MUSIC algorithm \cite{liao_music_2016}}
    \label{alg:MUSIC}
    \begin{algorithmic}[1]
        \STATE require: solution ${\bm {x}}$, parameter $R$ and $N$

        \STATE form Hankel matrix $Z\in \mathbb{C}^{N\times N}$

        \STATE SVD $\bm{Z}=[\bm{U}_1\ \bm{U}_2]\begin{array}{l}\left[\begin{array}{*{20}{c}}
        \bm{\Sigma}_1 &0\\
        0 &0
        \end{array}\right]\end{array}[\bm{V}_1\ \bm{V}_2]^*$ with $\bm{U}_1\in\mathbb{C}^{N\times R}$ and  $\bm{\Sigma}_1\in\mathbb{C}^{R\times R}$

        \STATE compute imaging function $J(f)=\frac{\|\bm{\phi}^N(f)\|_2}{\|\bm{U}_2^*\bm{\phi}^N(f)\|_2}, f\in[0,1)$

        \STATE get set $\hat{\mathcal{F}}$=\{ $f$: $f$ corresponds to $R$ largest local maxima of $J(f)$\}

    \end{algorithmic}
\end{algorithm}

In \cite{liao_music_2016}, the author showed that the MUSIC algorithm can exactly recover all the frequencies by finding the local maximal of $J(f)$. Namely,
for undamped signal (\ref{Def:SignalModel}) with the set of frequencies $\mathcal{F}$, if $N\geq R$, then ``$f \in \mathcal{F}$'' is equivalent to ``$J(f)=\infty$.''

\section{Worst-case performance guarantees of separation-free super-resolution}\label{sec:RecoveryGuarantees}
In this section, we provide the worst-case performance guarantees of Hankel matrix recovery for recovering the superposition of complex exponentials. Namely, we provide the conditions under which the Hankel matrix recovery can uniformly recover the superposition of every possible $R$ complex exponentials. Our results show that the Hankel matrix recovery can achieve separation-free super-resolution, even if we consider the criterion of worst-case performance guarantees. Later in Section \ref{sec:tightness}, we further show that our derived worst-performance guarantees are tight, namely we can find examples where $R$ is bigger than the predicted recoverable sparsity by our theory, and the nuclear norm minimization fails to recover the superposition of complex exponentials.

We let $w_i$ be the number of elements in the $i$-th anti-diagonal of matrix $\bm{H}$, namely,
\begin{align}\label{Def:AntiDiagonalElementsNumber}
w_{i}=
\begin{cases}
i, i=1,2,\cdots,N,\\
2N-i, i=N+1,\cdots,2N-1.
\end{cases}
\end{align}

Here we call the $(2N-1)$ anti-diagonals of $\bm{H}(\bm{x})$ from the left top to the right bottom as the $1$-st anti-diagonal, ..., and the $(2N-1)$-th anti-diagonal.
We also define $w_{min}$ as
\begin{align}\label{Def:MinAntiDiagonalElementsNumber}
w_{min}=\min_{i \in \{0,1,2,...,2N-2\}\setminus \mathcal{M}}  w_{i+1}.
\end{align}

With the setup above, we give the following Theorem \ref{thm:HankelMC} concerning the worst-case performance guarantee of Hankel matrix recovery.

\begin{theorem}\label{thm:HankelMC}
Let us consider the signal model of the superposition of $R$ complex exponentials (\ref{Def:SignalModel}), the observation set $\mathcal{M} \subseteq \{0,1,2,..., 2N-2\}$.
We further define $w_{i}$ of an $N \times N$ Hankel matrix $\bm{H}(\bm{x})$ as in (\ref{Def:AntiDiagonalElementsNumber}), and define
$w_{min}$ as in (\ref{Def:MinAntiDiagonalElementsNumber}). Then the nuclear norm minimization (\ref{eq:min}) will uniquely recover
$\bm{H}(\bm{\x})$, regardless of the (frequency) separation between the $R$ continuously-valued (frequencies) parameters if
\begin{align}\label{Bound:KinHankelMC}
R<\frac{w_{min}}{2(2N-1-|\mathcal{M}|)}.
\end{align}
\end{theorem}

\begin{proof}

We can change (\ref{eq:min}) to the following optimization problem:
\begin{equation}\label{eq:revisedmin}
\min_{\bm{B}, \bm{x}}\|\bm{B}\|_*,\qquad\mbox{subject to}\quad \bm{B}=\bm{H}(\bm{x}),~~\mathcal{A}(\bm{x})=\bm{b}.
\end{equation}
We can think of (\ref{eq:revisedmin}) as a nuclear norm minimization problem, where the null space of the linear operator (applied to ($\bm{B}$,$\bm{x}$)) in the constraints of (\ref{eq:revisedmin}) is
given by $(\bm{H}(\bm{z}), \bm{z})$ such that $\mathcal{A}(\bm{z})=0$.

From \cite{recht_null_2011,oymak_new_2010}, we have the following lemma about
the null space condition for successful signal recovery via nuclear norm minimization.
\begin{lem}\label{lem:StrongNullSpaceCondition} \cite{oymak_new_2010}
Let $\bm{X}_0$ be any $N \times N$ matrix of rank $R$, and we observe it through a linear mapping $\mathcal{A}(\bm{X}_0)=\bm{b}$. Then the nuclear norm minimization (\ref{eq:NNM})
\begin{equation}\label{eq:NNM}
\min_{\bm{X}}\|\bm{X}\|_*,\qquad\mbox{subject to}\quad \mathcal{A}(\bm{X})=\mathcal{A}(\bm{X}_0),
\end{equation}
can uniquely and correctly recover every matrix $\bm{X}_0$ with rank no more than $R$ if and only if, for all nonzero $ \bm{Z} \in \mathcal{N}(\mathcal{A})$,
\begin{align}\label{eq:NullSpaceUniqueness}
2\|\bm{Z}\|_{*R}<\|\bm{Z}\|_{*},
\end{align}
where $\|\bm{Z}\|_{*R}$ is the sum of the largest $R$ singular values of $\bm{Z}$, and $\mathcal{N}(\mathcal{A})$ is the null space of $\mathcal{A}$.
\end{lem}

Using this lemma, we can see that (\ref{eq:revisedmin}) or (\ref{eq:min}) can correctly recover $\bm{x}$ as a superposition of $R$ complex exponentials if
$\|\bm{H}(\bm{z})\|_{*R} < \|\bm{H}(\bm{z})\|_{*}$ holds true for every nonzero $\bm{z}$ from the null space of $\mathcal{A}$.

Considering the sampling set $\mathcal{M} \subseteq \{0,1,2,..., 2N-2\}$, the null space of the sampling operator $\mathcal{A}$ is
composed of $(2N-1)\times 1$ vectors $\bm{z}$'s such that $\bm{z}_{i+1}=0$ if $i \in \mathcal{M}$. For such a vector $\bm{z}$, let us denote the element across the $i$-th anti-diagonal
of $\bm{H}(\bm{z})$ as $a_{i}$, and thus $\bm{Q}=\bm{H}(\bm{z})$ is a Hankel matrix with its $(i+1)$ anti-diagonal element equal to $0$ if $i \in \mathcal{M}$. Let $\sigma_1,\cdots,\sigma_{N}$ be the $N$
singular values of $\bm{H}(\bm{z})$ arranged in a descending order. To verify the null space condition for nuclear norm minimization, we would like to find the largest $R$ such that
\begin{align}
\sigma_1+\cdots+\sigma_{R} < \sigma_{R+1}+\cdots+\sigma_{N},
\end{align}
for every nonzero $\bm{z}$ in the null space of $\mathcal{A}$.

Towards this goal, we first obtain a bound for its largest singular value (for a matrix $\bm{Q}$, we use $\bm{Q}_{i, :}$ to denote its $i$-th row vector):
\begin{align}
\sigma_1&= \max_{{\bm u} \in \mathbb{C}^N, \|{\bm u}\|_2=1}\|\bm{Q}{\bm u}\|_2\\
&= \max_{{\bm u} \in \mathbb{C}^N, \|{\bm u}\|_2=1}\sqrt{    \sum_{i=1}^{N} |\bm{Q}_{i, :} {\bm u}|^2       }\\
& = \max_{{\bm u} \in \mathbb{C}^N, \|{\bm u}\|_2=1}\sqrt{    \sum_{i=1}^{N}
 \left|\sum_{j\in\{j:\ i \in Ind(j),a_j \neq 0\}} a_j \bm{u}_{j-i+1}\right|^2
  }\label{Eq:RowElementCount}\\
&\leq \max_{{\bm u} \in \mathbb{C}^N, \|{\bm u}\|_2=1} \sqrt{    \sum_{i=1}^{N} \left(\sum_{j\in\{j: i \in Ind(j),a_j \neq 0\}}  |a_j|^2 \right ) \left(\sum_{j\in\{j: i \in Ind(j),a_j \neq 0\}}  |{\bm u}_{j-i+1}|^2 \right)       } \label{cauchy}\\
&= \max_{{\bm u} \in \mathbb{C}^N, \|{\bm u}\|_2=1} \sqrt{     \left(\sum_{j\in\{j: i \in Ind(j),a_j \neq 0\}}  |a_j|^2 \right )
\sum_{i=1}^{N}
\left(\sum_{j\in\{j: i \in Ind(j),a_j \neq 0\}}  |{\bm u}_{j-i+1}|^2 \right)       } \\
&\leq \max_{{\bm u} \in \mathbb{C}^N, \|{\bm u}\|_2=1} \sqrt{ \sum_{j_1=1}^{2N-1} \left[|a_{j_1}|^2 \left(\sum_{i \in Ind(j_{1})}~\sum_{j_2\in\{j: i \in Ind(j),a_j \neq 0\} } |{\bm u}_{j_2-i+1}|^2    \right)\right]} \label{Eq:AntidiagonalElementCount}\\
&\leq \max_{{\bm u} \in \mathbb{C}^N, \|{\bm u}\|_2=1} \sqrt{ \sum_{j_1=1}^{2N-1} \left[|a_{j_1}|^2 (2N-1-M)\right]} \label{adding}\\
&= \sqrt{ \sum_{j \in \{0, 1, ..., 2N-2\}\setminus \mathcal{M}} |a_{j+1}|^2 (2N-1-M)}.
\end{align}
where (\ref{Eq:RowElementCount}) is obtained by looking at the rows (with $Ind(j)$ being the set of indices of rows which intersect with the $j$-th anti-diagonal and $\{j: i \in Ind(j),a_j \neq 0\}$ being the set of all non-zero anti-diagonals intersecting with the $i$-th row),  (\ref{cauchy}) is due to the Cauchy-Schwarz inequality, and (\ref{adding}) is because $\|\bm{u}\|_2=1$ and, for each $i$,
$|\bm{u}_i|^2$ appears for no more than $(2N-1-M)$ times in $\left(\sum_{i \in Ind(j_{1})}~\sum_{j_2\in\{j: i \in Ind(j),a_j \neq 0\} } |{\bm u}_{j_2-i+1}|^2    \right)$.

Furthermore, summing up the energy of the matrix $\bm{Q}$, we have
\begin{align}
\sum_{i=1}^{N} \sigma_{i}^2= \sum_{i \in \{0, 1, ..., 2N-2\}\setminus \mathcal{M}} |a_{i+1}|^2 w_{i+1}.
\end{align}
Thus for any integer $k \leq N$, we have
\begin{align}
 \frac{\sum_{i=1}^{N} \sigma_{i}}{\sum_{i=1}^{k} \sigma_{i}}
  & \geq \frac{\sum_{i=1}^{N} \sigma_{i}}{k\sigma_1} \nonumber\\
  & \geq  \frac{\sum_{i=1}^{N} \sigma_{i}^2}{k\sigma_1^2} \nonumber\\
  & \geq \frac{\sum_{i\in \{0, 1, ..., 2N-2\}\setminus \mathcal{M} } |a_{i+1}|^2 w_{i+1}}{k\sum_{j\in \{0, 1, ..., 2N-2\}\setminus \mathcal{M}} |a_{j+1}|^2 (2N-1-M)} \nonumber\\
  &   \geq  \frac{\min_{i \in \{0, 1, ..., 2N-2\}\setminus \mathcal{M}} w_{i+1} }{k(2N-1-M)} \nonumber\\
  & =\frac{w_{min}}{k(2N-1-M)}.
\end{align}

So if
\begin{align}
\frac{\min_{i \in \{0, 1, ..., 2N-2\}\setminus \mathcal{M}} w_{i+1} }{R(2N-1-M)}>2,
\end{align}
then for ever nonzero vector $\bm{z}$ in the null space of $\mathcal{A}$, and the corresponding Hankel matrix $\bm{Q}=\bm{H}(\bm{z})$,
\begin{align}
\sum_{i=1}^{N} \sigma_{i}>2{\sum_{i=1}^{R} \sigma_{i}}.
\end{align}

It follows that, for any superposition of $R < \frac{w_{min}}{2(2N-1-M)} $ complex exponentials, we can correctly recover $\bm{x}$ over the whole set $\{0, 1, ..., 2N-2\}$ using the incomplete sampling set $\mathcal{M}$,
regardless of the separations between different frequencies (or between the continuously-valued parameters $z_{k}$'s for damped complex exponentials).
\end{proof}

On the one hand,  the performance guarantees given in Theorem \ref{thm:HankelMC} can be conservative: for average-case performance guarantees,even when the number of complex exponentials $R$ is bigger than predicted by Theorem \ref{thm:HankelMC}, the Hankel matrix recovery can still recover the missing data, even though the sinusoids can be very close to each other. On the other hand, the bounds on recoverable sparsity level $R$ given in Theorem \ref{thm:HankelMC} is tight for worst-case performance guarantees, as shown in the next section.
\subsection{Tightness of recoverable sparsity guaranteed by Theorem \ref{thm:HankelMC}}
We will show the tightness of recoverable sparsity guaranteed by Theorem \ref{thm:HankelMC} in Section \ref{sec:tightness}, which is built on the developments in Section \ref{sec:orthogonal}.

\section{Average-case performance guarantees}
\label{sec:average}
In this section, we study the performance guarantees for Hankel matrix recovery, when the phases of the coefficients of the $R$ sinusoids are iid and uniformly distributed over $[0, 2\pi)$.
For average-case performance guarantees, we show that we can recover the superposition of a larger number of complex exponentials than Theorem \ref{thm:HankelMC} offers.

\subsection{Average-case performance guarantees for orthogonal frequency atoms and the tightness of worst-case performance guarantees}
\label{sec:orthogonal}
\begin{theorem}
\label{thm:OrthogonalAtoms}
Let us consider the signal model of the superposition of $R$ complex exponentials (\ref{Def:SignalModel}) with $\tau_k=0$.  We assume that the
$R$ frequencies $f_1$, $f_2$,..., and $f_R$ are such that the atoms $(e^{\imath {2 \pi f_i 0}}, e^{\imath {2 \pi f_i 1}},..., e^{\imath {2 \pi f_i (N-1)}} )^T$, $1\leq i \leq R$, are orthogonal to each other. We let the observation set
be $\mathcal{M} = \{0,1,2,..., 2N-2\}\setminus \{N-1\}$. We assume the phases of coefficients $c_1,\cdots,c_R$ in signal model (\ref{Def:SignalModel}) are independent and uniformly distributed over $[0, 2\pi)$.
Then the nuclear norm minimization (\ref{eq:min}) will successfully and uniquely recover
$\bm{H}(\bm{\x})$ and ${\bm x}$, with probability approaching $1$ as $N \rightarrow \infty$ if
\begin{align}
R = N-c \sqrt{\log(N)N},
\end{align}
where $c>0$ is a constant.
\end{theorem}

\begin{proof}
We use the following Lemma \ref{lem:WeakNullSpaceCondition} about the condition for successful signal recovery through nuclear norm minimization. This lemma is an extension of Lemma 13 in \cite{oymak_new_2010}.
The key difference is Lemma \ref{lem:WeakNullSpaceCondition} deals with complex-numbered matrices. Moreover, Lemma \ref{lem:WeakNullSpaceCondition} gets rid of the ``iff'' claim for the null space condition in Lemma 13 of \cite{oymak_new_2010},
because we find that the condition in Lemma 13 of \cite{oymak_new_2010} is a sufficient condition for the success of nuclear norm minimization,
but not a necessary condition for the success of nuclear norm minimization. We give the proof of Lemma \ref{lem:WeakNullSpaceCondition}, and prove
the null space condition is only a sufficient condition for nuclear norm minimization in Appendix \ref{ProofofLem:WeakNullSpaceCondition} and Appendix \ref{OnlySufficient} respectively.

\begin{lem}\label{lem:WeakNullSpaceCondition}
Let $\bm{X}_0$ be any $M \times N$ matrix of rank $R$ in $\mathbb{C}^{M \times N}$, and we observe it through a linear mapping $\mathcal{A}(\bm{X}_0)=\bm{b}$. We also assume that $\bm{X}_0$ has a singular value decomposition (SVD)
 $\bm{X}_0= \bm{U} \bm{\Sigma} \bm{V}^*$, where $\bm{U} \in \mathbb{C}^{M \times R}$, $\bm{V}\in \mathbb{C}^{N \times R}$, and $\bm{\Sigma} \in \mathbb{C}^{R \times R}$ is a diagonal matrix.
Then the nuclear norm minimization (\ref{eq:newNNM2})
\begin{equation}\label{eq:newNNM2}
\min_{\bm{X}}\|\bm{X}\|_*,\qquad\mbox{subject to}\quad \mathcal{A}(\bm{X})=\mathcal{A}(\bm{X}_0),
\end{equation}
correctly and uniquely recovers $\bm{X}_0$ if, for every nonzero element $\bm{Q} \in \mathcal{N}(\mathcal{A})$,
\begin{align}
-|{\rm Tr}(\bm{U}^*\bm{Q}\bm{V})|+\|\bar{\bm{U}}^*\bm{Q}\bar{\bm{V}}\|_*> 0,
\end{align}
\noindent where $\bar{\bm{U}}$ and $\bar{\bm{V}}$ are such that $[\bm{U}\ \bar{\bm{U}}]$ and $[\bm{V}\ \bar{\bm{V}}]$ are unitary.
\end{lem}

We can change (\ref{eq:min}) to the following optimization problem
\begin{equation}\label{eq:newrevisedmin}
\min_{\bm{B}, \bm{x}}\|\bm{B}\|_*,\qquad\mbox{subject to}\quad \bm{B}=\bm{H}(\bm{x}),~~\mathcal{A}(\bm{x})=\bm{b}.
\end{equation}
We can think of (\ref{eq:newrevisedmin}) as a nuclear norm minimization, where he null space of the linear operator (applied to $\bm{B}$ and $\bm{x}$) in the constraints of (\ref{eq:revisedmin}) is
given by $(\bm{H}(\bm{z}), \bm{z})$ such that $\mathcal{A}(\bm{z})=0$.

Then we can see that (\ref{eq:newrevisedmin}) or (\ref{eq:min}) can correctly recover $\bm{x}$ as a superposition of $R$ complex exponentials if
\begin{align}\label{eq:WeakNullSpaceCondition}
|{ {\rm Tr}(\bm{V}\bm{U}^* \bm{H}(\bm{z}))}|< \|\bar{\bm{U}}^* \bm{H}(\bm{z}) \bar{\bm{V}} \|_{*},
\end{align}
hold true for any $\bm{z}$ which is a nonzero vector from the null space of $\mathcal{A}$, where $\bm{H}(\bm{x})=\bm{U}\bm{{\Sigma}}\bm{V}^*$ is the singular value decomposition (SVD) of $\bm{H}(\bm{x})$
with $\bm{U} \in \mathbb{C}^{N \times R}$ and $\bm{V} \in \mathbb{C}^{N \times R}$, and $\bar{\bm{U}}$ and $\bar{\bm{V}}$ are such that $[\bm{U}, \bar{\bm{U}}]$ and $[\bm{V}, \bar{\bm{V}}]$ are unitary.

Without loss of generality, let $f_k=\frac{s_k}{N}$ for $1\leq k \leq R$, where $s_k$'s are distinct integers between $0$ and $N-1$. Then
\begin{align}\label{Def:SVD-I}
\bm{U}
= \frac{1}{\sqrt{N}}
\begin{array}{l}
\left[\begin{array}{*{20}{c}}
e^{\imath {2 \pi \frac{s_1}{N} 0}} & e^{\imath {2 \pi \frac{s_2}{N} 0}} &\cdots &e^{\imath {2 \pi \frac{s_R}{N} 0}}  \\
e^{\imath {2 \pi \frac{s_1}{N} 1}} & e^{\imath {2 \pi \frac{s_2}{N} 1}} &\cdots &e^{\imath {2 \pi \frac{s_R}{N} 1}}  \\
\vdots &\vdots &\ddots &\vdots \\
e^{\imath {2 \pi \frac{s_1}{N} (N-1)}} & e^{\imath {2 \pi \frac{s_2}{N} (N-1)}} & \cdots  &e^{\imath {2 \pi \frac{s_R}{N} (N-1)}}
\end{array}\right]
\left[\begin{array}{*{20}{c}}
e^{-\imath\theta_1} &0 &\cdots &0  \\
0 & e^{-\imath \theta_2} &\cdots &0  \\
\vdots  &\vdots & \ddots &\vdots \\
0 & 0 &\cdots&e^{-\imath\theta_R}
\end{array}\right]
\end{array},
\end{align}
and
\begin{align}\label{Def:SVD-II}
\bm{V}=\frac{1}{\sqrt{N}}
\begin{array}{l}
\left[\begin{array}{*{20}{c}}
e^{-\imath {2 \pi \frac{s_1}{N} 0}} & e^{-\imath {2 \pi \frac{s_2}{N} 0}}  &\cdots &e^{-\imath {2 \pi \frac{s_R}{N} 0}}  \\
e^{-\imath {2 \pi \frac{s_1}{N} 1}} & e^{-\imath {2 \pi \frac{s_2}{N} 1}}  &\cdots &e^{-\imath {2 \pi \frac{s_R}{N} 1}}  \\
\vdots &\vdots & \ddots &\vdots \\
e^{\imath {-2 \pi \frac{s_1}{N} (N-1)}} & e^{-\imath {2 \pi \frac{s_2}{N} (N-1)}} & \cdots &e^{-\imath {2 \pi \frac{s_R}{N} (N-1)}}
\end{array}\right]
\end{array},
\end{align}
and
\begin{align}\label{Def:SVD-III}
\bm{\Sigma}=N
\begin{array}{l}
\left[\begin{array}{*{20}{c}}
|c_1| &0 &\cdots &0\\
0 & |c_2| &\cdots &0\\
\vdots &\vdots &\ddots &\vdots\\
0 &0 &\cdots &|c_R|
\end{array}\right]
\end{array},
\end{align}
where $\theta_k$'s ($1\leq k \leq R$) are iid random variables uniformly distributed over $[0, 2\pi)$.

When the observation set $\mathcal{M} = \{0,1,2,..., 2N-2\}\setminus \{N-1\}$, any $\bm{Q}=\bm{H}(\bm{z})$ with $\bm{z}$ from the null space of $\mathcal{A}$ takes the following form:
\begin{align}\label{Def:NullSpaceElement}
\bm{Q}=a
\begin{array}{l}
\left[\begin{array}{*{20}{c}}
0 &\cdots &1\\
\vdots &\reflectbox{$\ddots$} &\vdots\\
1 &\cdots &0
\end{array}\right],
\end{array}
a\in\mathbb{C}.
\end{align}
Thus
\begin{align}\label{RandomVariableSum}
   |{ {\rm Tr}(\bm{V}\bm{U}^* \bm{Q})}|
   & = \left|a{ {\rm Tr}\left(\bm{U}^*
   \begin{array}{l}
  \left[\begin{array}{*{20}{c}}
  0 &\cdots &1\\
  \vdots &\reflectbox{$\ddots$} &\vdots\\
  1 &\cdots &0
\end{array}\right]
\end{array}
\bm{V}\right)}\right| \nonumber\\
&=|a|\left(\frac{1}{\sqrt{N}}\right)^2\sum_{i=1}^{R} \sum_{t=0}^{N-1}e^{\imath\theta_i} e^{-\imath {2 \pi \frac{s_i (N-1-t)}{N}}} \times  e^{-\imath {2 \pi \frac{s_i (t)}{N} }} \\
   &=\frac{|a|}{N}\sum_{i=1}^{R} \sum_{t=0}^{N-1} e^{\imath\theta_i}e^{-\imath {2 \pi \frac{s_i (N-1)}{N}}}\\
   &={|a|}\sum_{i=1}^{R} e^{\imath\theta_i} e^{-\imath {2 \pi \frac{s_i (N-1)}{N}}}.
\end{align}

Notice that random variable $e^{\imath\theta_i} e^{-\imath {2 \pi \frac{s_i (N-1)}{N}}}$ are mutually independent random variables uniformly distributed over the complex unit circle. We will use the following lemma (see for example, \cite{tropp_introduction_2015}) to provide a concentration of measure result for the summation of $e^{\imath\theta_i} e^{-\imath {2 \pi \frac{s_i (N-1)}{N}}}$:
\begin{lem}\label{lem:BernsteinConcentration}
(see for example, \cite{tropp_introduction_2015}) For a sequence of i.i.d. random matrix matrices $\bm{M}_1,\cdots,\bm{M}_K$ with dimension $d_1\times d_2$ and their sum $\bm{M}=\sum_{i=1}^K \bm{M}_i$, if $\bm{M}_i$ satisfies
\begin{align}
\mathbb{E}[\bm{M}_i]=0, \|\bm{M}_i\|\leq L, \forall i=1,\cdots,K,
\end{align}
and $\bm{M}$ satisfies
\begin{align}
\nu(\bm{M})=\max\left\{\left|\left|\sum_{i=1}^K \mathbb{E}[\bm{M}_i \bm{M}_i^*]\right|\right|, \left|\left|\sum_{i=1}^K \mathbb{E}[\bm{M}_i^* \bm{M}_i]\right|\right|\right\},
\end{align}
then
\begin{align}
\mathbb{P}(\|\bm{M}\|\geq t) \leq (d_1+d_2)\cdot \exp\left(\frac{-t^2/2}{v(\bm{M})+Lt/3}\right), \forall t\geq 0.
\end{align}
\end{lem}

Applying Lemma \ref{lem:BernsteinConcentration} to the $1\times 2$ matrix composed of the real and imaginary parts of $e^{\imath\theta_i} e^{-\imath {2 \pi \frac{s_i (N-1)}{N}}}$, with $\nu=max(R,R/2)=R$, $L=1$, $d_1+d_2=3$, we have
\begin{align}\label{Concentration:UnperturbedCase}
\mathbb{P}\left(|{ {\rm Tr}(\bm{V}\bm{U}^* \bm{Q})}|\geq |a|t\right) \leq 3 e^{-\frac{t^2/2}{\nu+Lt/3}}=3 e^{-\frac{t^2/2}{R+t/3}}, \forall t>0.
\end{align}
We further notice that, $\bar{\bm{U}}$'s ( $\bar{\bm{V}}$' s) columns are normalized orthogonal frequency atoms (or their complex conjugates) with frequencies $\frac{l_i}{N}$,
with integer $l_i$'s different from the integers $s_k$'s of those $R$ complex exponentials. Thus
 \begin{align}\label{eq:NuclearNormBound}
\|\overline{\bm{U}}^* \bm{Q} \overline{\bm{V}} \|_*=|a|(N-R).
\end{align}
Pick $t=N-R$, and let $ \frac{(N-R)^2/2}{R+(N-R)/3}=c\log(N)$ with $c$ being a positive constant. Solving for $R$,
we obtain $R=N-\left (\frac{2}{3} \log(N)+ \sqrt{\frac{4}{9} c^2 \log^2(N)+2c \log(N)N } \right)$, which implies (\ref{eq:WeakNullSpaceCondition}) holds with probability
approaching 1 if $N\rightarrow\infty$. This proves our claims.
\end{proof}

\subsection{Tightness of recoverable sparsity guaranteed by Theorem \ref{thm:OrthogonalAtoms}}
\label{sec:tightness}
We show that the recoverable sparsity $R$ provided by Theorem \ref{thm:OrthogonalAtoms} is tight for $\mathcal{M}=\{0,1,2,..., 2N-2\} \setminus \{ N-1\}$.  For such a sampling set, $w_{{min}}=N$, and Theorem
\ref{thm:OrthogonalAtoms} provides a bound $R<\frac{w_{min}}{2}=\frac{N}{2}$.

In fact, we can show that if $R\geq \frac{N}{2}$, we can construct signal examples where the Hankel matrix recovery approach cannot recover the original signal $\bm{x}$.
 Consider the signal in Theorem \ref{thm:OrthogonalAtoms}. We choose the coefficients $c_i$'s such that  $e^{\imath\theta_i} e^{-\imath {2 \pi \frac{s_i (N-1)}{N}}}=1$, then
\begin{align}\label{Concentra1}
|{ {\rm Tr}(\bm{V}\bm{U}^* \bm{Q})}| ={|a|}\sum_{i=1}^{R} e^{\imath\theta_i} e^{-\imath {2 \pi \frac{s_i (N-1)}{N}}}=|a|R,
\end{align}
and we also have
 \begin{align}\label{eq:NuclearNormBound1}
\|\overline{\bm{U}}^* \bm{Q} \overline{\bm{V}} \|_*=|a|(N-R).
\end{align}

Thus $|{ {\rm Tr}(\bm{V}\bm{U}^* \bm{Q})}|\geq \|\overline{\bm{U}}^* \bm{Q} \overline{\bm{V}} \|_*$ for every $\bm{Q}=\bm{H}(\bm{z})$ with $\bm{z}$ in the null space of the sampling operator.
Thus the Hankel matrix recovery cannot recover the ground truth $\bm{x}$.  This shows that the prediction by Theorem \ref{thm:OrthogonalAtoms} is tight.

\subsection{Average-Case performance analysis with arbitrarily close frequency atoms}\label{sec:PerturbationAnalysis}

In this section, we further show that the Hankel matrix recovery can successfully recover the superposition of complex exponentials with arbitrarily close frequency atoms. In particular, we give average performance guarantees on recoverable sparsity $R$ when frequency atoms are arbitrarily close, and show that the Hankel matrix recovery can deal with much larger recoverable sparsity $R$ for average-case signals than predicted by Theorem \ref{thm:HankelMC}.

We consider a signal $\bm{x}$ composed of $R$ complex exponentials. Among them, $R-1$ orthogonal complex exponentials
(without loss of generality, we assume that these $R-1$ frequencies take values $\frac{l_i}{N}$, where $l_i$'s are integers). The other complex exponential $c_{cl}e^{\imath 2\pi f_{cl} j}$ has a frequency arbitrarily close to one of the $R-1$ frequencies, i.e.,
\begin{align}\label{Def:PerturbedSignal}
\bm{x}_j=\sum_{k=1}^{R-1} c_k e^{\imath 2\pi f_{k} j} + c_{cl} e^{\imath 2\pi f_{cl} j},~ j\in\{0,1,\cdots,2N-2\},
\end{align}
where $f_{cl}$ is arbitrarily close to one of the first $(R-1)$ frequencies. For this setup with arbitrarily close atoms, we have the following average-case performance guarantee.
\begin{theorem}\label{thm:PerturbationGuarantees}
Consider the signal model (\ref{Def:PerturbedSignal}) with $R$ complex exponentials, where the coefficients of these complex exponentials are iid uniformly distributed random over $[0,2\pi)$, and the first $(R-1)$ of the complex exponentials are such that the corresponding atom vectors in (\ref{Def:AtomVector}) are mutually orthogonal, while
the $R$-th exponential has frequency arbitrarily close to one of the first $(R-1)$ frequencies (in wrap-around distance). Let $c_{min}=\min \{|c_1|, |c_2|, ..., |c_{R-1}|\}$, and define
\begin{align}\label{Defn:dRelative2}
d_{rel}=\frac{|c_{cl}|}{c_{min}}+1.
\end{align}

If we have the observation set $\mathcal{M}=\{0,1,2,..., 2N-2\} \setminus \{ N-1\}$, and, for any constant $c>0$, if
\begin{align}\label{Def:PerturbedK}
R = \frac{(2N-8\sqrt{N}d_{rel}+2c\log(N))+\sqrt{12cN\log(N)-48c\sqrt{N}d_{rel}\log(N)+4c^2\log^2(N)}}{2},
\end{align}
then we can recover the true signal ${\bm{x}}$ via Hankel matrix recovery with high probability as $N \rightarrow \infty$, regardless of frequency separations.
\end{theorem}

\textbf{Remarks:} 1. We can extend this result to cases where neither of the two arbitrarily close frequencies has on-the-grid frequencies, and the Hankel matrix recovery method can recover a similar sparsity; 2.
Under a similar number of complex exponentials, with high probability, the Hankel matrix recovery can correctly recover the signal $\bm{x}$, uniformly over every possible phases of the two complex exponentials with arbitrarily close frequencies; 3. We can extend our results to other sampling sets, but for clarity of presentations, we choose  $\mathcal{M}=\{0,1,2,..., 2N-2\} \setminus \{ N-1\}$

\begin{proof}
To prove this theorem, we consider a perturbed signal $\bm{\tilde{x}}$. The original signal $\bm{x}$ and the perturbed $\bm{\tilde{x}}$ satisfy
\begin{align}\label{Def:PerturbaedSignal}
\bm{x}_j = \tilde{\bm{x}}_j+ c_{cl}e^{\imath2\pi f_{cl} j}- c_{rm} e^{\imath 2 \pi f_{rm} j}, j\in\{0,1,\cdots,2N-2\},
\end{align}
where $f_{rm}$ is a frequency such that its corresponding frequency atom is mutually orthogonal with the atoms corresponding to $f_1$, ..., and $f_{R-1}$. We further define
\begin{align}\label{Eq:MinMag}
d_{min}=\min\{|c_1|,\cdots,|c_{R-1}|, |c_{rm}|\}.
\end{align}

Let us define $\tilde{\bm{X}}=\bm{H}(\tilde{\bm{{x}}})$, ${\bm{X}}=\bm{H}(\bm{{x}})$, and the error matrix $\bm{E}$ such that:
\begin{align}
\bm{X}=\tilde{\bm{X}}+\bm{E},
\end{align}
where
\begin{align}
\bm{E}=
\begin{array}{l}
\left[\begin{array}{*{20}{c}}
c_{cl}e^{\imath2\pi f_{cl} \cdot 0}- c_{rm}e^{\imath 2 \pi f_{rm} \cdot 0} & \cdots & c_{cl}e^{\imath2\pi f_{cl} \cdot (N-1)}- c_{rm}e^{\imath 2 \pi f_{rm} \cdot (N-1)} \\
\vdots &\ddots &\vdots\\
c_{cl}e^{\imath2\pi f_{cl} \cdot (N-1)}- c_{rm}e^{\imath 2 \pi f_{rm} \cdot (N-1)} &\cdots & c_{cl}e^{\imath2\pi f_{cl} \cdot (2N-2)}- c_{rm}e^{\imath 2 \pi f_{rm} \cdot (2N-2)}
\end{array}\right]
\end{array}.
\end{align}
Both $\bm{X}$ and $\tilde{\bm{X}}$ are rank-$R$ matrices. For the error matrix $\bm{E}$, we have
\begin{align}\label{Eq:EBound}
\|\bm{E}\|_F
& = \left(
\sum_{i,k\in\{0,\cdots,N-1\}} |c_{cl}e^{\imath2\pi f_{cl} \cdot (i+k)}- c_{rm}e^{\imath 2 \pi f_{rm} \cdot (i+k)}|^2
\right)^{1/2} \nonumber\\
& \leq \left(
\sum_{i,k\in\{0,\cdots,N-1\}} \left(|c_{cl}|+|c_{rm}|\right)^2
\right)^{1/2}\nonumber\\
& = \left(N^2 \left(|c_{cl}|+|c_{rm}|\right)^2\right)^{1/2} \nonumber\\
& = N\left(|c_{cl}|+|c_{rm}|\right).
\end{align}

Following the derivations in Theorem \ref{thm:OrthogonalAtoms}, $\bm{\tilde{X}}$ has the following SVD
\begin{align}\label{SVD:UnperturbedMatrix}
\bm{\tilde{X}}=\bm{\tilde{\bm{U}}}\bm{\tilde{\bm{\Sigma}}} \bm{\tilde{\bm{V}}}^H,
\bm{\tilde{\bm{U}}}=[\bm{\tilde{\bm{U}}}_1\ \bm{\tilde{\bm{U}}}_2]\in\mathbb{C}^{N\times N},
\bm{\tilde{\bm{V}}}=[\bm{\tilde{\bm{V}}}_1\ \bm{\tilde{\bm{V}}}_2]\in\mathbb{C}^{N\times N},
\bm{\tilde{\bm{\Sigma}}}=
\begin{array}{l}
\left[\begin{array}{*{20}{c}}
\bm{\tilde{\bm{\Sigma}}}_1 &0\\
0 &0
\end{array}\right]
\end{array}\in\mathbb{C}^{N\times N}
\end{align}
where $\bm{\tilde{\bm{U}}}_1 \in \mathbb{C}^{N \times R}, \bm{\tilde{\bm{\Sigma}}}_1 \in \mathbb{C}^{R \times R}$ and $\bm{\tilde{\bm{V}}}_1\in \mathbb{C}^{N \times R}$ are defined as in
(\ref{Def:SVD-I}), (\ref{Def:SVD-II}) and (\ref{Def:SVD-III}). The polar decomposition of $\bm{\tilde{X}}$ using its SVD is given by
\begin{align}
\bm{\tilde{X}}=\bm{\tilde{P}}\bm{\tilde{H}},
\bm{\tilde{P}}=\bm{\tilde{\bm{U}}}_1\bm{\tilde{V}}_1^*,
\bm{\tilde{H}}=\bm{\tilde{\bm{V}}}_1\bm{\tilde{\bm{\Sigma}}}_1 \bm{\tilde{\bm{V}}}_1^*,
\end{align}
and the matrix $\tilde{\bm{P}}$ is called the unitary polar factor. Similarly, for $\bm{X}$, we have
\begin{align}\label{SVD:PerturbedMatrix}
{\bm{X}}={\bm{U}}{\bm{\Sigma}} {\bm{V}},
{\bm{U}}=[{\bm{U}}_1\ {\bm{U}}_2]\in\mathbb{C}^{N\times N},
{\bm{V}}=[{\bm{V}}_1\ {\bm{V}}_2]\in\mathbb{C}^{N\times N},
\bm{\bm{\Sigma}}=
\begin{array}{l}
\left[\begin{array}{*{20}{c}}
\bm{\bm{\Sigma}}_1 &0\\
0 &0
\end{array}\right]
\end{array},
\end{align}
and its polar decomposition for ${\bm{X}}$,
\begin{align}
{\bm{X}}={\bm{P}}{\bm{H}},
{\bm{P}}={\bm{U}}_1{\bm{V}}_1^*,
\bm{H}={\bm{V}}_1{\bm{\bm{\Sigma}}}_1 {\bm{V}}_1^*.
\end{align}

Suppose that we try to recover $\bm{x}$ through the following nuclear norm minimization:
\begin{align}\label{PerturbatedNNM}
& \min_{\bm{x}} \|{\bm{H}(\bm{x})}\|_* \nonumber\\
& {\rm s.t.\ }\mathcal{A}({\bm{x}})=\bm{b}.
\end{align}
As in the proof of Theorem \ref{thm:OrthogonalAtoms}, we analyze the null space condition for successful recovery using Hankel matrix recovery.
Towards this, we first bound the unitary polar factor of $\bm{X}$ through
the perturbation theory for polar decomposition.
\begin{lem}\label{Perturbation:PolarFactor}
(\cite{li_new_2003}, Theorem 3.4) For matrices $\bm{X}\in\mathbb{C}_r^{m\times n}$ and $\tilde{\bm{X}}\in\mathbb{C}_r^{m\times n}$ with SVD defined as (\ref{SVD:UnperturbedMatrix}) and (\ref{SVD:PerturbedMatrix}), let $\sigma_r$ and $\tilde{\sigma}_r$ be the smallest nonzero singular values of $\bm{X}$ and $\tilde{\bm{X}}$, respectively, then
\
\begin{align}
|||\bm{P}-\tilde{\bm{P}}|||\leq \left(\frac{2}{\sigma_r+\tilde{\sigma}_r}+\frac{2}{\max\{\sigma_r,\tilde{\sigma}_r\}}\right)|||\bm{X}-\tilde{\bm{X}}|||,
\end{align}
where $|||\cdot|||$ is any unitary invariant norm.
\end{lem}

For our problem, we have $m=n=N, r=R$. Let $\sigma_{R}$ be the $R$-th singular value of $\bm{X}$. From (\ref{Def:SVD-II}) and (\ref{Eq:MinMag}), an explicit form for $\sigma_R$ is
\begin{align}
\sigma_{R}= Nd_{min}.
\end{align}

Then  Lemma \ref{Perturbation:PolarFactor}  implies that
\begin{align}
\|\bm{P}-\tilde{\bm{P}}\|_F
\leq \frac{4}{\sigma_R}\|\bm{X}-\tilde{\bm{X}}\|_F.
\end{align}

From the null space condition for nuclear norm minimization (\ref{PerturbatedNNM}), we can correctly recover ${\bm{x}}$ if
\begin{align}
|{\rm Tr}({\bm{V}}_1{\bm{U}}_1^*\bm{Q})|
<
\|\bar{{\bm{U}}}_1^*\bm{Q}\bar{{\bm{V}}}_1\|_*,
\end{align}
for every nonzero $\bm{Q}=\bm{H}(\bm{z})$ with $ \bm{z}\in\mathcal{N}(\mathcal{A})$, where $\bar{{\bm{U}}}_1$ and $\bar{{\bm{V}}}_1$ are such that $[{\bm{U}}_1\ \bar{{\bm{U}}}_1]$ and $[{\bm{V}}_1\ \bar{{\bm{V}}}_1]$ are unitary,
i.e., $\bar{{\bm{U}}}_1={{\bm{U}}}_2$ and $\bar{{\bm{V}}}_1={{\bm{V}}}_2$.  Since the observation set $\mathcal{M}=\{0,1,2,..., 2N-2\} \setminus \{ N-1\}$, $\bm{Q}$ takes the form
in (\ref{Def:NullSpaceElement}) with $a\in \mathbb{C}$, and
\begin{align}
\|\bar{{\bm{U}}}_1^*\bm{Q}\bar{{\bm{V}}}_1\|_* = |a|(N-R).
\end{align}
Let us define $\Delta \bm{P}=\tilde{\bm{P}}-\bm{P}=\tilde{\bm{U}}_1\tilde{\bm{V}}^{*}_1-{\bm{U}}_1{\bm{V}}^{*}_1$. Then it follows from Lemma \ref{Perturbation:PolarFactor} and (\ref{Eq:EBound}) that
\begin{align}
|{\rm Tr}({\bm{V}}_1{\bm{U}}_1^*\bm{Q})|
& =|<\bm{Q},{\bm{U}}_1{\bm{V}}_1^*>| \nonumber\\
& \leq |<\bm{Q},\tilde{\bm{P}}>|+|<\bm{Q},\Delta \bm{P}>| \nonumber\\
& = |{\rm Tr}(\bm{\tilde{\bm{V}}}_1\bm{\tilde{\bm{U}}}_1^*\bm{Q})| + |<\bm{Q},\Delta \bm{P}>| \nonumber\\
& \leq |{\rm Tr}(\bm{\tilde{\bm{V}}}_1\bm{\tilde{\bm{U}}}_1^*\bm{Q})| + (\|\bm{Q}\|_F^2\|
\|\Delta \bm{P}\|_F^2)^{1/2} \nonumber\\
& \leq |{\rm Tr}(\bm{\tilde{\bm{V}}}_1\bm{\tilde{\bm{U}}}_1^*\bm{Q})| + |a|\sqrt{N}\|\Delta \bm{P}\|_F \nonumber\\
& \leq |{\rm Tr}(\bm{\tilde{\bm{V}}}_1\bm{\tilde{\bm{U}}}_1^*\bm{Q})| + |a|\frac{4 \sqrt{N} }{\tilde{\sigma}_R}\|\bm{E}\|_F\nonumber\\
& \leq |{\rm Tr}(\bm{\tilde{\bm{V}}}_1 \bm{\tilde{\bm{U}}}_1^*\bm{Q})| + |a|\frac{ 4\sqrt{N} } {N d_{min}} \cdot N\left(|c_{cl}|+|c_{rm}|\right) \nonumber\\
& \leq |{\rm Tr}(\bm{\tilde{\bm{V}}}_1\bm{\tilde{\bm{U}}}_1^*\bm{Q})| + 4|a|\sqrt{N}\cdot\frac{|c_{cl}|+|c_{rm}|}{d_{min}}.
\end{align}
Thus if
\begin{align}
|{\rm Tr}(\bm{\tilde{\bm{V}}}_1\bm{\tilde{\bm{U}}}_1^*\bm{Q})| + 4|a|\sqrt{N}\cdot\frac{|c_{cl}|+|c_{rm}|}{d_{min}}
< |a|(N-R),
\end{align}
namely,
\begin{align}
|{\rm Tr}(\bm{\tilde{\bm{V}}}_1\bm{\tilde{\bm{U}}}_1^*\bm{Q})|
< |a|\left(N-R-4\sqrt{N}\cdot d_{rel}\right),
\end{align}
where $d_{rel}$ is defined as $\frac{|c_{cl}|+|c_{rm}|}{d_{min}}$,
then solving (\ref{PerturbatedNNM}) will correctly recover ${\bm{x}}$.

The proof of Theorem \ref{thm:OrthogonalAtoms} leads to the concentration inequality (\ref{Concentration:UnperturbedCase}):
\begin{align*}
\mathbb{P}\left(|{ {\rm Tr}(\bm{\tilde{\bm{V}}}\bm{\tilde{\bm{U}}}^* \bm{Q})}|\geq |a|t\right) \leq 3 e^{-\frac{t^2/2}{R+t/3}}, \forall~~t>0.
\end{align*}

Taking $t=N-R-4\sqrt{N}d_{rel}$, we have
\begin{align}
\frac{t^2/2}{R+t/3}
& = \frac{(N-R-4\sqrt{N}d_{rel})^2/2}{R+(N-R-4\sqrt{N}d_{rel})/3} \nonumber\\
& = \frac{3}{2} \cdot
\frac{N^2+R^2+16Nd_{rel}^2-2NR-8N\sqrt{N}d_{rel}+8R\sqrt{N}d_{rel}}{3R+N-R-4\sqrt{N}d_{rel}}.
\end{align}
To have successful signal recovery with high probability, we let
\begin{align}\label{Bound:K}
\frac{3}{2} \cdot
\frac{N^2+R^2+16Nd_{rel}^2-2NR-8N\sqrt{N}d_{rel}+8R\sqrt{N}d_{rel}}{3R+N-R-4\sqrt{N}d_{rel}}=c\log(N),
\end{align}
where $c>0$ is a constant.
So we have
\begin{align}
R^2+sR+r=0,
\end{align}
where
\begin{align}
s= 8\sqrt{N}d_{rel}-2N-2c\log(N),
r=16Nd_{rel}^2-8\sqrt{N}Nd_{rel}-cN\log(N)+4c\sqrt{N}d_{rel}\log(N).
\end{align}
Solving this leads to
\begin{align}
R=\frac{(2N-8\sqrt{N}d_{rel}+2c\log(N))+\sqrt{12cN\log(N)-48c\sqrt{N}d_{rel}\log(N)+4c^2\log^2(N)}}{2}.
\end{align}

Since we can freely choose the coefficient $c_{rm}$, we choose $c_{rm}$ such that $|c_{rm}|=c_{min}=\min \{|c_1|, |c_2|, ..., |c_{R-1}|\}$. Under such a choice for $|c_{rm}|$,
$d_{min}=|c_{rm}|=c_{min}$, leading to $d_{rel}=\frac{|c_{cl}|+|c_{rm}|}{d_{min}}=\frac{|c_{cl}|+c_{min}}{c_{min}}$. This concludes the proof of Theorem \ref{thm:PerturbationGuarantees}.
\end{proof}

%\begin{comment}
%We take $t=N-R-4\sqrt{N}d_{rel}$, leading to
%\begin{align}
%\frac{t^2/2}{R+t/3}
%& = \frac{(N-R-4\sqrt{N}d_{rel})^2/2}{R+(N-R-4\sqrt{N}d_{rel})/3} \nonumber\\
%& = \frac{3}{2} \cdot
%\frac{N^2+R^2+16Nd_{rel}^2-2NR-8N\sqrt{N}d_{rel}+8R\sqrt{N}d_{rel}}{3R+N-R-4\sqrt{N}d_{rel}} \nonumber\\
%& \geq \frac{3}{2} \cdot
%\frac{N^2+16Nd_{rel}^2-2NR-8N\sqrt{N}d_{rel}+8R\sqrt{N}d_{rel}}{3R+N-R-4\sqrt{N}d_{rel}}.
%\end{align}
%Thus if for some positive constant $c>0$,
%\begin{align}\label{Bound:K}
%\frac{3}{2} \cdot
%\frac{N^2+16Nd_{rel}^2-2NR-8N\sqrt{N}d_{rel}+8R\sqrt{N}d_{rel}}{3R+N-R-4\sqrt{N}d_{rel}}=c\log(N),
%\end{align}
%then with probability at least $1-3N^{-c}$,  $|{ {\rm Tr}(\bm{\tilde{\bm{V}}}\bm{\tilde{U}}^* \bm{Q})}|<|a|t$ holds and (\ref{PerturbatedNNM}) will correctly recover the signal ${\bm{x}}$.
%Solving (\ref{Bound:K}), we obtain that as long as,
%\begin{align}\label{rawbound}
%R \leq \frac{N^2-8d_{rel}N\sqrt{N}-cN\log(N)+16d_{rel}^2 N+4cd_{rel}\sqrt{N}\log(N)}{2N-8d_{rel}\sqrt{N}+2c\log(N)},
%\end{align}
%(\ref{PerturbatedNNM}) will correctly recover the signal ${\bm{x}}$ with probability at least $1-3N^{-c}$.
%\end{comment}

\section{Separation is always necessary for the success of atomic norm minimization}
\label{sec:separation}
In the previous sections, we have shown that the Hankel matrix recovery can recover the superposition of complex exponentials, even though the frequencies of the complex exponentials can be arbitrarily close. In this
section, we show that, broadly, the atomic norm minimization must obey a non-trivial resolution limit. This is very different from the behavior of Hankel matrix recovery. Our results in this section also greatly generalize the
the necessity of resolution limit results in \cite{resolutionlimit}, to general continuously-parametered dictionary, beyond the dictionary of frequency atoms.  Moreover, our analysis is very different from the derivations in \cite{resolutionlimit}, which
used the Markov-Bernstein type inequalities for finite-degree polynomials.

\begin{theorem}
Let us consider a dictionary with its atoms parameterized by a continuous-valued parameter $\tau \in \mathbb{C}$. We also assume that each atom $\bm{a}(\tau)$ belongs to $\mathbb{C}^N$,
where $N$ is a positive integer. We assume that the set of all the atoms span a $Q$-dimensional subspace in $\mathbb{C}^N$.

Suppose the signal is the superposition of several atoms:
\begin{align}\label{Defn:atomdef}
\bm{x}=\sum_{k=1}^{R} c_{k} \bm{a}(\tau_{k}),
\end{align}
where the nonzero $c_{k} \in \mathbb{C}$ for each $k$, and $R$ is a positive integer representing the number of active atoms.

Consider any active atoms $\bm{a}(\tau_{k_1})$ and $\bm{a}(\tau_{k_2})$. With the other $(R-2)$ active atoms and their coefficients fixed (this includes the case $R=2$, namely there are only two active atoms
in total), if the atomic norm minimization can always identify the two active atoms, and correctly recover their coefficients for
$\bm{a}(\tau_{k_1})$ and $\bm{a}(\tau_{k_2})$, then the two atoms $\bm{a}(\tau_{k_1})$ and $\bm{a}(\tau_{k_2})$ must be well separated such that
\begin{align}
\label{def:sep}
\| \bm{a}(\tau_{k_1})-\bm{a}(\tau_{k_2})  \|_2 \geq 2 \max_{S\geq Q}\max_{\bm{A} \in \bm{M}^{S} }\frac{\sigma_{min}(\bm{A}) }{\sqrt{S}},
\end{align}
where $S$ is a positive integer, $\bm{M}^{S}$ is the set of matrices with $S$ columns and with each of these $S$ columns corresponding to an atom, and $\sigma_{min} (\cdot)$
is the smallest singular value of a matrix.
\end{theorem}

\begin{proof}

Define the sign of a coefficient $c_k$ as
\begin{align}\label{Defn:CoefficientSign}
{\rm sign}(c_k)=\frac{c_k}{|c_k|}.
\end{align}
Then according to \cite{candes_towards_2014} and \cite{resolutionlimit}, a necessary condition for the atomic norm to  identify the two active atoms, and correctly recover their coefficients is that, there exists
a dual vector $\bm{q} \in \mathbb{C}^N $ such that
%\begin{prop}\label{Thm:SufficientConditionUniqueSolution}
%([Tang and Bhaskar etc 2013]) Given a set of atoms $\mathcal{A}$, the whole set of time instances $\mathcal{T}$ and sampling index set $\Omega$, then the existence of a dual polynomial $Q(f)$, i.e,,
%\begin{align}
%Q(f)=<q,a(f,0)>=\sum_{t\in\mathcal{T}} q_{t} e^{-i2\pi ft},
%q_{t}=0, \forall t\notin \mathcal{T},
%\end{align}
%satisfying
\begin{align}
\label{reasoning}
\begin{cases}
\bm{a}(\tau_{k_1})^*\bm{q}={\rm sign}(c_{k_1}),~\bm{a}(\tau_{k_2})^*\bm{q}={\rm sign}(c_{k_2}) \\
|\bm{a (\tau)}^*\bm{q} |\leq 1, \forall \tau \notin \{\tau_{k_1}, \tau_{k_2}\}.
\end{cases}
\end{align}
%will guarantee the uniqueness of solution to problem (\ref{Defn:RecoveryModelAtomicNorm}).
%\end{prop}
We pick $c_{k_1}$ and $c_{k_2}$ such that $|{\rm sign}(c_k)-{\rm sign}(c_j)|=2$. Then

%Based on Proposition \ref{Thm:SufficientConditionUniqueSolution}, for arbitrary two atoms $a(f_k,0)$ and $a(f_j,0)$ with $f_k,f_j\in\mathcal{F}$, we get
\begin{align}
\| {\rm sign}(c_k)-{\rm sign}(c_j)\|
= \|\bm{a}(\tau_{k_1})^*\bm{q}- \bm{a}(\tau_{k_2})^*\bm{q}\|_2
\leq \|\bm{q}\|_2\|\bm{a}(\tau_{k_1})-\bm{a}(\tau_{k_2})\|_2,
\end{align}
Thus
\begin{align}\label{Eq:DualNorm1}
\|\bm{q}\|_2 \geq \frac{2}{\|\bm{a}(\tau_{k_1})-\bm{a}(\tau_{k_2})\|_2}.
\end{align}

Now we take a group of $S$ atoms, denoted by $\bm{a}_{select,1}$, ... , and $\bm{a}_{select,S}$, and use them to form the $S$ columns of a matrix $\bm{A}$.
%$\Omega_\kappa=\{a(f_{k_1},0),a(f_{k_\kappa},0)\}$ of size $|\Omega_\kappa|=\kappa$  where $f_{k_i}$ can be in or out of $\mathcal{F}$, and form a matrix, i.e.,
%\begin{align}
%A=[a(f_{k_1},0)\ \cdots\ a(f_{k_\kappa},0)],
%\end{align}
Then
\begin{align}
\sigma_{min}(\bm{A})\|\bm{q}\|_2
& \leq \|\bm{A}^*\bm{q}\|_2 \nonumber\\
& = \sqrt{\sum_{j=1}^S |\bm{a}_{select,j}^*\bm{q}|^2} \nonumber\\
& \leq \sqrt{S},
\end{align}
where the last inequality comes form (\ref{reasoning}). It follows that
\begin{align}\label{Eq:DualNorm2}
\|\bm{q}\|_2 \leq \frac{\sqrt{S}}{\sigma_{min}(\bm{A})}.
\end{align}
Define $\beta$ as the maximal value of $\frac{\sigma_{min}(\bm{A})}{\sqrt{S}}$ among all the choices for $\bm{A}$ and  $S$, namely,
\begin{align}
\beta=\max_{S\geq Q}\max_{\bm{A} \in \bm{M}^{S} }\frac{\sigma_{min}(\bm{A}) }{\sqrt{S}}.
\end{align}
Combining (\ref{Eq:DualNorm1}) and (\ref{Eq:DualNorm2}), we have
\begin{align}\label{Eq:AtomSeparation}
\| \bm{a}(\tau_{k_1})-\bm{a}(\tau_{k_2})  \|_2\geq 2\beta,
\end{align}
proving this theorem.
%and
%\begin{align}
%\frac{\sqrt{\kappa}}{\sigma_{min}(A)}
%\geq \frac{2}{\|a(f_k,0)-a(f_j,0)\|_2},
%\end{align}
\end{proof}

\section{Orthonormal Atomic Norm Minimization: Hankel matrix recovery can be immune from atom separation requirements}
\label{sec:OANM}
 Our results in the previous sections naturally raise the following question: why can Hankel matrix recovery work without requiring separations between the underlying atoms while it is necessary for the atomic norm minimization to require separations between the underlying atoms? In this section, we introduce the concept of orthonormal atomic norm and its minimization, which explains the success of Hankel matrix recovery in recovering the superposition of complex
exponentials regardless of the separations between frequency atoms.

Let us consider a vector $\bm{w} \in \mathbb{C}^N$, where $N$ is a positive integer. We denote the set of atoms by $\mathcal{ATOMSET}$, and assume that each atom $\bm{a} (\tau)$ (parameterized by $\tau$) belongs to $\mathbb{C}^N$. Then the atomic norm of $\|{\bm w}\|_\mathcal{ATOMIC}$ is given by  \cite{candes_towards_2014,tang_compressed_2013}
\begin{align}\label{Defn:AtomicNorm}
\|{\bm w}\|_\mathcal{ATOMIC} = \inf_{s, \tau_k, c_k} \left\{\sum_{k=1}^s |c_k|: \bm{w}=\sum_{k=1}^s c_k\bm{a}(\tau_k)\right\}.
\end{align}
We say the atomic norm $\|{\bm x}\|_\mathcal{ATOMIC}$ is an orthonormal atomic norm if, for every $\bm{x}$,
\begin{align}\label{Defn:OrthoAtomicNorm}
\|{\bm w}\|_\mathcal{ATOMIC} =  \sum_{k=1}^s |c_k| ,
\end{align}
where $\bm{w}=\sum_{k=1}^s c_k\bm{a}(\tau_k)$, $\|\bm{a}(\tau_k)\|_2=1$ for every $k$, and $\bm{a}(\tau_k)$'s are mutually orthogonal to each other.

In the Hankel matrix recovery, the atom set $\mathcal{ATOMSET}$ is composed of all the rank-1 matrices in the form $\bm{u} \bm{v}^* $, where $\bm{u}$ and $\bm{v}$ are unit-norm vectors in $\mathbb{C}^N$.
Let us assume $\bm{x} \in \mathbb{C}^{2N-1}$. We can see the nuclear norm of a Hankel matrix is an orthonormal atomic norm of $\bm{H}(\bm{x})$:
  \begin{align}\label{HankelOrthoAtomicNorm}
\|\bm{H}(\bm{x})\|_{*} =  \sum_{k=1}^R |c_k|,
\end{align}
where $\bm{H}(\bm{x})= \sum_{k=1}^s c_k \bm{u}_{k} \bm{v}_{k}^*$,  $\bm{H}(\bm{x})=\bm{U}\bm{\Sigma}\bm{V}^*$ is the singular value decomposition of $\bm{H}(\bm{x})$, $\bm{u}_{k}$ is the $k$-th column of $\bm{U}$, and $\bm{v}_{k}$ is the $k$-th column of $\bm{V}$.
This is because the matrices $\bm{u}_{k} \bm{v}_{k}^*$'s are orthogonal to each other and each of these rank-1 matrices has unit energy.

Let us now further assume that $\bm{x} \in \mathbb{C}^{2N-1}$ is the superposition of $R$ complex exponentials with $R\leq N$, as defined in (\ref{Def:SignalModel}). Then $\bm{H}(\bm{x})$ is a rank-R matrix, and can be written as
$\bm{H}(\bm{x})=\sum_{k=1}^R c_k \bm{u}_{k} \bm{v}_{k}^*$,
where $\bm{H}(\bm{x})=\bm{U}\bm{\Sigma}\bm{V}^*$ is the singular value decomposition of $\bm{H}(\bm{x})$, $\bm{u}_{k}$ is the $k$-th column of $\bm{U}$, and $\bm{v}_{k}$ is the $k$-th column of $\bm{V}$.
Even though the original $R$ frequency atoms $\bm{a}(\tau_k)$'s for $\bm{x}$ can be arbitrarily close, we can always write $\bm{H}(\bm{x})$ as a superposition of $R$ orthonormal atoms $\bm{u}_{k} \bm{v}_{k}^*$'s
from the singular value decomposition of $\bm{H}(\bm{x})$.  Because the original $R$ frequency atoms $\bm{a}(\tau_k)$'s for $\bm{x}$ can be arbitrarily close, they can violate the necessary separation condition set
forth in (\ref{def:sep}). However, for $\bm{H}(\bm{x})$, its composing atoms can be $R$ orthonormal atoms $\bm{u}_{k} \bm{v}_{k}^*$'s from the singular value decomposition of $\bm{H}(\bm{x})$. These
atoms $\bm{u}_{k} \bm{v}_{k}^*$'s are of unit energy, and are orthogonal to each other. Thus these atoms are well separated and have the opportunity of not violating the necessary separation condition set
forth in (\ref{def:sep}). This explains why the Hankel matrix recovery approach can break free from the separation condition which is required for traditional atomic norm minimizations.

\section{A matrix-theoretic inequality of nuclear norms and its proof from the theory of compressed sensing}
\label{sec:nuclearnorminequality}
In this section, we present a new matrix-theoretic inequality of nuclear norms, and give a proof of it from the theory of compressed sensing (using nuclear norm minimization). To the best of our knowledge, we have not seen the statement of this inequality of nuclear norms, or its proof elsewhere in the literature.

\begin{theorem}\label{SV4MatrixAndSubmatrix}
Let $\bm{A}\in \mathbb{C}^{m \times n}$, and $t=\min (m, n)$. Let $\sigma_1$, $\sigma_2$, ..., and $\sigma_{t}$ be the singular values of $\bm{A}$ arranged in descending order, namely
$$\sigma_1\geq \sigma_2\geq \cdots\geq \sigma_{t}.$$

Let $k$ be any integer such that
$$\sigma_1+\cdots+\sigma_{k} < \sigma_{k+1}+\cdots+\sigma_{t}.$$

Then for any orthogonal projector $P$ onto $k$-dimensional subspaces in $\mathbb{C}^m$, and any orthogonal projector $\bm{Q}$ onto $k$-dimensional subspaces in $\mathbb{C}^n$, we have
 \begin{align}\label{OperatorCaseNN}
 \|\bm{P} \bm{A} \bm{Q}^* \|_* \leq \|(\bm{I}-\bm{P}) \bm{A} (\bm{I}-\bm{Q})^* \|_*.
 \end{align}

In particular,
 \begin{align}\label{NuclearNorm4Submatrix}
 \|\bm{A}_{1:k,1:k}\|_* \leq \|\bm{A}_{(k+1):m,(k+1):n}\|_*,
 \end{align}
where $\bm{A}_{1:k,1:k}$ is the submatrix of $A$ with row indices between 1 and $k$ and column indices between $1$ and $k$, and $\bm{A}_{(k+1):m,(k+1):n}$ is the submatrix of $\bm{A}$ with row indices between $k+1$ and $m$ and column indices between $k+1$ and $n$.
\end{theorem}

 \begin{proof}
We first consider the case where all the elements of $\bm{A}$ are real numbers. Without loss of generality, we consider
\begin{align}
\bm{P}=
\begin{array}{l}
\left[\begin{array}{*{20}{c}}
\bm{I}_{k\times k} & \bm{0}_{k\times(m-k)} \\
\bm{0}_{(m-k) \times k} & \bm{0}_{(m-k) \times (m-k)}
\end{array}\right],
\end{array}
\end{align}
and
\begin{align}
\bm{Q}=
\begin{array}{l}
\left[\begin{array}{*{20}{c}}
\bm{I}_{k\times k} & \bm{0}_{k\times(n-k)} \\
\bm{0}_{(n-k) \times k} & \bm{0}_{(n-k) \times (n-k)}
\end{array}\right].
\end{array}
\end{align}

Then
\begin{align}\label{ColumnSelector}
\bm{P}\bm{A}\bm{Q}^{*}
=
\begin{array}{l}
\left[\begin{array}{*{20}{c}}
\bm{A}_{1,1} &\cdots & \bm{A}_{1,k} &0 &\cdots &0  \\
\vdots &\vdots &\ddots &\vdots & \ddots &\vdots \\
\bm{A}_{k,1} &\cdots &\bm{A}_{k,k} &0 &\cdots &0 \\
\ &0 &\ &\ &0 &\
\end{array}\right]
\end{array},
\end{align}
and
\begin{align}\label{SubmatrixOperator2}
(\bm{I}-\bm{P})\bm{A}(\bm{I}-\bm{Q})^{*}
=\bm{A}-\bm{P}\bm{A}-\bm{A}\bm{Q}^{*}+\bm{P}\bm{A}\bm{Q}^{*}
=
\begin{array}{l}
\left[\begin{array}{*{20}{c}}
\ &0 &\ &\ &0 &\ \\
0 &\cdots &0 &\bm{A}_{k+1,k+1} &\cdots &\bm{A}_{k+1,n}\\
\vdots &\ddots& \vdots  &\vdots &\ddots& \vdots\\
0 &\cdots &0 &\bm{A}_{m,k+1} &\cdots &\bm{A}_{m,n}
\end{array}\right]
\end{array}.
\end{align}

Thus
\begin{align}
\|\bm{P}\bm{A}\bm{Q}^{*}\|_{*} = \|\bm{A}_{1:k,1:k}\|_{*},~
\| (\bm{I}-\bm{P})\bm{A}(\bm{I}-\bm{Q})^{*}\|_{*}=\|\bm{A}_{k+1:m,k+1:n}\|_{*}.
\end{align}

To prove this theorem, we first show that if $\sigma_1+\cdots+\sigma_{k} < \sigma_{k+1}+\cdots+\sigma_{t}$, for any matrix $\bm{X} \in \mathbb{C}^{m \times n}$ with rank no more than $k$, for any positive number $l>0$, we have
\begin{align}
\|\bm{X}+l\bm{A}\|_*> \|\bm{X}\|_{*}
\label{nu_increase}
\end{align}

The proof of (\ref{nu_increase}) follows similar arguments as in \cite{oymak_new_2010}.
\begin{align}
&\|\bm{X}+l\bm{A}\|_*\\
&\geq \sum_{i=1}^{t} |\sigma_i(\bm{X})-\sigma_i(l\bm{A})|\\
&\geq \sum_{i=1}^{k} (\sigma_i(\bm{X})-\sigma_i(l\bm{A}))+\sum_{i=k+1}^{t} |\sigma_i(\bm{X})-\sigma_i(l\bm{A})|\\
&\geq \sum_{i=1}^{k} \sigma_i(\bm{X})+(\sum_{i=k+1}^{t} \sigma_i(l\bm{A})-\sum_{i=1}^{k}\sigma_i(l\bm{A}))\\
&>\sum_{i=1}^{k} \sigma_i(\bm{X}) = \|\bm{X}\|_{*},
\end{align}
where, for the first inequality, we used the following lemma, which instead follows from Lemma \ref{NuclearNorm4MatrixDiffandNuclearNormforSingularValueDiff}.
\begin{lem}
  Let $\bm{G}$ and $\bm{H}$ be two matrices of the same dimension. Then $\sum_{i=1}^{t} |\sigma_i(\bm{G})-\sigma_i(\bm{H})| \leq \|\bm{G}-\bm{H}\|_*$.
\end{lem}

\begin{lem}\label{NuclearNorm4MatrixDiffandNuclearNormforSingularValueDiff}
( \cite{horn_matrix_2012,bhatia_matrix_2013})  For arbitrary matrices $\bm{X}, \bm{Y}$, and $\bm{Z}=\bm{X}-\bm{Y}\in \mathbb{C}^{m\times n}$. Let Let $\sigma_1$, $\sigma_2$, ..., and $\sigma_{t}$ ($t=\min\{m,n\}$) be the singular values of $\bm{A}$ arranged in descending order, namely $\sigma_1\geq \sigma_2\geq \cdots\geq \sigma_{t}.$ Let $s_{i}(\bm{X},\bm{Y})$ be the distance between the $i$-th singular value of $\bm{X}$ and $\bm{Y}$, namely,
\begin{align}
s_{i}(\bm{X},\bm{Y})=|\sigma_{i}(\bm{X})-\sigma_{i}(\bm{Y})|,i=1,2,\cdots,k.
\end{align}

Let $s_{[i]}(\bm{X},\bm{Y})$ be the $i$-th largest value of sequence $s_{1}(\bm{X},\bm{Y}), s_2(\bm{X},\bm{Y}), \cdots, s_t(\bm{X},\bm{Y})$, then
\begin{align}
\sum_{i=1}^{k} s_{[i]}(\bm{X},\bm{Y}) \leq \|\bm{Z}\|_{k}, \forall k=1,2, \cdots, t,
\end{align}
where $\|\bm{Z}\|_k$ is defined as $\sum_{i=1}^k \sigma_i(\bm{Z})$.
\end{lem}

We next show that if $\|\bm{A}_{1:k,1:k}\|_* > \|\bm{A}_{(k+1):m,(k+1):n}\|_*$, one can construct a matrix $\bm{X}$ with rank at most $k$ such that
$$ \|\bm{X}+l\bm{A}\|_*\leq \|\bm{X}\|_{*}$$
for a certain $l>0$. We divide this construction into two cases: when $\bm{A}_{1:k,1:k}$ has rank equal to $k$, and when $\bm{A}_{1:k,1:k}$ has rank smaller than $k$.

When $\bm{A}_{1:k,1:k}$ has rank $k$, we denote its SVD as
$$\bm{A}_{1:k,1:k}={\bm{U}_1} {\bm{\Sigma}} {\bm{V}_1}^{*}. $$
Then the SVD of $\begin{array}{l}
\left[\begin{array}{*{20}{c}}
 {\bm{A}_{1:k,1:k}} & 0\\
 0 & 0
\end{array}\right]
\end{array}$ is given by

\begin{align}
\begin{array}{l}
\left[\begin{array}{*{20}{c}}
 \bm{A}_{1:k,1:k} & 0\\
 0 & 0
\end{array}\right]
\end{array}
=
\begin{array}{l}
\left[\begin{array}{*{20}{c}}
{\bm{U}}_{1} \\
0
\end{array}\right]
\end{array}
{\bm{\Sigma}}
\begin{array}{l}
\left[\begin{array}{*{20}{c}}
{\bm{V}}_{1} \\
0
\end{array}\right]^{*}
\end{array}
\end{align}

We now construct
$$\bm{X}=-\begin{array}{l}
\left[\begin{array}{*{20}{c}}
{\bm{U}}_{1} \\
0
\end{array}\right]
\end{array}
\begin{array}{l}
\left[\begin{array}{*{20}{c}}
{\bm{V}}_{1} \\
0
\end{array}\right]^{*}
\end{array}.$$

Let us denote
$$\bm{U}_2=\begin{array}{l}
\left[\begin{array}{*{20}{c}}
{\bm{U}}_{1} \\
0
\end{array}\right]
\end{array}$$
and
$$\bm{V}_2=\begin{array}{l}
\left[\begin{array}{*{20}{c}}
{\bm{V}}_{1} \\
0
\end{array}\right]
\end{array},$$
then the subdifferential of $\| \cdot \|_{*}$ at $\bm{X}$ is given by% ,from Lemma \ref{lemma:complexdifferential},
\begin{align}\label{Subdifferential4NN}
\partial \| \bm{X} \|_{*}
=\{{\bm{Z}}: {\bm{Z}}=
-\bm{U}_2 \bm{V}_2^*      + \bm{M},
{\rm where\ }
\|\bm{M}\|_2\leq 1,
\bm{M}^{*} \bm{X}={0},
\bm{X}\bm{M}^{*}=0\}.
\end{align}
For any $\bm{Z}\in \partial \| \bm{X} \|_{*}$,
\begin{align}\label{Nega1}
    \langle \bm{Z},\bm{A} \rangle
 =-I_{1}+I_{2},
\end{align}
where
\begin{align}
I_{1} = {\rm Tr}( \bm{V}_2 \bm{U}_2^* \bm{A}), I_{2}={\rm Tr}(\bm{M}^{*} \bm{A}).
\end{align}
Let us partition the matrix $\bm{A}$ into four blocks:
\begin{align}
     \begin{array}{l}
\left[\begin{array}{*{20}{c}}
\bm{A}_{11} & \bm{A}_{12}\\
\bm{A}_{21} & \bm{A}_{22}
\end{array}\right],
\end{array}
\end{align}
where $\bm{A}_{11} \in \mathbb{R}^{k \times k}$, $\bm{A}_{12} \in \mathbb{R}^{k \times (n-k)}$, $\bm{A}_{21} \in \mathbb{R}^{(m-k) \times (n-k)}$, and $\bm{A}_{21} \in \mathbb{R}^{(m-k) \times (n-k)}$. Then we have

%For $I_{1}$, we have
\begin{align}\label{I1Estimate}
     I_{1}
     &= {\rm Tr}({\bm{V}_2}{\bm{U}_2}^{*}\bm{A})
     ={\rm Tr}\left(
     \begin{array}{l}
     \left[\begin{array}{*{20}{c}}
     {\bm{V}}_{1}\\
     0
     \end{array}\right]
     \end{array}
     \begin{array}{l}
\left[\begin{array}{*{20}{c}}
{\bm{U}}_{1}^{*} & 0
\end{array}\right]
\end{array}
     \begin{array}{l}
\left[\begin{array}{*{20}{c}}
\bm{A}_{11} & \bm{A}_{12}\\
\bm{A}_{21} & \bm{A}_{22}
\end{array}\right]
\end{array}
     \right) \nonumber\\
&     ={\rm Tr}\left(
     \begin{array}{l}
     \left[\begin{array}{*{20}{c}}
     {\bm{V}}_{1}{\bm{U}}_{1}^{*} & 0\\
     0 &0
     \end{array}\right]
     \end{array}
     \begin{array}{l}
\left[\begin{array}{*{20}{c}}
\bm{A}_{11} & \bm{A}_{12}\\
\bm{A}_{21} & \bm{A}_{22}
\end{array}\right]
\end{array}
     \right) \nonumber\\
     & ={\rm Tr}({\bm{V}}_{1}{\bm{U}}_{1}^{*} \bm{A}_{11})={\rm Tr}({\bm{V}}_{1}{\bm{U}}_{1}^{*} {\bm{U}}_{1} {\bm{\Sigma}} {\bm{V}}_{1}^{*})
     =\sum_{i=1}^{k} \sigma_{i}({\bm{A}_{11}})
     =\|{\bm{A}_{11}}\|_{*}.
\end{align}

Since $\bm{M}^{*} \bm{X}={0}$, $\bm{X}\bm{M}^{*}=0$, $\|M\|_2\leq 1$ and $\bm{X}$ is a rank-$k$ left top corner matrix, we must have
$$\bm{M}=\left[\begin{array}{*{20}{c}}
0 & 0\\
0 & \bm{M}_{22}
\end{array}\right],$$
where $\bm{M}_{22}$ is of dimension $(m-k)\times (n-k)$, and $\|\bm{M}_{22}\|_2\leq 1$.

Then
\begin{align}\label{I2Estimate1}
    I_{2}
    &={\rm Tr}(\bm{M}^{*}\bm{A})\\
    &={\rm Tr}\left(
    \begin{array}{l}
\left[\begin{array}{*{20}{c}}
0 & 0\\
0 & \bm{M}_{22}^*
\end{array}\right]
\end{array}
\begin{array}{l}
\left[\begin{array}{*{20}{c}}
\bm{A}_{11} & \bm{A}_{12}\\
\bm{A}_{21} & \bm{A}_{22}
\end{array}\right]
\end{array}
\right)\\
&={\rm Tr}\left( \bm{M}_{22}^* \bm{A}_{22}\right) \leq \|\bm{A}_{22}\|_{*},
\end{align}
where the last inequality is from the fact that the nuclear norm is the dual norm of the spectral norm, and $\|\bm{M}_{22}\|_2\leq 1$.

Thus we have
\begin{align}
    \langle \bm{Z},\bm{A} \rangle
 =-I_{1}+I_{2} \leq - \|{\bm{A}_{11}}\|_{*}+\|\bm{A}_{22}\|_{*}<0,
\end{align}
because we assume that $\|\bm{A}_{1:k,1:k}\|_* > \|\bm{A}_{(k+1):m,(k+1):n}\|_*$.
Since $\langle \bm{Z},\bm{A} \rangle<0$ for every $\bm{Z} \in \partial \| \bm{X} \|_{*}$, $\bm{A}$ is in the normal cone of the convex cone generated by $\partial \|\cdot\|_*$ at the point
$\bm{X}$. By Theorem 23.7 in \cite{rockafellar_convex_2015},  we know that the normal cone of the convex cone generated by $\partial \|\cdot\|_*$ at the point
$\bm{X}$ is the cone of descent directions for $\|\cdot\|_{*}$ at the point of $\bm{X}$. Thus $\bm{A}$ is in the descent cone of $\| \cdot \|_{*}$ at the point $\bm{X}$.
This means that, when $\bm{A}_{11}$ has rank equal to $k$, there exists a positive number $l>0$, such that
$$\|\bm{X}+l\bm{A}\|_{*} \leq \|\bm{X}\|_*.$$

Let us suppose instead that $\bm{A}_{11}$ has rank $b<k$. We can write the SVD of $\bm{A}_{11}$ as
\begin{align}
 \bm{A}_{11}
=\begin{array}{l}
\left[\begin{array}{*{20}{c}}
{\bm{U}}_{1} & \bm{U}_{3}
\end{array}\right]
\end{array}
\begin{array}{l}
\left[\begin{array}{*{20}{c}}
\bm{\Sigma} & 0\\
0 & 0
\end{array}\right]
\end{array}
\begin{array}{l}
\left[\begin{array}{*{20}{c}}
{\bm{V}}_{1} & \bm{V}_{3}
\end{array}\right]^{*}
\end{array}
\end{align}

Then we construct $$\bm{X}=-\begin{array}{l}
\left[\begin{array}{*{20}{c}}
{\bm{U}}_{1} & \bm{U}_{3}\\
0 & 0
\end{array}\right]
\end{array}  \begin{array}{l}
\left[\begin{array}{*{20}{c}}
{\bm{V}}_{1} & \bm{V}_{3}\\
0 & 0
\end{array}\right]^{*}
\end{array}.
    $$
By going through similar arguments as above (except for taking care of extra terms involving $\bm{U}_3$ and $\bm{V}_3$), one can obtain that $\langle \bm{Z},\bm{A} \rangle<0$ for every $\bm{Z} \in \partial \| \bm{X} \|_{*}$.

In summary, no matter whether $\bm{A}_{11}$ has rank equal to $k$ or smaller to $k$, there always exists a positive number $l>0$, such that$\|\bm{X}+l\bm{A}\|_{*} \leq \|\bm{X}\|_*$.
However, this contradicts (\ref{nu_increase}), and we conclude $\|\bm{A}_{1:k,1:k}\|_* \leq \|\bm{A}_{(k+1):m,(k+1):n}\|_*$, when $\bm{A}$ has real-numbered elements.

We further consider the case when $\bm{A}$ is a complex-numbered matrix. We first derive the subdifferential of $\|\bm{X}\|_{*}$ for any complex-numbered matrix $m \times n$ $\bm{X}$.
For any $\bm{\alpha} \in \mathbb{R}^{m \times n}$ and any $\bm{\beta}\in \mathbb{R}^{m \times n}$, we define $\mathcal{F}~:~\mathbb{R}^{2m \times n}\mapsto\mathbb{R}$ as
\begin{equation}\label{def:F}
\mathcal{F}\left(\left[\begin{matrix}\bm{\alpha}\cr\bm{\beta}\end{matrix}\right]\right)=\|(\bm{\alpha}+\imath\bm{\beta})\|_*.
\end{equation}
To find the subdifferential of $\|\bm{X}\|_{*}$, we need to derive $\partial \mathcal{F}\left( \left[\begin{matrix}{\rm Re}(\bm{X})\cr {\rm Im}(\bm{X})\end{matrix}\right]   \right)$, for which we have the following lemma,
the proof of which is given in Appendix \ref{proofofcomplexsubdifferential}. (Note that in our earlier work \cite{cai_robust_2016} and the corresponding preprint on arxiv.org, we have already shown one direction of \ref{eq:GsubsetsubdiffF}, namely $\mathfrak{H} \subseteq \partial\mathcal{F}\left(\left[\begin{matrix}{\rm Re}(\bm{X})\cr {\rm Im}(\bm{X})\end{matrix}\right]\right)$. Now we show the two sets are indeed equal. )

\begin{lem}
\label{lemma:complexdifferential}
Suppose a rank-$R$ matrix $\bm{X}\in \mathbb{C}^{M \times N}$ admits a singular value decomposition  $\bm{X}=\bm{U}\bm{\bm{\Sigma}}\bm{V}^*$,
where $\bm{{\Sigma}}\in\mathbb{R}^{R\times R}$ is a diagonal matrix, and $\bm{U}\in\mathbb{C}^{M\times R}$ and
$\bm{V}\in\mathbb{C}^{N\times R}$ satisfy $\bm{U}^*\bm{U}=\bm{V}^*\bm{V}=\bm{I}$.

Define $\mathfrak{S} \subseteq \mathbb{C}^{M \times N}$ as:
\begin{equation}\label{eq:setS}
\mathfrak{S}=\left\{\bm{U}\bm{V}^*+\bm{W}~|~\bm{U}^*\bm{W}=\bm{0},~\bm{W}\bm{V}=\bm{0},
~\|\bm{W}\|_2\leq 1,~\bm{W} \in \mathbb{C}^{M \times N}\right\},
\end{equation}
and define
$$\mathcal{F}\left(\left[\begin{matrix}{\rm Re}(\bm{X})\cr {\rm Im}(\bm{X})\end{matrix}\right]\right)=\|\bm{X}\|_*.$$
Then we have
\begin{equation}\label{eq:GsubsetsubdiffF}
\mathfrak{H}\equiv \left\{\left[\begin{matrix}\bm{\alpha}\cr\bm{\beta}\end{matrix}\right]~\Big|~
\bm{\alpha}+\imath\bm{\beta}\in\mathfrak{S}\right\}
= \partial\mathcal{F}\left(\left[\begin{matrix}{\rm Re}(\bm{X})\cr {\rm Im}(\bm{X})\end{matrix}\right]\right).
\end{equation}
\end{lem}

For a complex-numbered matrix $\bm{A}$, we will similarly show that, if $\|\bm{A}_{1:k,1:k}\|_* > \|\bm{A}_{(k+1):m,(k+1):n}\|_*$, one can construct a matrix $\bm{X}$ with rank at most $k$ such that
$$ \|\bm{X}+l\bm{A}\|_*\leq \|\bm{X}\|_{*}$$
for a certain $l>0$. We divide this construction into two cases: when $\bm{A}_{1:k,1:k}$ has rank equal to $k$, and when $\bm{A}_{1:k,1:k}$ has rank smaller than $k$.

When $\bm{A}_{1:k,1:k}$ has rank $k$, we denote its SVD as
$$\bm{A}_{1:k,1:k}={\bm{U}_1} {\bm{\Sigma}} {\bm{V}_1}^{*}. $$
Then the SVD of $\begin{array}{l}
\left[\begin{array}{*{20}{c}}
 {\bm{A}_{1:k,1:k}} & 0\\
 0 & 0
\end{array}\right]
\end{array}$ is given by

\begin{align}
\begin{array}{l}
\left[\begin{array}{*{20}{c}}
 \bm{A}_{1:k,1:k} & 0\\
 0 & 0
\end{array}\right]
\end{array}
=
\begin{array}{l}
\left[\begin{array}{*{20}{c}}
{\bm{U}}_{1} \\
0
\end{array}\right]
\end{array}
{\bm{\Sigma}}
\begin{array}{l}
\left[\begin{array}{*{20}{c}}
{\bm{V}}_{1} \\
0
\end{array}\right]^{*}
\end{array}
\end{align}

We now construct
$$\bm{X}=-\begin{array}{l}
\left[\begin{array}{*{20}{c}}
{\bm{U}}_{1} \\
0
\end{array}\right]
\end{array}
\begin{array}{l}
\left[\begin{array}{*{20}{c}}
{\bm{V}}_{1} \\
0
\end{array}\right]^{*}
\end{array}.$$
We denote
$$\bm{U}_2=\begin{array}{l}
\left[\begin{array}{*{20}{c}}
{\bm{U}}_{1} \\
0
\end{array}\right]
\end{array},~\text{and}~~\bm{V}_2=\begin{array}{l}
\left[\begin{array}{*{20}{c}}
{\bm{V}}_{1} \\
0
\end{array}\right]
\end{array},$$
then by Lemma \ref{lemma:complexdifferential},  the subdifferential of $\| \cdot \|_{*}$ at $\bm{X}$ is given by
\begin{align}\label{Subdifferential4NN_Complex}
\partial \| \bm{X} \|_{*}
=\{{\bm{Z}}: {\bm{Z}}=
-\bm{U}_2 \bm{V}_2^*      + \bm{M}],
{\rm where\ }
\|\bm{M}\|_2\leq 1,
\bm{M}^{*} \bm{X}={0},
\bm{X}\bm{M}^{*}=0\}.
\end{align}
For any $\bm{Z}\in \partial \| \bm{X} \|_{*}$,
\begin{align}\label{Nega1_complex}
    \langle \bm{Z},\bm{A} \rangle
 =-I_{1}+I_{2},
\end{align}
where %{\rm Re} \left\{ \right\}
\begin{align}
I_{1} = {\rm Re} \left(  {\rm Tr}( \bm{V}_2 \bm{U}_2^* \bm{A}) \right),
I_{2}= {\rm Re} \left( {\rm Tr}(M^{*} \bm{A}) \right) .
\end{align}
Similar to the real-numbered case, let us partition the matrix $\bm{A}$ into four blocks:
\begin{align}
     \begin{array}{l}
\left[\begin{array}{*{20}{c}}
\bm{A}_{11} & \bm{A}_{12}\\
\bm{A}_{21} & \bm{A}_{22}
\end{array}\right],
\end{array}
\end{align}
where $\bm{A}_{11} \in \mathbb{C}^{k \times k}$, $\bm{A}_{12} \in \mathbb{C}^{k \times (n-k)}$, $\bm{A}_{21} \in \mathbb{C}^{(m-k) \times (n-k)}$, and $\bm{A}_{21} \in \mathbb{C}^{(m-k) \times (n-k)}$.
We still have
\begin{align}\label{I1Estimate_complex}
{\rm Tr}({\bm{V}_2}{\bm{U}_2}^{*}\bm{A})
     &={\rm Tr}\left(
     \begin{array}{l}
     \left[\begin{array}{*{20}{c}}
     {\bm{V}}_{1}\\
     0
     \end{array}\right]
     \end{array}
     \begin{array}{l}
\left[\begin{array}{*{20}{c}}
{\bm{U}}_{1}^{*} & 0
\end{array}\right]
\end{array}
     \begin{array}{l}
\left[\begin{array}{*{20}{c}}
\bm{A}_{11} & \bm{A}_{12}\\
\bm{A}_{21} & \bm{A}_{22}
\end{array}\right]
\end{array}
     \right) \nonumber\\
     &={\rm Tr}\left(
     \begin{array}{l}
     \left[\begin{array}{*{20}{c}}
     {\bm{V}}_{1}{\bm{U}}_{1}^{*} & 0\\
     0 &0
     \end{array}\right]
     \end{array}
     \begin{array}{l}
\left[\begin{array}{*{20}{c}}
\bm{A}_{11} & \bm{A}_{12}\\
\bm{A}_{21} & \bm{A}_{22}
\end{array}\right]
\end{array}
     \right) \nonumber\\
     &={\rm Tr}({\bm{V}}_{1}{\bm{U}}_{1}^{*} {\bm A}_{11})={\rm Tr}({\bm{V}}_{1}{\bm{U}}_{1}^{*} {\bm{U}}_{1} {\bm{\Sigma}} {\bm{V}}_{1}^{*})
     =\sum_{i=1}^{k} \sigma_{i}({\bm{A}_{11}})
     =\|{\bm{A}_{11}}\|_{*}.
\end{align}
So $$I_1= {\rm Re} \left( {\rm Tr}({\bm{V}_2}{\bm{U}_2}^{*}\bm{A})\right)=\|{\bm{A}_{11}}\|_{*}.$$

Since $\bm{M}^{*} \bm{X}={0}$, $\bm{X}\bm{M}^{*}=0$, $\|\bm{M}\|_2\leq 1$ and $\bm{X}$ is a rank-$k$ left top corner matrix, we must have
$$M=\left[\begin{array}{*{20}{c}}
0 & 0\\
0 & \bm{M}_{22}
\end{array}\right],$$
where $\bm{M}_{22}$ is of dimension $(m-k)\times (n-k)$, and $\|\bm{M}_{22}\|_2\leq 1$.
Then we have
\begin{align}\label{I2Estimate1_Complex}
    I_{2}
    &={\rm Re} \left( {\rm Tr}(\bm{M}^{*}\bm{A}) \right)    \\
    &={\rm Re} \left(    {\rm Tr}\left(
    \begin{array}{l}
\left[\begin{array}{*{20}{c}}
0 & 0\\
0 & \bm{M}_{22}^*
\end{array}\right]
\end{array}
\begin{array}{l}
\left[\begin{array}{*{20}{c}}
\bm{A}_{11} & \bm{A}_{12}\\
\bm{A}_{21} & \bm{A}_{22}
\end{array}\right]
\end{array}
\right)   \right)  \\
&= {\rm Re} \left( {\rm Tr}\left( \bm{M}_{22}^* \bm{A}_{22}\right) \right) \leq \|\bm{A}_{22}\|_{*},
\end{align}
where the last inequality is because the nuclear norm is the dual norm of the spectral norm, and $\|\bm{M}_{22}\|_2\leq 1$.

Thus we have
\begin{align}
    \langle \bm{Z},\bm{A} \rangle
 =-I_{1}+I_{2} \leq - \|{\bm{A}_{11}}\|_{*}+\|\bm{A}_{22}\|_{*}<0,
\end{align}
because we assume that $\|\bm{A}_{1:k,1:k}\|_* > \|\bm{A}_{(k+1):m,(k+1):n}\|_*$.
Since $\langle \bm{Z},\bm{A} \rangle<0$ for every $\bm{Z} \in \partial \| \bm{X} \|_{*}$, $\bm{A}$ is in the descent cone of $\| \cdot \|_{*}$ at the point $\bm{X}$.
This means that, when $\bm{A}_{11}$ has rank equal to $k$, there exists a positive number $l>0$, such that
$$\|\bm{X}+l\bm{A}\|_{*} \leq \|\bm{X}\|_*.$$

Let us suppose instead that $\bm{A}_{11}$ has rank $b<k$. We can write the SVD of $\bm{A}_{11}$ as

\begin{align}
 \bm{A}_{11}
=\begin{array}{l}
\left[\begin{array}{*{20}{c}}
{\bm{U}}_{1} & \bm{U}_{3}
\end{array}\right]
\end{array}
\begin{array}{l}
\left[\begin{array}{*{20}{c}}
\bm{\Sigma} & 0\\
0 & 0
\end{array}\right]
\end{array}
\begin{array}{l}
\left[\begin{array}{*{20}{c}}
{\bm{V}}_{1} & \bm{V}_{3}
\end{array}\right]^{*}
\end{array}
\end{align}

Then we construct $$\bm{X}=-\begin{array}{l}
\left[\begin{array}{*{20}{c}}
{\bm{U}}_{1} & \bm{U}_{3}\\
0 & 0
\end{array}\right]
\end{array}  \begin{array}{l}
\left[\begin{array}{*{20}{c}}
{\bm{V}}_{1} & \bm{V}_{3}\\
0 & 0
\end{array}\right]^{*}
\end{array}.
    $$
By going through similar arguments as above (except for taking care of extra terms involving $\bm{U}_3$ and $\bm{V}_3$), one can obtain that $\langle \bm{Z},\bm{A} \rangle<0$ for every $\bm{Z} \in \partial \| \bm{X} \|_{*}$.

In summary, no matter whether complex-numbered $\bm{A}_{11}$ has rank equal to $k$ or smaller to $k$, there always exists a positive number $l>0$, such that$\|\bm{X}+l\bm{A}\|_{*} \leq \|\bm{X}\|_*$. However, this contradicts (\ref{nu_increase}), and we conclude $\|\bm{A}_{1:k,1:k}\|_* \leq \|\bm{A}_{(k+1):m,(k+1):n}\|_*$.

\end{proof}

\section{Numerical results}
\label{sec:NumericalResult}

In this section, we perform numerical experiments to demonstrate the empirical performance of Hankel matrix recovery, and show its robustness to the separations between atoms.
We use superpositions of complex sinusoids as test signals. But we remark that Hankel matrix recovery can also work for superpositions of complex exponentials.
We consider the non-uniform sampling of entries studied in \cite{tang_compressed_2013,chen_robust_2014}, where we uniformly randomly observe $M$ entries (without replacement) of ${\bm{x}}$
from  $\{0,1,\ldots,2N-2\}$. We also consider two signal
(frequency) reconstruction algorithms: the Hankel nuclear norm minimization and the atomic norm minimization.

We fix $N=64$, i.e., the dimension of the ground truth signal ${\bm{x}}$ is $127$. We conduct experiments under different $M$ and $R$ for different approaches. For each approach
with a fixed $M$ and $R$, we test $100$ trials, where each trial is performed as follows. We first generate the true signal ${\bm{x}}=[{\bm{x}}_{0},{\bm{x}}_1,\ldots,{\bm{x}}_{126}]^T$ with
${\bm{x}}_{t}=\sum_{k=1}^{R}c_ke^{\imath 2\pi f_k t}$ for $t=0,1,\ldots,126$, where $f_k$ are frequencies drawn from the interval $[0,1]$ uniformly at random, and $c_k$ are complex
coefficients satisfying the model $c_k=(1+10^{0.5m_k})e^{i2\pi\theta_k}$ with $m_k$ and $\theta_k$ uniformly randomly drawn from the interval $[0,1]$. Let the reconstructed signal be represented by $\hat{\bm{x}}$. If $\frac{\|\hat{\bm{x}}-{\bm{x}}\|_2}{\|{\bm{x}}\|_2}\leq 10^{-3}$, then we regard it as a successful reconstruction. We also provide the simulation results under the Gaussian measurements of $\bm{x}$ as in
\cite{cai_robust_2016}.

We plot in Figure  \ref{fig:mean and std of nets} the rate of successful reconstruction with respect to different $M$ and $R$ for different approaches. The black and white region indicate a 0\% and 100\% of
successful reconstruction respectively, and a grey between 0\% and 100\%.  From the figure, we see that the atomic norm minimization still suffers from non-negligible failure probability even if the number of
measurements approach the full 127 samples. The reason is that, since the underlying frequencies are randomly chosen, there is a sizable probability that some frequencies are close to each other. When the
frequencies are close to each other violating the atom separation condition, the atomic norm minimization can still fail even if we observe the full 127 samples.
By comparison, the Hankel matrix recovery approach experiences a sharper phase transition, and is robust to the frequency separations. We also see that under both
the Gaussian projection and the non-uniform sampling models, both the atomic norm minimization and the Hankel matrix recovery approach have similar performance.
%we see that the atomic norm minimization has similar performance under the random Gaussian sampling and the
%non-uniform sampling of entries.  Moreover, the Hankel nuclear norm minimization also has similar performance under these two types of different sampling schemes.  Compared
%with the atomic norm minimization, the Hankel nuclear norm minimization method is more robust when neighboring frequencies are close, despite different sampling schemes used.
    \begin{figure}
        \centering
        \begin{subfigure}[b]{0.475\textwidth}
            \centering
            \includegraphics[width=\textwidth]{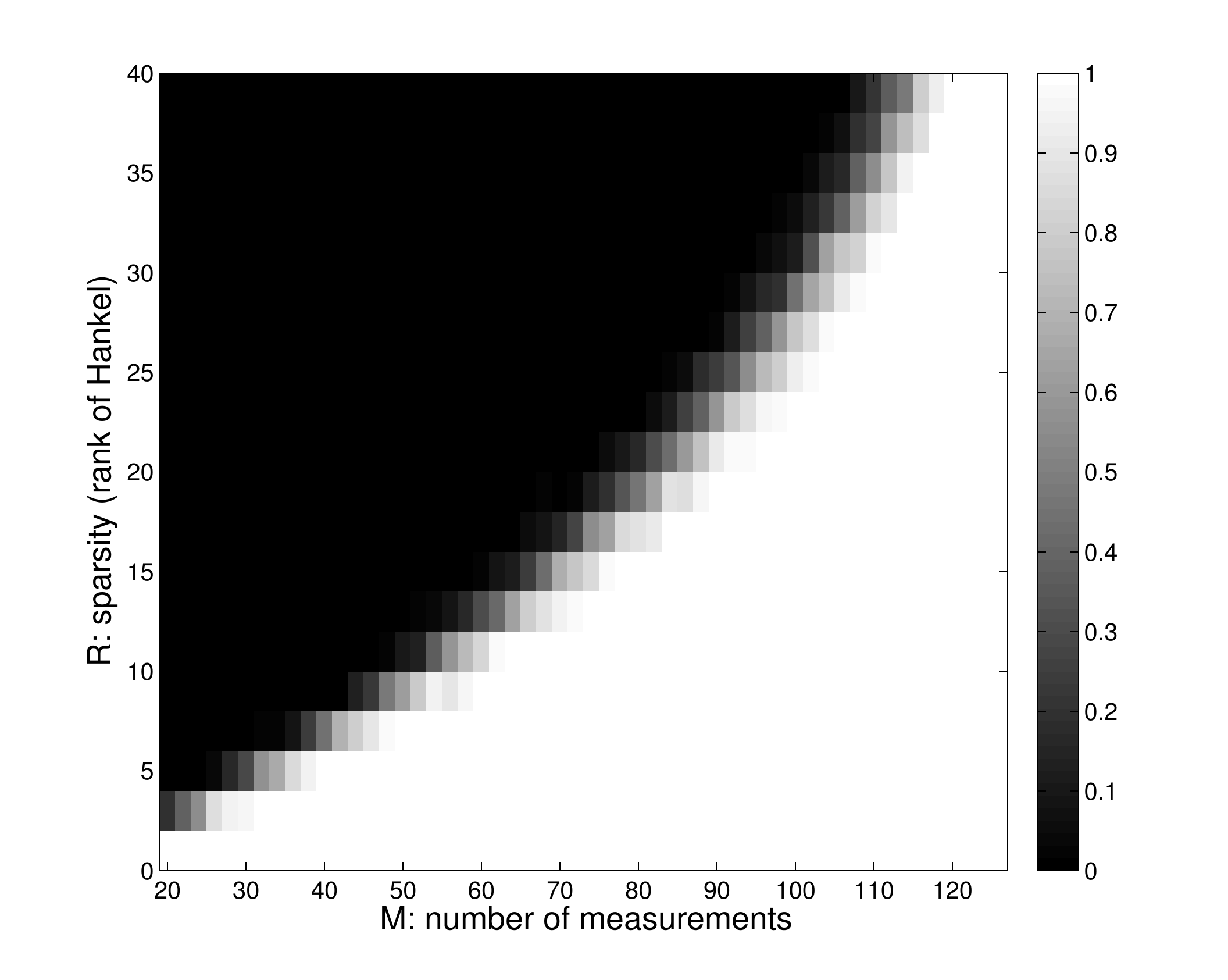}
            \caption[]%
            {{\small Hankel nuclear norm minimization with random Gaussian projections}}
        \end{subfigure}
        \hfill
        \begin{subfigure}[b]{0.475\textwidth}
            \centering
            \includegraphics[width=\textwidth]{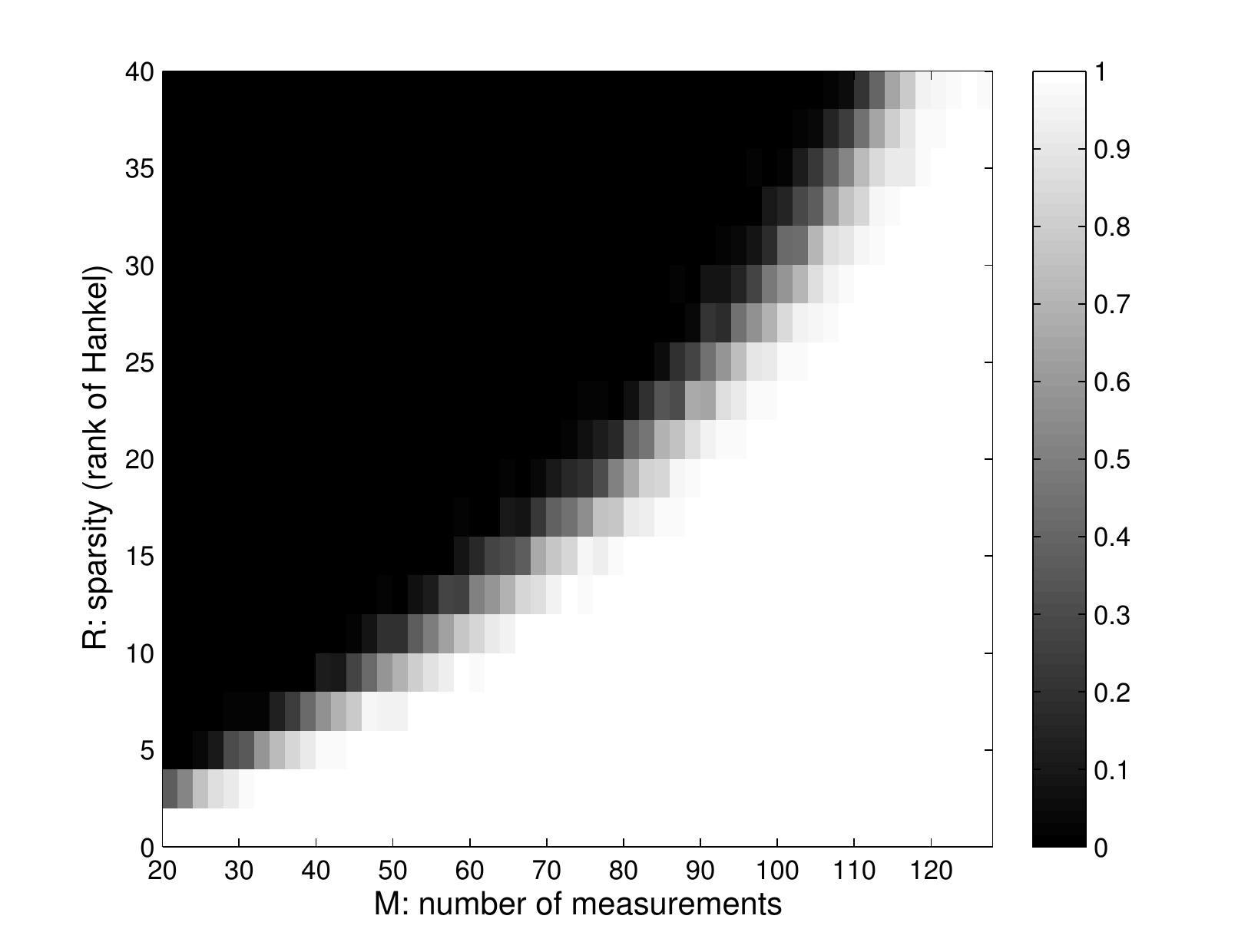}
            \caption[]%
            {{\small Hankel nuclear norm minimization with non-uniform sampling of entries}}
        \end{subfigure}
        \vskip\baselineskip
        \begin{subfigure}[b]{0.475\textwidth}
            \centering
            \includegraphics[width=\textwidth]{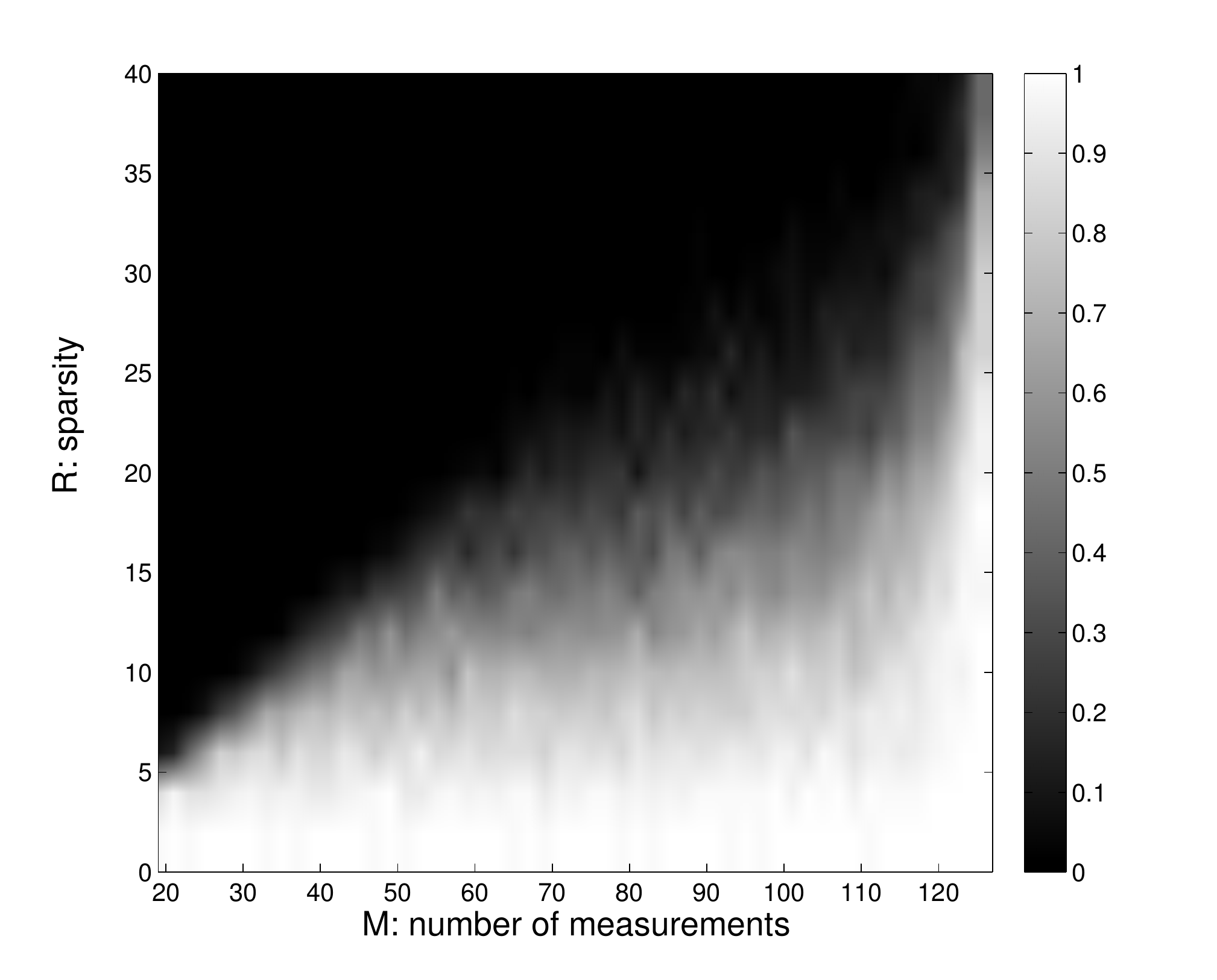}
            \caption[]%
{{\small Atomic norm minimization with random Gaussian projections}}
        \end{subfigure}
        \quad
        \begin{subfigure}[b]{0.475\textwidth}
            \centering
            \includegraphics[width=\textwidth]{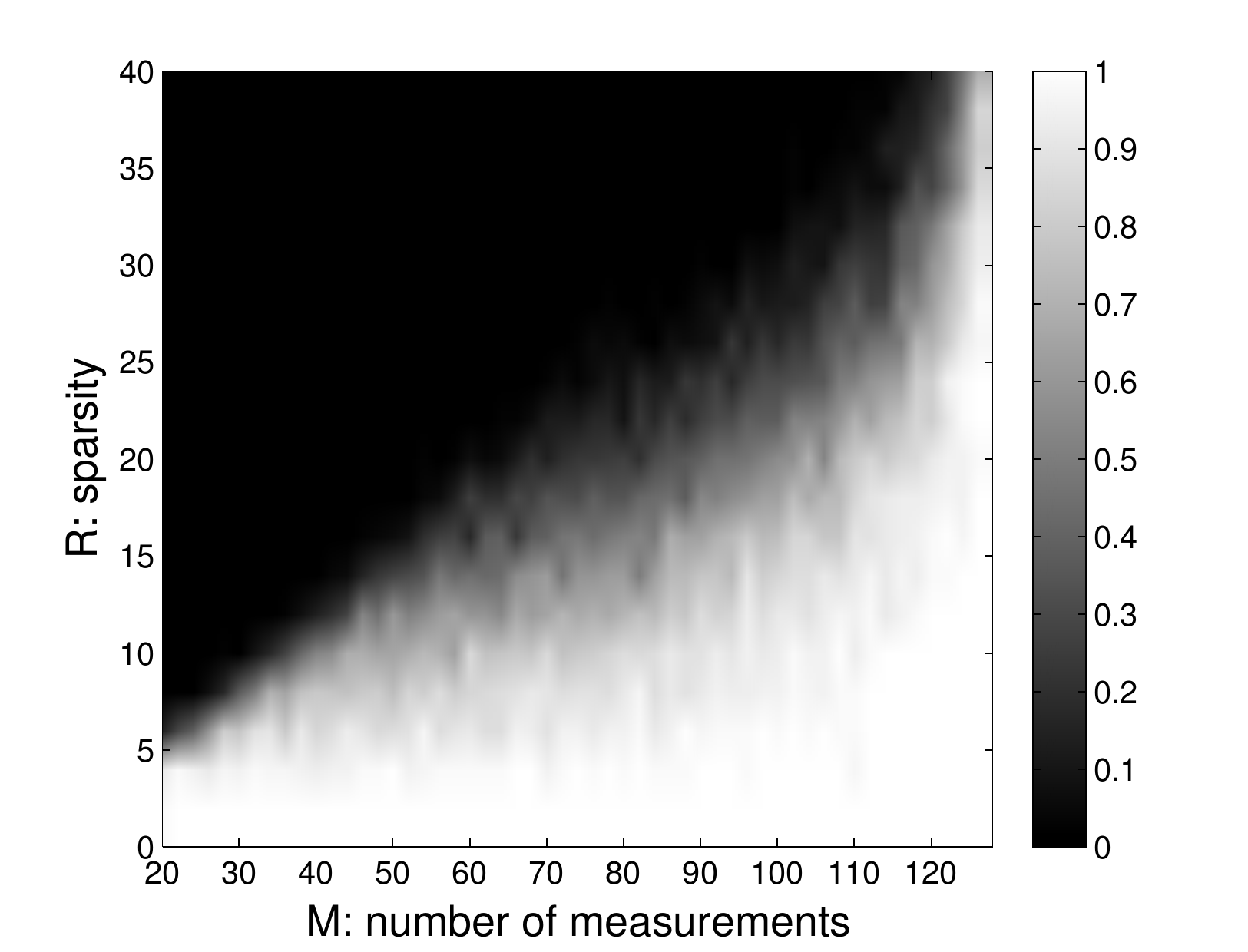}
            \caption[]%
            {{\small Atomic norm minimization with non-uniform sampling of entries}}
        \end{subfigure}
        \caption[]
        {\small Performance comparisons between atomic norm minimization and Hankel matrix recovery}
        \label{fig:mean and std of nets}
    \end{figure}

We further demonstrate the robustness of the Hankel matrix recovery approach to the separations between frequency atoms, as we vary the separations between frequencies. In our first set of experiments, we take $N=64, |\mathcal{M}|=M=65$ ($\approx 51\%$ sampling rate) and $R=8$, and consider noiseless measurements. Again we generate the magnitude of the coefficients as $1+10^{0.5 p}$ where $p$ is uniformly randomly
generated from $[0,1)$, and the realized magnitudes are 3.1800, 2.5894, 2.1941, 2.9080, 3.9831, 4.0175, 4.1259, 3.6182 in this experiment.
The corresponding phases of the coefficients are randomly generated as $2\pi s$, where $s$ is uniformly randomly generated from $[0,1)$. In this experiment, the realized phases
are 4.1097, 5.4612, 5.4272, 4.7873, 1.0384, 0.4994, 3.1975, and 0.5846. The first $R-1=7$ frequencies of exponentials are generated uniformly randomly over $[0,1)$, and
then the last frequency is added in the proximity of the 3rd frequency. In our 6 experiments, the 8th frequency is chosen such that the frequency separation between the 8th frequency
and the 3rd frequency is respectively $ 0.03, 0.01, 0.003, 0.001, 0.0003$, and $0.0001$.
Specifically, in our 6 experiments, the locations of the 8 frequencies are respectively
\{0.3923, 0.9988, 0.3437, 0.9086, 0.6977, 0.0298, 0.4813, 0.3743\},
\{0.3923, 0.9988, 0.3437, 0.9086, 0.6977, 0.0298, 0.4813, 0.3537\},
\{0.3923, 0.9988, 0.3437, 0.9086, 0.6977, 0.0298, 0.4813, 0.3467\},
\{0.3923, 0.9988, 0.3437, 0.9086, 0.6977, 0.0298, 0.4813, 0.3447\},
\{0.3923, 0.9988, 0.3437, 0.9086, 0.6977, 0.0298, 0.4813, 0.3440\},
 and
\{0.3923, 0.9988, 0.3437, 0.9086, 0.6977, 0.0298, 0.4813, 0.3438\}.
Hankel matrix recovery approach gives relative error $\frac{\|\hat{\bm{x}}-\bm{x}\|_2}{\|\bm{x}\|_2}=$ $3.899*10^{-9}$,
$5.3741*10^{-9}$,
$3.078*10^{-9}$,
$2.3399*10^{-9}$,
$9.8142*10^{-9}$, and
$8.1374*10^{-9}$,
 respectively.  With the recovered data $\hat{\bm{x}}$, we use the MUSIC algorithm to
identify the frequencies. The recovered frequencies for these 6 cases are respectively:
\{0.0298, 0.3437, 0.3737, 0.3923, 0.4813, 0.6977, 0.9086, 0.9988\},
\{0.0298, 0.3437, 0.3537, 0.3923, 0.4813, 0.6977, 0.9086, 0.9988\},
\{0.0298, 0.3437, 0.3467, 0.3923, 0.4813, 0.6977, 0.9086, 0.9988\},
\{0.0298, 0.3437, 0.3447, 0.3923, 0.4813, 0.6977, 0.9086, 0.9988\},
\{0.0298, 0.3437, 0.3440, 0.3923, 0.4813, 0.6977, 0.9086, 0.9988\},
and \{0.0298, 0.3437, 0.3438, 0.3923, 0.4813, 0.6977, 0.9086, 0.9988\}.  We illustrate these 6 cases in Figure \ref{Fig:NoiselessIdentification-3-100},
\ref{Fig:NoiselessIdentification-1-100} ,
\ref{Fig:NoiselessIdentification-3-1000} ,
\ref{Fig:NoiselessIdentification-1-1000} ,
\ref{Fig:NoiselessIdentification-3-10000} ,
\ref{Fig:NoiselessIdentification-1-10000},
respectively,  where the peaks of the imaging function $J(f)$ are the locations of the recovered frequencies. We can see the Hankel matrix recovery
successfully  recovers the missing data and correctly locate the frequencies.

\begin{figure}
\centering
\includegraphics[width=4in]{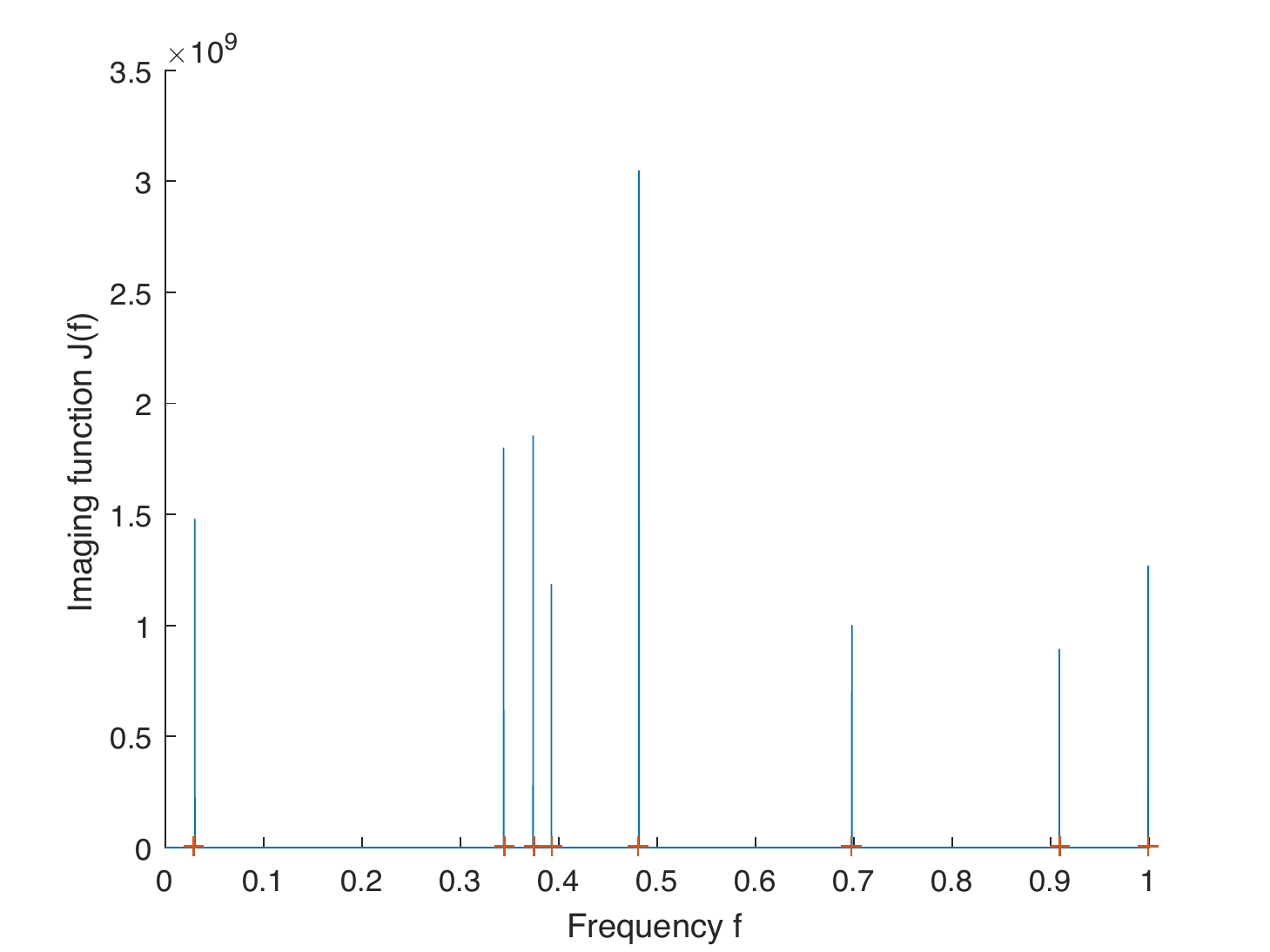}
\caption{Noiseless measurements, Hankel matrix recovery: frequency separation ${\rm dist}\mathcal{F}=0.03$}\label{Fig:NoiselessIdentification-3-100}
\end{figure}

\begin{figure}
\centering
\includegraphics[width=4in]{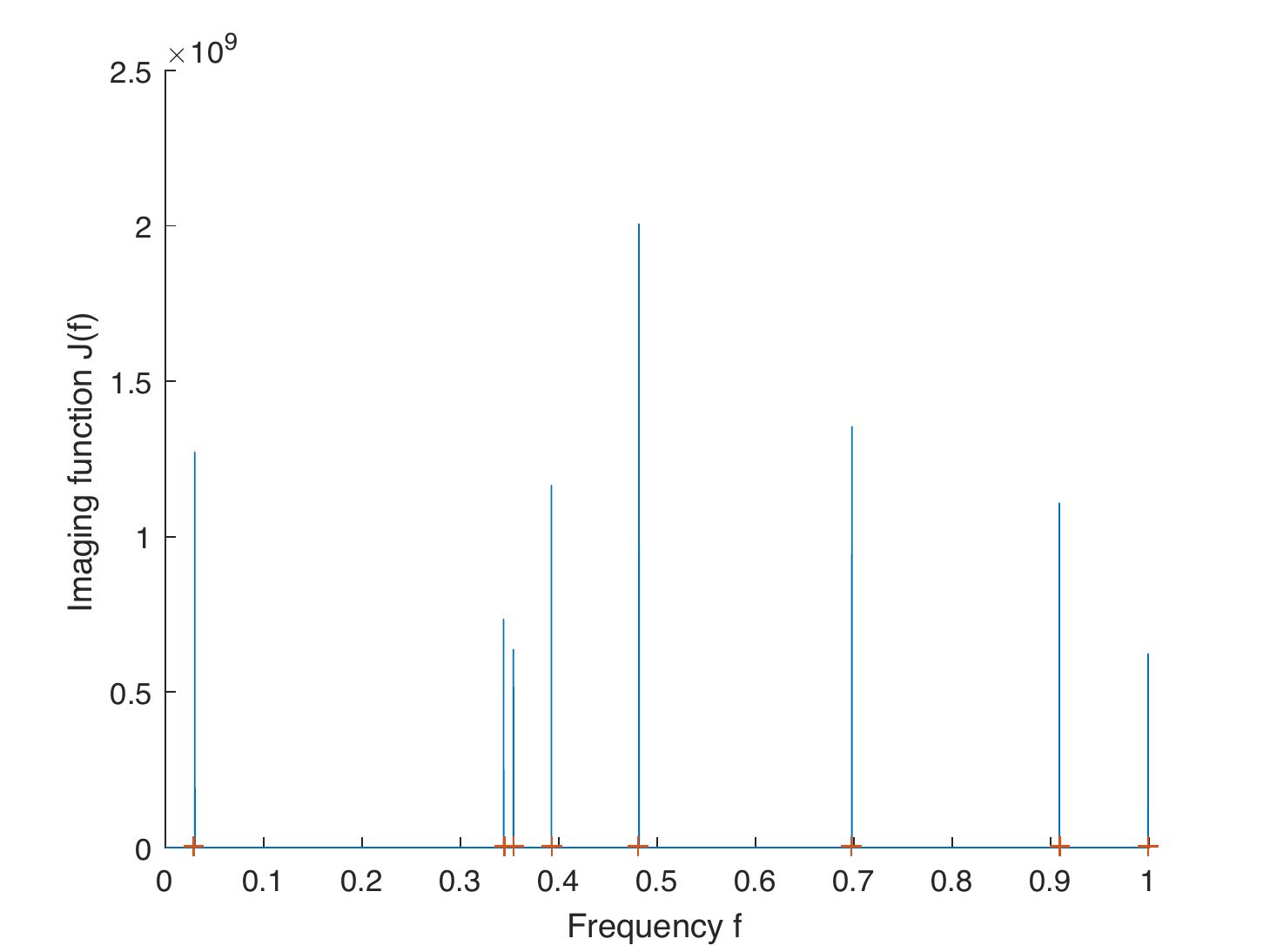}
\caption{Noiseless measurements, Hankel matrix recovery: frequency separation ${\rm dist}\mathcal{F}=0.01$}\label{Fig:NoiselessIdentification-1-100}
\end{figure}

\begin{figure}
\centering
\includegraphics[width=4in]{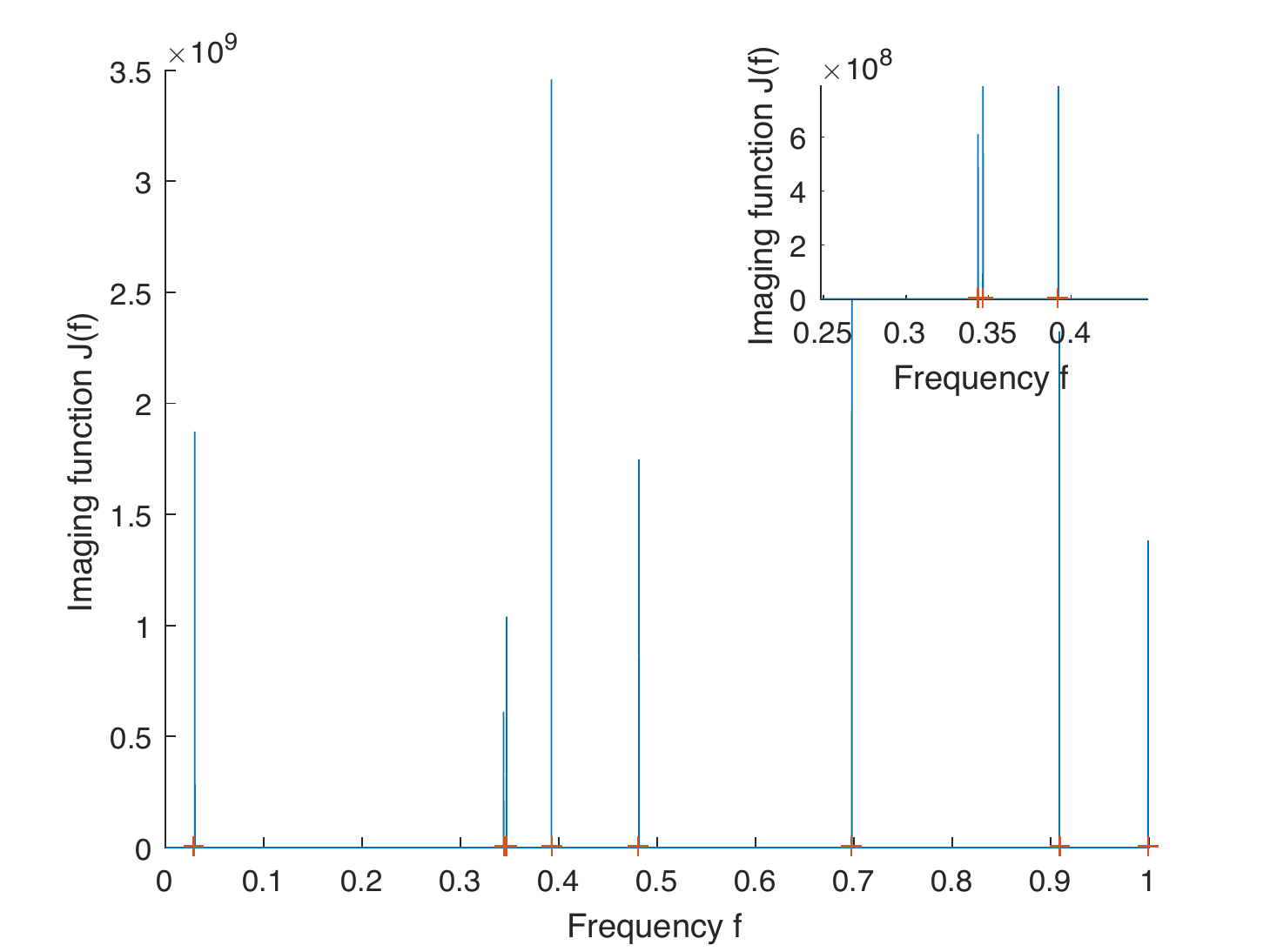}
\caption{Noiseless measurements, Hankel matrix recovery: frequency separation ${\rm dist}\mathcal{F}=0.003$}\label{Fig:NoiselessIdentification-3-1000}
\end{figure}

\begin{figure}
\centering
\includegraphics[width=4in]{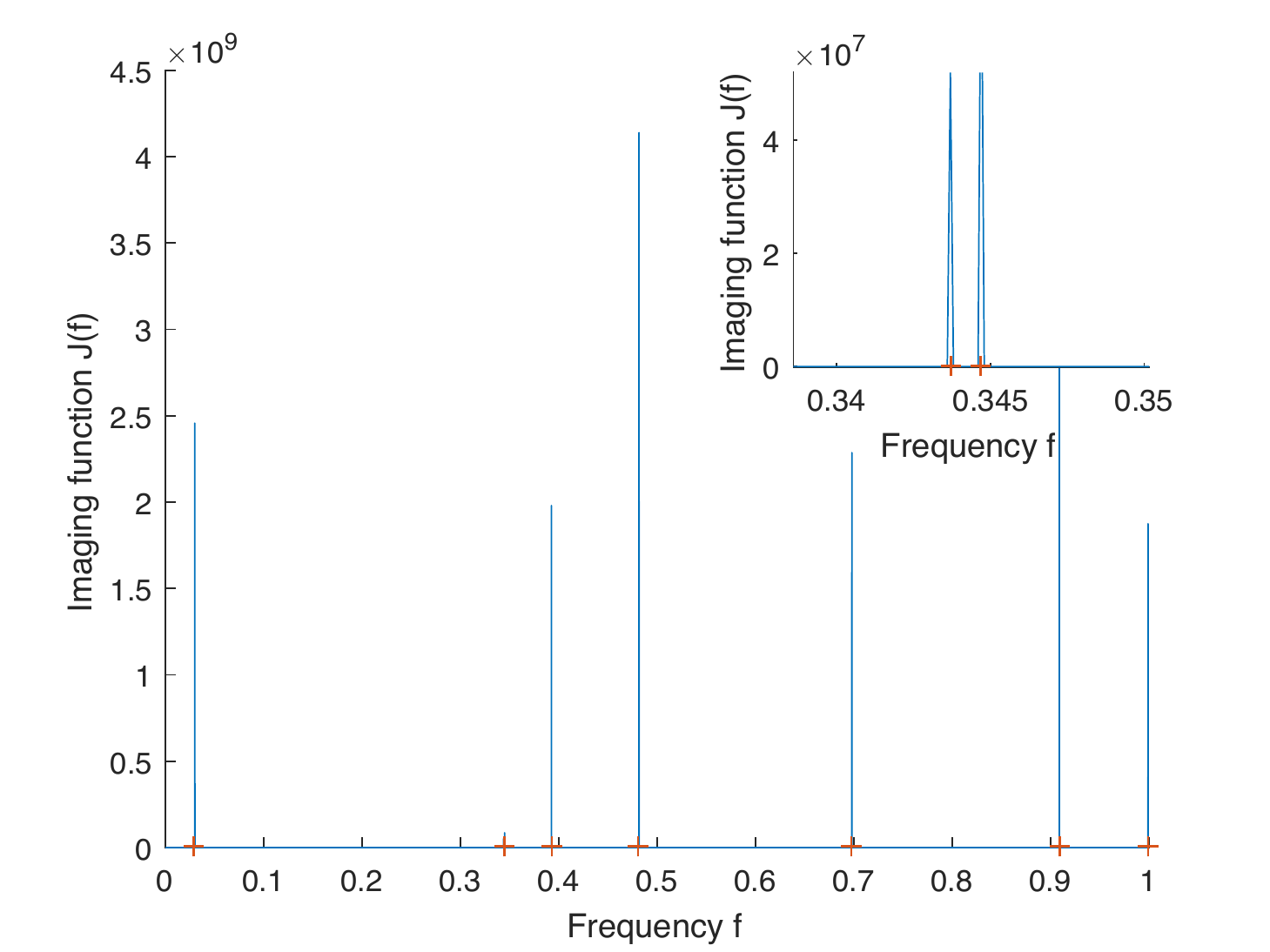}
\caption{Noiseless measurements, Hankel matrix recovery: frequency separation ${\rm dist}\mathcal{F}=0.001$}\label{Fig:NoiselessIdentification-1-1000}
\end{figure}

\begin{figure}
\centering
\includegraphics[width=4in]{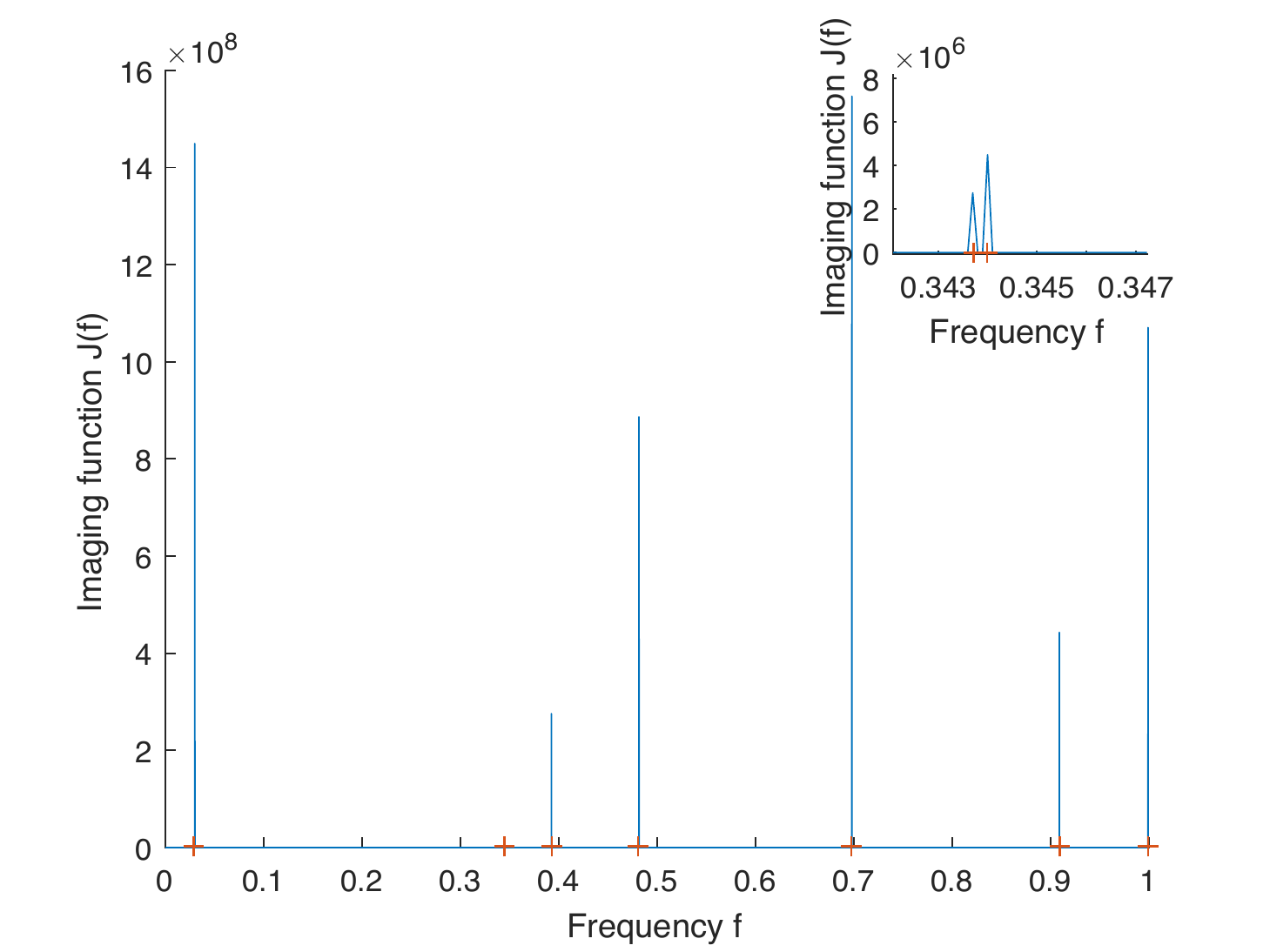}
\caption{Noiseless measurements, Hankel matrix recovery: frequency separation ${\rm dist}\mathcal{F}=0.0003$}\label{Fig:NoiselessIdentification-3-10000}
\end{figure}

\begin{figure}
\centering
\includegraphics[width=4in]{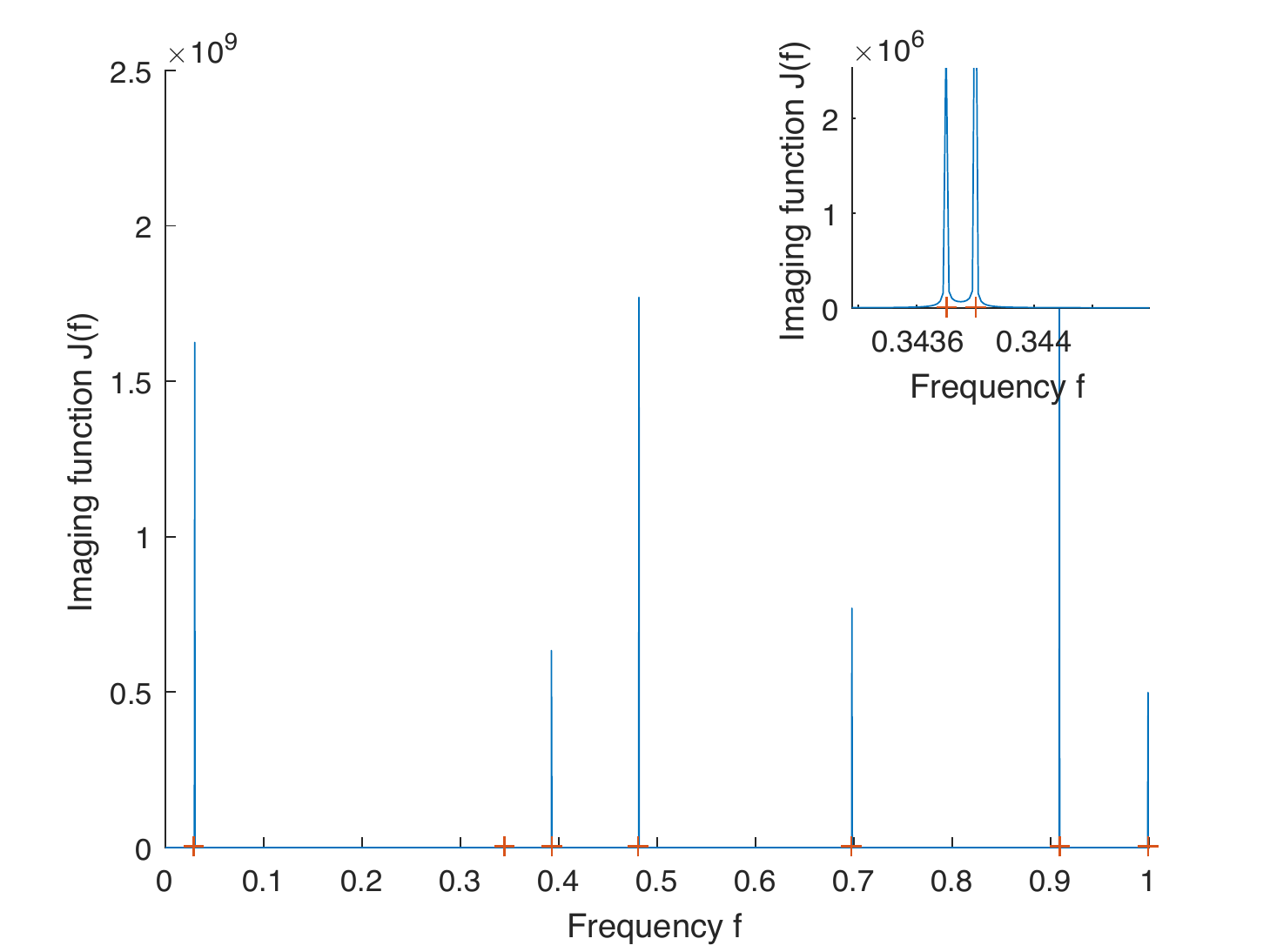}
\caption{Noiseless measurements, Hankel matrix recovery: frequency separation ${\rm dist}\mathcal{F}=0.0001$}\label{Fig:NoiselessIdentification-1-10000}
\end{figure}

We further demonstrate the performance of Hankel matrix recovery under noisy measurements. Again, we consider $N=64, |\mathcal{M}|=65$  and $R=8$.
The magnitude of the coefficients is obtained by $1+10^{0.5 p}$ where $p$ is uniformly randomly generated from $[0,1)$. The realized magnitudes
are 3.9891, 3.6159, 3.7868, 3.9261, 2.1606, 2.4933, 3.2741, and 3.0539 respectively in our experiment.
The phase of the coefficients are obtained by $2\pi s$ where $s$ is uniformly randomly generated from $[0,1)$.
In this example, the realized phases are 5.2378, 1.3855, 2.0064, 1.3784, 0.1762, 4.2739, 1.7979, and 0.1935 respectively.
The first $R-1=7$ frequencies of exponentials are generated uniformly randomly from $[0,1)$, and then the last frequency is added with frequency separation $5*10^{-3}$
from the third frequency. In this example, the ground truth frequencies are 0.8822, 0.0018, 0.6802, 0.2825, 0.8214, 0.2941, 0.3901, and 0.6852.

We generate the noise vector $\bm{v}\in\mathbb{C}^{2N-1}$ as $s_1+\imath s_2$ where each element of $s_1\in\mathbb{R}^{2N-1}$ and $s_2\in\mathbb{R}^{2N-1}$ is independently generated from the zero mean standard Gaussian distribution. And we further normalize $\bm{v}$ such that $\|\bm{v}\|_2=0.1$. In this noisy case, we solve the problem
\begin{align}\label{Defn:NoisyNuclearNormMinimization}
& \min_{{\bm{x}}} \|\bm{H}(\bm{{x}})\|_* \nonumber\\
& {\rm subject to\ }\mathcal{A}(\bm{{x}})=\bm{b},
\end{align}
to get the recovered signal $\hat{\bm{x}}$, and a relative error $\frac{\|\hat{\bm{x}}-\bm{x}^{}\|_2}{\|\bm{x}^{}\|_2}=1.2*10^{-3}$ is achieved, and the location of recovered
frequencies are 0.0018, 0.2825, 0.2941, 0.3901, 0.6802, 0.6852, 0.8214, and 0.8822. We illustrate the locations of the recovered frequencies in
Figure \ref{Fig:NoisyIdentification}. We can see that the Hankel matrix recovery can also provide robust data and frequency recovery under noisy measurements.

\begin{figure}
\centering
\includegraphics[width=4in]{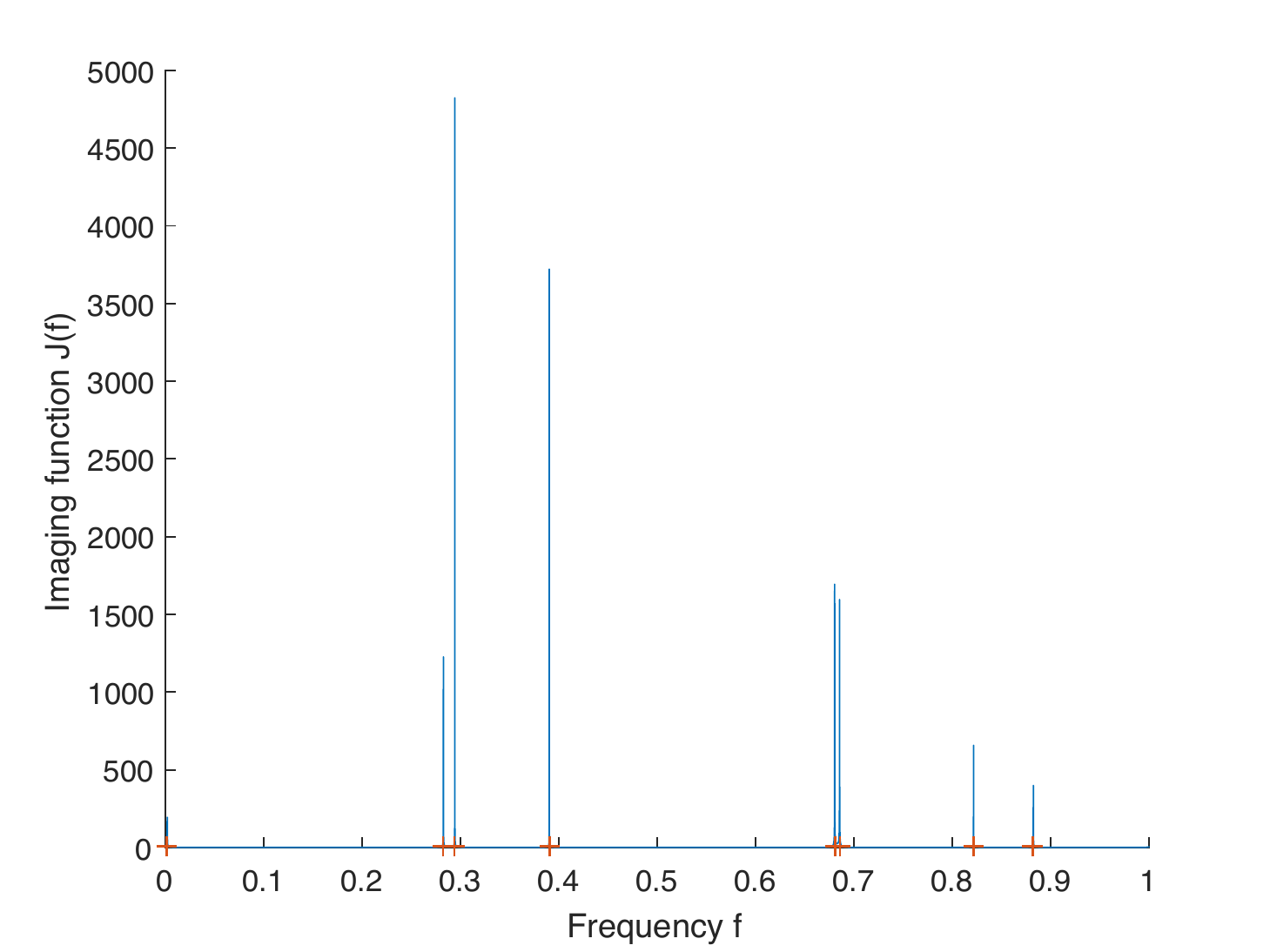}
\caption{Noisy measurements, Hankel matrix recovery}\label{Fig:NoisyIdentification}
\end{figure}

\section{Conclusions and future works}
\label{sec:OpenProblem}
In this paper, we have shown, theoretically and numerically, that the Hankel matrix recovery can be robust to frequency separations in super-resolving the superposition of complex exponentials.
By comparison, the TV minimization and atomic norm minimization require the underlying frequencies to be well-separated to super-resolve the superposition of complex exponentials even when
the measurements are noiseless.  In particular, we show that Hankel matrix recovery approach can super-resolve
the $R$ frequencies, regardless of how close the frequencies are to each other, from compressed non-uniform measurements. We presented a new concept of orthonormal atomic norm minimization (OANM), and showed that this concept helps
us understand the success of Hankel matrix recovery in separation-free super-resolution. We further show that,  in traditional atomic norm minimization,  the underlying parameters \emph{must} be well separated so
that the signal can be successfully recovered if the atoms are changing continuously with respect to the continuously-valued parameters; however, for OANM,
it is possible the atomic norm minimization are successful even though the original atoms can be arbitrarily close.
As a byproduct of this research, we also provide one matrix-theoretic inequality of nuclear norm, and give its proof from the theory of compressed sensing. In future works, it would be interesting to extend the results
in this paper to super-resolving the superposition of complex exponentials with higher dimensional frequency parameters \cite{xu_precise_2014,chen_robust_2014,2DChiChen}

\section{Appendix}\label{secProof}

%\begin{lem}\label{NuclearNorm4MatrixDiffandNuclearNormforSingularValueDiff}
%( \cite{horn_matrix_2012,bhatia_matrix_2013})  For arbitrary matrices $\bm{X}, \bm{Y}$, and $\bm{Z}=\bm{X}-\bm{Y}\in \mathbb{C}^{N\times N}$, let $s_{i}(\bm{X},\bm{Y})$ be the distance between the $i$-th singular value of $\bm{X}$ and $\bm{Y}$, explicitly
%\begin{align}
%s_{i}(\bm{X},\bm{Y})=|\sigma_{i}(\bm{X})-\sigma_{i}(\bm{Y})|,i=1,2,\cdots,k.
%\end{align}
%
%Let $s_{[i]}(\bm{X},\bm{Y})$ be the $i$-th largest value of sequence $s_{1}(\bm{X},\bm{Y}), s_2(\bm{X},\bm{Y}), \cdots, s_m(\bm{X},\bm{Y})$, then
%\begin{align}
%\sum_{i=1}^{k} s_{i}(\bm{X},\bm{Y}) \leq \|\bm{Z}\|_{k}, \forall k=1,2, \cdots,m,
%\end{align}
%where $\|\bm{Z}\|_k$ is defined as
%%\begin{align}
%$\sum_{i=1}^k \sigma_i(\bm{Z})$.
%%\end{align}
%\end{lem}

%\begin{lem}\label{InnerProductRelation}
%(\cite{mirsky_trace_1975}) For arbitrary matrices $\bm{X},\bm{Y}\in \mathbb{C}^{N\times N}$,
%\begin{align}
%{\rm Tr}(\bm{X}^{T}\bm{Y})\leq \sum_{i=1}^{N} \sigma_{i}(\bm{X}) \sigma_{i}(\bm{Y}).
%\end{align}
%\end{lem}

\subsection{Proof of Lemma \ref{lem:WeakNullSpaceCondition}}
\label{ProofofLem:WeakNullSpaceCondition}
Any solution to the nuclear norm minimization must be $\bm{X}_0+\bm{Q}$, where $\bm{Q}$ is from the null space of $\mathcal{A}$. Suppose that the singular value decomposition of ${\bm{X}}_{0}$ is given by
$$
{\bm{X}}_{0}={\bm{U}}{\Lambda} {\bm{V}}^*
$$
where ${\bm{U}} \in \mathbb{C}^{M \times R}$, $\bm{\Lambda} \in \mathbb{C}^{R \times R}$, and ${\bm{V}}\in \mathbb{C}^{N \times R}$.

From Lemma \ref{lemma:complexdifferential}, we know the subdifferential of $\|\cdot\|_{*}$ at the point $\bm{X}_0$ is given by
$$
\{
{\bm{Z}}~|
{\bm{Z}}={\bm{U}}{\bm{V}}^{*}+\bar{{\bm{U}}}\bm{M} {\bar{{\bm{V}}}}^{*}, \mbox{~where~} \|\bm{M}\|_2\leq 1, {\bm{U}}^* {\bar{{\bm{U}}}}={0},{\bar{{\bm{U}}}}^* {\bar{{\bm{U}}}}={I}, {\bm{V}}^* {\bar{{\bm{V}}}}={0},{\bar{{\bm{V}}}}^* {\bar{{\bm{V}}}}={I}\}.
$$
Then from the property of subdifferential of a convex function, for any $\bm{Z}=\bm{U}\bm{V}+\bar{\bm{U}}\bm{M} \bar{\bm{V}}^*$ (with $\|\bm{M}\|_2\leq 1$) from the subdifferetial of $\|\cdot \|_{*}$ at the point $\bm{X}_0$, we have
\begin{align}
&\|\bm{X}_0+\bm{Q}\|_{*}\\
&\geq \|\bm{X}_0\|_{*}+\langle \bm{Z}, \bm{Q} \rangle\\
&=\|\bm{X}_0\|_{*}+\text{Re}\left({\rm Tr}\left(\bm{U}^*\bm{Q} \bm{V}\right)\right)+\text{Re}\left({\rm Tr}\left(\bar{\bm{U}}^*\bm{Q} \bar{\bm{V}}\bm{M}^*\right)\right)\\
&\geq \|\bm{X}_0\|_{*}-|{\rm Tr}\left(\bm{U}^*\bm{Q} \bm{V}\right)|+\text{Re}\left({\rm Tr}\left(\bar{\bm{U}}^*\bm{Q} \bar{\bm{V}}\bm{M}^*\right)\right).
\end{align}
Because we can take any $\bm{M}$ with $\|\bm{M}\|_2\leq 1$, when $\bm{Q} \neq 0$, we have
\begin{align}
\|\bm{X}_0+\bm{Q}\|_{*} &\geq \|\bm{X}_0\|_{*}-|{\rm Tr}\left(\bm{U}^*\bm{Q} \bm{V}\right)|+ \|\bar{\bm{U}}^*\bm{Q} \bar{\bm{V}}\|_{*}\\
&>\|\bm{X}_0\|_{*},
\end{align}
where the last step is because the dual norm of the spectral norm is the nuclear norm (which also holds for complex-numbered matrices). Thus $\bm{X}_{0}$ is the unique solution to the nuclear norm minimization.
%
%
%
%Define the descent cone of the nuclear norm of the Hankel matrix as
%$$
%\mathcal{D}(\bm{x})=Cone(\{{w}~|~\|Hankel(\bm{x}+{w})\|_{*}\leq \|Hankel(\bm{x})\|_{*}\}),
%$$
%where $Cone(\cdot)$ is the conic hull of a set.
%
%Then the Hankel matrix completion succeeds if and only if for any nonzero vector ${w}$ in the null space of $\mathcal{\bm{A}}$, ${w}$ is not in the descent cone of the nuclear norm. Namely,
%$$
%\mathcal{N}(\mathcal{\bm{A}}) \cap \mathcal{D}=\{{0}\}.
%$$
%
%Equivalently, for any vector ${w}$ in the null space of $\mathcal{\bm{A}}$,
%$$
%Hankel({w}) \cap \mathcal{D}({\bm{X}})={0},
%$$
%where $$
%\mathcal{D}({\bm{X}})=Cone(\{{W}~|~\|{\bm{X}}+{W}\|_{*}\leq \|{\bm{X}}\|_{*}\}).
%$$
%
%Suppose that the singular value decomposition of an $m \times n$ matrix ${\bm{X}}$ is given by
%$$
%{\bm{X}}={\bm{U}}{\Lambda} {\bm{V}}^*
%$$
%where ${\bm{U}} \in R^{m \times r}$, ${\Lambda} \in R^{r \times r}$, and ${\bm{V}}\in R^{n \times r}$
%
%
%The descent cone of the nuclear norm $f({\bm{X}})=\|{\bm{X}}\|_{*}$ is the polar cone of the subdifferential cone of $f({\bm{X}})=\|{\bm{X}}\|_{*}$.
%The subdifferential of $f({\bm{X}})=\|{\bm{X}}\|_{*}$ at ${\bm{X}}={\bm{X}}_{0}={\bm{U}}_{0}{\Lambda}_{0} {\bm{V}}_{0}^*$ is given by
%
%$\star$ $\bar{\bm{U}}$ is the matrix consisting with columns being orthogonal complement basis of the orthogonal complement subspace of $\bm{U}$

\subsection{Strict inequality is not always necessary in the null space condition for successful recovery via nuclear norm minimization}
\label{OnlySufficient}
In this subsection, we show that the null space condition for Lemma \ref{lem:WeakNullSpaceCondition} is not a necessary condition for the success of nuclear norm minimization, which is in contrast to the null space condition
in Lemma 13 of \cite{oymak_new_2010}. Specifically, we show the following claim:
\begin{claim}\label{lem:notnecessary}
Let $\bm{X}_0$ be any $M \times N$ matrix of rank $R$ in $\mathbb{C}^{M \times N}$, and we observe it through a linear mapping $\mathcal{A}(\bm{X}_0)=\bm{b}$. We also assume that $\bm{X}_0$ has a singular value decomposition (SVD)
 $\bm{X}_0= \bm{U} \bm{\Sigma} \bm{V}^*$, where $\bm{U} \in \mathbb{C}^{M \times R}$, $\bm{V}\in \mathbb{C}^{N \times R}$, and $\bm{\Sigma} \in \mathbb{C}^{R \times R}$ is a diagonal matrix.
Consider the nuclear norm minimization (\ref{eq:newNNM3})
\begin{equation}\label{eq:newNNM3}
\min_{\bm{X}}\|\bm{X}\|_*,\qquad\mbox{subject to}\quad \mathcal{A}(\bm{X})=\mathcal{A}(\bm{X}_0),
\end{equation}
Then for the nuclear norm minimization to correctly and uniquely recovers $\bm{X}_0$, it is not always necessary
that ``for every nonzero element $\bm{Q} \in \mathcal{N}(\mathcal{A})$,
\begin{align}
-|{\rm Tr}(\bm{U}^*\bm{Q}\bm{V})|+\|\bar{\bm{U}}^*\bm{Q}\bar{\bm{V}}\|_*> 0,
\end{align}
\noindent where $\bar{\bm{U}}$ and $\bar{\bm{V}}$ are such that $[\bm{U}\ \bar{\bm{U}}]$ and $[\bm{V}\ \bar{\bm{V}}]$ are unitary.
''
\end{claim}

  For simple presentation, we first use an example in the field of real numbers (where every element in the null space is a real-numbered matrix
  ) to illustrate the calculations, and prove this claim. Building on this real-numbered example, we further give an example in the field of complex numbers to prove this claim.

Suppose \begin{align}
\bm{X}_0=
\begin{array}{l}
\left[\begin{array}{*{20}{c}}
-1 &0\\
0 &0
\end{array}\right]
\end{array},~
\bm{Q}=
\begin{array}{l}
\left[\begin{array}{*{20}{c}}
1 &1\\
1 &1
\end{array}\right].
\end{array}
\end{align}
We also assume that the linear mapping $\mathcal{A}$ is such that $\bm{Q}$ is the only nonzero element in the null space of $\mathcal{A}$.
Then the solution to (\ref{eq:newNNM3}) must be of the form $\bm{A}+t\bm{Q}$, where $t$ is any real number.

\begin{align}
\bm{U}=
\begin{array}{l}
\left[\begin{array}{*{20}{c}}
1\\ 0
\end{array}\right]
\end{array}
,
\bm{V}=
\begin{array}{l}
\left[\begin{array}{*{20}{c}}
-1\\ 0
\end{array}\right]
\end{array}
,
\bar{\bm{U}}=
\begin{array}{l}
\left[\begin{array}{*{20}{c}}
0\\ 1
\end{array}\right]
\end{array}
,
\bar{\bm{V}}=
\begin{array}{l}
\left[\begin{array}{*{20}{c}}
0\\ 1
\end{array}\right]
\end{array}.
\end{align}
One can check that, for this example,
\begin{align}
-|{\rm Tr}(\bm{U}^*\bm{Q}\bm{V})|+\|\bar{\bm{U}}^*\bm{Q}\bar{\bm{V}}\|_*= 0.
\end{align}
However, we will show that
\begin{align}
\left\vert \left\vert \bm{X}_0+t\bm{Q} \right\vert\right\vert_*>1, \forall t\neq0,
\end{align}
implying that $\bm{X}_0$ is the unique solution to (\ref{eq:newNNM3}).
In fact, we calculate
\begin{align}
\bm{B}
=(\bm{X}_0+t\bm{Q})(\bm{X}_0+t\bm{Q})^T
=
\begin{array}{l}
\left[\begin{array}{*{20}{c}}
(-1+t)^2+t^2 & (-1+t)t+t^2\\
(-1+t)t+t^2 & 2t^2
\end{array}\right]
\end{array},
\end{align}
and then the singular values of $\bm{X}_0+t\bm{Q}$ are the square roots of the eigenvalues of $\bm{B}$. The eigenvalues of $\bm{B}$ can be obtained by solving for $\lambda$ using
\begin{align}
{\rm det}(\bm{B}-\lambda \bm{I})=0.
\end{align}
This results in
\begin{align}
\lambda=\frac{a(t)+b(t)\pm\sqrt{(a(t)-b(t))^2+4c(t)}}{2},
\end{align}
where
\begin{align}
a(t)=(-1+t)^2+t^2,
b(t)=2t^2,
c(t)=((-1+t)t+t^2)^2.
\end{align}
Thus the two eigenvalues of $\bm{B}$ are
\begin{align}
&\lambda_1=\frac{4t^2-2t+1+\sqrt{16t^4-16t^3+8t^2-4t+1}}{2},\\
&\lambda_2=\frac{4t^2-2t+1-\sqrt{16t^4-16t^3+8t^2-4t+1}}{2},
\end{align}
and the singular values of $\bm{X}_0+t\bm{Q}$ are
\begin{align}
&\sigma_1=\sqrt{\lambda_1}=
\sqrt{\frac{4t^2-2t+1+\sqrt{16t^4-16t^3+8t^2-4t+1}}{2}},\\
&\sigma_2=\sqrt{\lambda_2}=
\sqrt{\frac{4t^2-2t+1-\sqrt{16t^4-16t^3+8t^2-4t+1}}{2}}.
\end{align}
After some algebra, we get
\begin{align}
\|\bm{X}_0+t\bm{Q}\|_*
=\sigma_1+\sigma_2
&=
\begin{cases}
\sqrt{4t^2+1}, t\geq0,\\
1-2t, t<0.
\end{cases}
\end{align}
This means $\|\bm{X}_0+t\bm{Q}\|_*$  is always greater than $1$ for $t\neq0$, showing $\bm{X}_0$ is the unique solution to the nuclear norm minimization. But $-|{\rm Tr}(\bm{U}^*\bm{Q}\bm{V})|+\|\bar{\bm{U}}^*\bm{Q}\bar{\bm{V}}\|_*\ngtr 0.$

Now we give an example in the field of complex examples, where the null space of $\mathcal{A}$ contains complex-numbered matrices. Suppose that we have the same
matrices $\bm{X}_0$ and $\bm{Q}$. Then the solution to (\ref{eq:newNNM3}) must be of the form $\bm{A}+t\bm{Q}$, where $t$ is any complex number. Without loss of generality, let us
take $t=-ae^{-\imath\theta}$, where $a\geq 0$ is a nonnegative real number, and $\theta$ is any real number between 0 and $2\pi$. We further denote $
\bm{B}
=(\bm{X}_0+t\bm{Q})(\bm{X}_0+t\bm{Q})^*$. Then by calculating the eigenvalues of $\bm{B}$, we obtain that
\begin{align}
&\|\bm{X}_0+t\bm{Q}\|_*
=\sigma_1+\sigma_2=\sqrt{4a^2+2a(1+\cos(\theta))+1}\\
&=\sqrt{4 \left (a+\frac{1+\cos(\theta)}{4}\right)^2+1-\frac{(1+\cos(\theta))^2}{4}}
\end{align}
where $t=-ae^{-\imath\theta}$ with $a\geq 0$ and $\theta\in [0, 2\pi)$. So $\|\bm{X}_0+t\bm{Q}\|_*>1$, if $a\neq0$ (namely $t\neq0$), implying that the nuclear norm minimization
can uniquely recovers $\bm{X}_{0}$ even though  $-|{\rm Tr}(\bm{U}^*\bm{Q}\bm{V})|+\|\bar{\bm{U}}^*\bm{Q}\bar{\bm{V}}\|_*\ngtr 0.$

\subsection{Proof of Lemma \ref{lemma:complexdifferential}}
\label{proofofcomplexsubdifferential}
\begin{proof}
We write
\begin{equation}\label{eq:SVDcomplexGy}
\bm{U} = \bm{\Theta}_1+\imath\bm{\Theta}_2,~~\bm{V}=\bm{\Xi}_1+\imath\bm{\Xi}_2,
\end{equation}
where $\bm{\Theta}_1\in\mathbb{R}^{M\times R},\bm{\Theta}_2\in\mathbb{R}^{M\times R},\bm{\Xi}_1\in\mathbb{R}^{N\times R}$, and $\bm{\Xi}_2\in\mathbb{R}^{N\times R}$. Then, by direct calculation,
\begin{equation}\label{eq:ThetaXi}
\bm{\Theta}\equiv\left[\begin{matrix}
\bm{\Theta}_1 & -\bm{\Theta}_2\cr \bm{\Theta}_2 & \bm{\Theta}_1
\end{matrix}\right]\in\mathbb{R}^{2M\times 2R}, \qquad
\bm{\Xi}\equiv\left[\begin{matrix}
\bm{\Xi}_1 & -\bm{\Xi}_2\cr \bm{\Xi}_2 & \bm{\Xi}_1
\end{matrix}\right]\in\mathbb{R}^{2N\times 2R}
\end{equation}
satisfy $\bm{\Theta}^T\bm{\Theta}=\bm{\Xi}^T\bm{\Xi}=\bm{I}$. Moreover, if we define
$\hat{\bm{\Omega}}=\left[\begin{matrix}
{\rm Re}(\bm{X}) & -{\rm Im}(\bm{X})\cr {\rm Im}(\bm{X}) & {\rm Re}(\bm{X})
\end{matrix}\right]$, then
\begin{equation}\label{eq:SVDrealOmega}
\hat{\bm{\Omega}}=
\bm{\Theta}\left[\begin{matrix}
\bm{\Sigma} & \cr & \bm{\Sigma}
\end{matrix}\right]
\bm{\Xi}^T
\end{equation}
is a singular value decomposition of the real-numbered matrix $\hat{\bm{\Omega}}$, and the singular values $\hat{\bm{\Omega}}$ are those of $\bm{X}$, each repeated twice. Therefore,
\begin{equation}\label{eq:Freal}
\mathcal{F}\left(\left[\begin{matrix}{\rm Re}(\bm{X})\cr {\rm Im}(\bm{X})\end{matrix}\right]\right)
=\|\bm{\Sigma}\|_*=\frac12\|\hat{\bm{\Omega}}\|_*.
\end{equation}
Define a linear operator $\mathcal{E}~:~\mathbb{R}^{2M \times N}\mapsto\mathbb{R}^{2M\times 2N}$ by
$$
\mathcal{E}\left(\left[\begin{matrix}\bm{\alpha}\cr\bm{\beta}\end{matrix}\right]\right)
=\left[\begin{matrix}\bm{\alpha}&-\bm{\beta}\cr
\bm{\beta}&\bm{\alpha}\end{matrix}\right],
\quad\mbox{with}\quad\bm{\alpha},\bm{\beta}\in\mathbb{R}^{M \times N}.
$$
By \eqref{eq:Freal} and the definition of $\hat{\bm{\Omega}}$, we obtain $$\mathcal{F}\left(\left[\begin{matrix}{\rm Re}(\bm{X})\cr {\rm Im}(\bm{X})\end{matrix}\right]\right)=
\frac12\left|\left|\mathcal{E}  \left(\left[\begin{matrix}{\rm Re}(\bm{X})\cr {\rm Im}(\bm{X})\end{matrix}\right]\right)               \right|\right|_*$$.

From convex analysis and $\hat{\Omega}=\mathcal{E} \left(\left[\begin{matrix}{\rm Re}(\bm{X})\cr {\rm Im}(\bm{X})\end{matrix}\right]\right)  $, the subdifferential of $\mathcal{F}$ is given by
\begin{equation}\label{eq:subdiffF}
\partial\mathcal{F}\left(\left[\begin{matrix}{\rm Re}(\bm{X})\cr {\rm Im}(\bm{X})\end{matrix}\right]\right)=\frac12\mathcal{E}^*\partial\left|\left|\hat{\bm{\Omega}}\right|\right|_*,
\end{equation}
where $\mathcal{E}^*$ is the adjoint of the linear operator $\mathcal{E}$.

On the one hand, the adjoint $\mathcal{E}^*$ is given by, for any
$\bm{\Delta}=\left[\begin{matrix}\bm{\Delta}_{11}&\bm{\Delta}_{12}\cr\bm{\Delta}_{21}&\bm{\Delta}_{22}\end{matrix}\right]\in\mathbb{R}^{2M\times 2N}$
with each block in $\mathbb{R}^{M\times N}$,
\begin{equation}\label{eq:adjE}
\mathcal{E}^*\bm{\Delta}=\left[\begin{matrix}\bm{\Delta}_{11}+\bm{\Delta}_{22}\cr \bm{\Delta}_{21}-\bm{\Delta}_{12}\end{matrix}\right].
\end{equation}
On the other hand, since \eqref{eq:SVDrealOmega} provides a singular value decomposition of $\hat{\bm{\Omega}}$,
\begin{equation}\label{eq:subdiffnuc}
\partial\|\hat{\bm{\Omega}}\|_*=\left\{\bm{\Theta}\bm{\Xi}^T+\bm{\Delta}~|~
\bm{\Theta}^T\bm{\Delta}=\bm{0},~\bm{\Delta}\bm{\Xi}=\bm{0},~\|\bm{\Delta}\|_2\leq1\right\}.
\end{equation}
Combining \eqref{eq:subdiffF}, \eqref{eq:adjE}, \eqref{eq:subdiffnuc}, and  \eqref{eq:ThetaXi} yields the subdifferential of $\mathcal{F}(\cdot)$ at
$  \left[\begin{matrix}{\rm Re}(\bm{X})\cr {\rm Im}(\bm{X})\end{matrix}\right]$:
\begin{align}
\label{sub_real_expression}
&~~\partial\mathcal{F}\left(\left[\begin{matrix}{\rm Re}(\bm{X})\cr {\rm Im}(\bm{X})\end{matrix}\right]\right)\\
&=\left\{\left[\begin{matrix}
\left(\bm{\Theta}_1\bm{\Xi}_1^T+\bm{\Theta}_2\bm{\Xi}_2^T+\frac{\bm{\Delta}_{11}+\bm{\Delta}_{22}}{2}\right)\cr
\left(\bm{\Theta}_2\bm{\Xi}_1^T-\bm{\Theta}_1\bm{\Xi}_2^T+\frac{\bm{\Delta}_{21}-\bm{\Delta}_{12}}{2}\right)
\end{matrix}\right]~\Big|~
\bm{\Delta}=\left[\begin{matrix}\bm{\Delta}_{11}&\bm{\Delta}_{12}\cr\bm{\Delta}_{21}&\bm{\Delta}_{22}\end{matrix}\right],~\bm{\Theta}^T\bm{\Delta}=\bm{0},~\bm{\Delta}\bm{\Xi}=\bm{0},~\|\bm{\Delta}\|_2\leq1
\right\}.
\end{align}

We are now ready to show (\ref{eq:GsubsetsubdiffF}).

Firstly, we show that any element in $\mathfrak{H}\equiv \left\{\left[\begin{matrix}\bm{\alpha}\cr\bm{\beta}\end{matrix}\right]~\Big|~
\bm{\alpha}+\imath\bm{\beta}\in\mathfrak{S}\right\}$ must also be in $\partial \mathcal{F}\left(\left[\begin{matrix}{\rm Re}(\bm{X})\cr {\rm Im}(\bm{X})\end{matrix}\right]\right)$, namely (\ref{sub_real_expression}).
In fact, for any $\bm{W}=\bm{\Delta}_1+\imath\bm{\Delta}_2$ satisfying $\bm{U}^*\bm{W}=0,\bm{W}\bm{V}=0$ and $\|\bm{W}\|_2\leq 1$, we choose $\bm{\Delta}=\left[\begin{matrix}\bm{\Delta}_1&-\bm{\Delta}_2\cr\bm{\Delta}_2&\bm{\Delta}_1\end{matrix}\right]$.
This choice of $\bm{\Delta}$ satisfies the constraints on $\bm{\Delta}$ in (\ref{sub_real_expression}).
Furthermore, $\bm{U}\bm{V}^*+\bm{W}=(\bm{\Theta}_1\bm{\Xi}_1^T+\bm{\Theta}_2\bm{\Xi}_2^T+\bm{\Delta}_1) +\imath(\bm{\Theta}_2\bm{\Xi}_1^T-\bm{\Theta}_1\bm{\Xi}_2^T+\bm{\Delta}_2)$.
Thus
\begin{equation}\label{setforward}
\mathfrak{H} \subseteq \partial \mathcal{F}\left(\left[\begin{matrix}{\rm Re}(\bm{X})\cr {\rm Im}(\bm{X})\end{matrix}\right]\right).
\end{equation}
Secondly, we show that
\begin{equation} \label{setreverse}
\partial \mathcal{F}\left(\left[\begin{matrix}{\rm Re}(\bm{X})\cr {\rm Im}(\bm{X})\end{matrix}\right]\right) \subseteq \mathfrak{H}.
\end{equation}
We let $\bm{\Delta}=\left[\begin{matrix}\bm{\Delta}_{11}&\bm{\Delta}_{12}\cr\bm{\Delta}_{21}&\bm{\Delta}_{22}\end{matrix}\right]$ be any matrix satisfying the the constraints on $\bm{\Delta}$ in (\ref{sub_real_expression}).
We claim that $\bm{W} \doteq   \frac{\bm{\Delta}_{11}+\bm{\Delta}_{22}}{2} + \imath \frac{\bm{\Delta}_{21}-\bm{\Delta}_{12}}{2}$ satisfies $\bm{U}^*\bm{W}=0,\bm{W}\bm{V}=0$ and $\|\bm{W}\|_2\leq 1$.

In fact, from $\bm{\Theta}^T  \left[\begin{matrix}\bm{\Delta}_{11}&\bm{\Delta}_{12}\cr\bm{\Delta}_{21}&\bm{\Delta}_{22}\end{matrix}\right]  =\bm{0} $, we have
\begin{align}
&+\bm{\Theta}_1^T \bm{\Delta}_{11}+\bm{\Theta}_2^T\bm{\Delta}_{21}=\bm{0}\label{eq1}\\
&+\bm{\Theta}_1^T \bm{\Delta}_{12}+\bm{\Theta}_2^T\bm{\Delta}_{22}=\bm{0}\label{eq2}\\
&-\bm{\Theta}_2^T \bm{\Delta}_{11}+\bm{\Theta}_1^T\bm{\Delta}_{21}=\bm{0}\label{eq3}\\
&-\bm{\Theta}_2^T \bm{\Delta}_{12}+\bm{\Theta}_1^T\bm{\Delta}_{22}=\bm{0} \label{eq4}
\end{align}

Thus we obtain
\begin{align}
\bm{U}^*\bm{W}&= (\bm{\Theta}_1^T-\imath \bm{\Theta}_2^T)(\frac{\bm{\Delta}_{11}+\bm{\Delta}_{22}}{2} + \imath \frac{\bm{\Delta}_{21}-\bm{\Delta}_{12}}{2})\\
&=\bm{\Theta}_1^T \frac{\bm{\Delta}_{11}+\bm{\Delta}_{22}}{2}+\bm{\Theta}_2^T \frac{\bm{\Delta}_{21}-\bm{\Delta}_{12}}{2} +\imath\left( \bm{\Theta}_1^T \frac{\bm{\Delta}_{21}-\bm{\Delta}_{12}}{2} -\bm{\Theta}_2^T\frac{\bm{\Delta}_{11}+\bm{\Delta}_{22}}{2}
\right)\\
&=\bm{0}+ \imath \bm{0}=\bm{0},
\end{align}
where the last two equalities come from adding up (\ref{eq1}) and (\ref{eq4}), and subtracting  (\ref{eq2}) from (\ref{eq3}),  respectively.

Similarly from $\left[\begin{matrix}\bm{\Delta}_{11}&\bm{\Delta}_{12}\cr\bm{\Delta}_{21}&\bm{\Delta}_{22}\end{matrix}\right] \bm{\Xi}=\bm{0}$, we can verify that
$$\bm{W}\bm{V}=0.$$
Moreover,
\begin{align*}
\|\bm{W}\|_2&=\left|\left| \left[\begin{matrix}  \frac{\bm{\Delta}_{11}+\bm{\Delta}_{22}}{2}
&  \frac{\bm{\Delta}_{12}-\bm{\Delta}_{21}}{2}   \cr \frac{\bm{\Delta}_{21}-\bm{\Delta}_{12}}{2}& \frac{\bm{\Delta}_{11}+\bm{\Delta}_{22}}{2}\end{matrix}\right]   \right|\right|_2      \\
            &=\left|\left|  \frac{1}{2} \left[\begin{matrix}\bm{\Delta}_{11}&\bm{\Delta}_{12}\cr\bm{\Delta}_{21}&\bm{\Delta}_{22}\end{matrix}\right] +
            \frac{1}{2} \left[\begin{matrix}\bm{\Delta}_{22}&-\bm{\Delta}_{21}\cr-\bm{\Delta}_{12}&\bm{\Delta}_{11}\end{matrix}\right]       \right|\right|_2\\
            &\leq \frac{1}{2} \left|\left|   \left[\begin{matrix}\bm{\Delta}_{11}&\bm{\Delta}_{12}\cr\bm{\Delta}_{21}&\bm{\Delta}_{22}\end{matrix}\right] \right|\right|_2+
             \frac{1}{2}\left|\left| \left[\begin{matrix}\bm{\Delta}_{22}&-\bm{\Delta}_{21}\cr-\bm{\Delta}_{12}&\bm{\Delta}_{11}\end{matrix}\right]       \right|\right|_2\\
            &\leq \frac{1}{2}+\frac{1}{2}\\
            &=1,
\end{align*}
where we used the Jensen's inequality for the spectral norm, and
the fact that $1\geq \left|\left|\left[\begin{matrix}\bm{\Delta}_{22}&-\bm{\Delta}_{21}\cr-\bm{\Delta}_{12}&\bm{\Delta}_{11}\end{matrix}\right] \right|\right|_2 = \left|\left|\left[\begin{matrix}\bm{\Delta}_{11}&\bm{\Delta}_{12}\cr\bm{\Delta}_{21}&\bm{\Delta}_{22}\end{matrix}\right] \right|\right|_2$
(which comes from using the variational characterization of spectral norm).
This concludes the proof of (\ref{setreverse}).

Combining (\ref{setforward}) and (\ref{setreverse}), we arrive at (\ref{eq:GsubsetsubdiffF}).

\end{proof}

\bibliography{Reference_Hankel_Nonuniform}

\begin{thebibliography}{10}

\bibitem{bhatia_matrix_2013}
R.~Bhatia.
\newblock {\em Matrix analysis}, volume 169 of {\em Graduate Texts in
  Mathematics}.
\newblock Springer-Verlag New York, 1997.

\bibitem{cai_robust_2016}
J.~Cai, X.~Qu, W.~Xu, and G.~Ye.
\newblock Robust recovery of complex exponential signals from random {Gaussian}
  projections via low rank {Hankel} matrix reconstruction.
\newblock {\em Applied and Computational Harmonic Analysis}, 41(2):470--490,
  September 2016.

\bibitem{candes_towards_2014}
E.~Cand{\`e}s and C.~Fernandez-Granda.
\newblock Towards a mathematical theory of super-resolution.
\newblock {\em Comm. Pure Appl. Math}, 67(6):906--956, June 2014.

\bibitem{candes_exact_2009}
E.~J. Candes and B.~Recht.
\newblock Exact matrix completion via convex optimization.
\newblock {\em Found Comput Math}, 9:717--772, 2009.

\bibitem{candes_robust_2006}
E.~J. Candes, J.~Romberg, and T.~Tao.
\newblock Robust uncertainty principles: exact signal reconstruction from
  highly incomplete frequency information.
\newblock {\em IEEE Transactions on Information Theory}, 52(2):489--509,
  February 2006.

\bibitem{chen_robust_2014}
Y.~Chen and Y.~Chi.
\newblock Robust spectral compressed sensing via structured matrix completion.
\newblock {\em IEEE Transactions on Information Theory}, 60(10):6576--6601,
  October 2014.

\bibitem{2DChiChen}
Y.~Chi and Y.~Chen.
\newblock Compressive two-dimensional harmonic retrieval via atomic norm
  minimization.
\newblock {\em IEEE Transactions on Signal Processing}, 63(4):1030--1042, Feb
  2015.

\bibitem{chi_sensitivity_2011}
Y.~Chi, L.~L. Scharf, A.~Pezeshki, and A.~R. Calderbank.
\newblock Sensitivity to basis mismatch in compressed sensing.
\newblock {\em IEEE Transactions on Signal Processing}, 59(5):2182--2195, May
  2011.

\bibitem{dai_nuclear_2015}
L.~Dai and K.~Pelckmans.
\newblock On the nuclear norm heuristic for a {Hankel} matrix completion
  problem.
\newblock {\em Automatica}, 51(Supplement C):268--272, January 2015.

\bibitem{de_prony_essai_1795}
B.~G.~R. de~Prony.
\newblock Essai experimental et analytique: {Sur} les lois de la dilatabilite
  de fluides elastique et sur celles de la force expansive de la vapeur de
  l'alkool, a differentes temperatures.
\newblock {\em J. de l'Ecole Polytechnique}, 1795.

\bibitem{donoho_compressed_2006}
D.~L. Donoho.
\newblock Compressed sensing.
\newblock {\em IEEE Transactions on Information Theory}, 52(4):1289--1306,
  April 2006.

\bibitem{fazel_log-det_2003}
M.~Fazel, H.~Hindi, and S.~P. Boyd.
\newblock Log-det heuristic for matrix rank minimization with applications to
  {Hankel} and {Euclidean} distance matrices.
\newblock In {\em Proceedings of the 2003 {American} {Control} {Conference},
  2003.}, volume~3, pages 2156--2162 vol.3, June 2003.

\bibitem{fazel_hankel_2013}
M.~Fazel, T.~Pong, D.~Sun, and P.~Tseng.
\newblock Hankel matrix rank minimization with applications to system
  identification and realization.
\newblock {\em SIAM. J. Matrix Anal. and Appl.}, 34(3):946--977, January 2013.

\bibitem{horn_matrix_2012}
R.~Horn and C.~Johnson.
\newblock {\em Matrix analysis}.
\newblock Cambridge university press, 2012.

\bibitem{hua_matrix_1990}
Y.~Hua and T.~K. Sarkar.
\newblock Matrix pencil method for estimating parameters of exponentially
  damped/undamped sinusoids in noise.
\newblock {\em IEEE Transactions on Acoustics, Speech, and Signal Processing},
  38(5):814--824, May 1990.

\bibitem{li_new_2003}
W.~Li and W.~Sun.
\newblock New perturbation bounds for unitary polar factors.
\newblock {\em SIAM. J. Matrix Anal. and Appl.}, 25(2):362--372, January 2003.

\bibitem{liao_music_2016}
W.~Liao and A.~Fannjiang.
\newblock {MUSIC} for single-snapshot spectral estimation: {Stability} and
  super-resolution.
\newblock {\em Applied and Computational Harmonic Analysis}, 40(1):33--67,
  January 2016.

\bibitem{lustig_sparse_2007}
M.~Lustig, D.~Donoho, and J.~M. Pauly.
\newblock Sparse {MRI}: {The} application of compressed sensing for rapid {MR}
  imaging.
\newblock {\em Magn. Reson. Med.}, 58(6):1182--1195, December 2007.

\bibitem{markovsky_structured_2008}
I.~Markovsky.
\newblock Structured low-rank approximation and its applications.
\newblock {\em Automatica}, 44(4):891--909, April 2008.

\bibitem{Prior}
K.~V. Mishra, M.~Cho, A.~Kruger, and W.~Xu.
\newblock Spectral super-resolution with prior knowledge.
\newblock {\em IEEE Transactions on Signal Processing}, 63(20):5342--5357, Oct
  2015.

\bibitem{oymak_new_2010}
S.~Oymak and B.~Hassibi.
\newblock New null space results and recovery thresholds for matrix rank
  minimization.
\newblock November 2010.
\newblock arXiv: 1011.6326.

\bibitem{recht_null_2011}
B.~Recht, W.~Xu, and B.~Hassibi.
\newblock Null space conditions and thresholds for rank minimization.
\newblock {\em Math. Program.}, 127(1):175--202, March 2011.

\bibitem{rockafellar_convex_2015}
R.~Rockafellar.
\newblock {\em Convex analysis}.
\newblock Princeton university press, 2015.

\bibitem{roy_esprit-estimation_1989}
R.~Roy and T.~Kailath.
\newblock {ESPRIT}-estimation of signal parameters via rotational invariance
  techniques.
\newblock {\em IEEE Transactions on Acoustics, Speech, and Signal Processing},
  37(7):984--995, July 1989.

\bibitem{schiebinger_superresolution_2015}
G.~Schiebinger, E.~Robeva, and B.~Recht.
\newblock Superresolution without separation.
\newblock In {\em 2015 {IEEE} 6th {International} {Workshop} on {Computational}
  {Advances} in {Multi}-{Sensor} {Adaptive} {Processing} ({CAMSAP})}, pages
  45--48, December 2015.

\bibitem{resolutionlimit}
G.~Tang.
\newblock Resolution limits for atomic decompositions via markov-bernstein type
  inequalities.
\newblock In {\em 2015 International Conference on Sampling Theory and
  Applications (SampTA)}, pages 548--552, May 2015.

\bibitem{tang_compressed_2013}
G.~Tang, B.~N. Bhaskar, P.~Shah, and B.~Recht.
\newblock Compressed sensing off the grid.
\newblock {\em IEEE Transactions on Information Theory}, 59(11):7465--7490,
  November 2013.

\bibitem{tropp_introduction_2015}
J.~Tropp.
\newblock An {Introduction} to matrix concentration inequalities.
\newblock {\em MAL}, 8(1-2):1--230, May 2015.

\bibitem{tropp_beyond_2010}
J.~A. Tropp, J.~N. Laska, M.~F. Duarte, J.~K. Romberg, and R.~G. Baraniuk.
\newblock Beyond {Nyquist}: {Efficient} sampling of sparse bandlimited signals.
\newblock {\em IEEE Transactions on Information Theory}, 56(1):520--544,
  January 2010.

\bibitem{HankelComon}
K.~Usevich and P.~Comon.
\newblock Hankel low-rank matrix completion: Performance of the nuclear norm
  relaxation.
\newblock {\em IEEE Journal of Selected Topics in Signal Processing},
  10(4):637--646, June 2016.

\bibitem{xu_precise_2014}
W.~Xu, J.~F. Cai, K.~V. Mishra, M.~Cho, and A.~Kruger.
\newblock Precise semidefinite programming formulation of atomic norm
  minimization for recovering d-dimensional off-the-grid frequencies.
\newblock In {\em 2014 {Information} {Theory} and {Applications} {Workshop}
  ({ITA})}, pages 1--4, February 2014.

\end{thebibliography}
\end{document}